\documentclass[12pt,english]{article}
\usepackage[T1]{fontenc}
\usepackage[latin9]{inputenc}
\usepackage{geometry}
\usepackage{babel}
\usepackage{array}
\usepackage{rotfloat}
\usepackage{booktabs}
\usepackage{mathtools}
\usepackage{bm}
\usepackage{multirow}
\usepackage{amsmath}
\usepackage{amsthm}
\usepackage{amssymb}
\usepackage{graphicx}
\usepackage{setspace}
\usepackage[authoryear]{natbib}
\usepackage[unicode=true,pdfusetitle,
 bookmarks=true,bookmarksnumbered=false,bookmarksopen=false,
 breaklinks=false,pdfborder={0 0 1},backref=false,colorlinks=false]
 {hyperref}

\makeatletter

\providecommand{\tabularnewline}{\\}

\numberwithin{equation}{section}
\numberwithin{figure}{section}
\theoremstyle{plain}
\newtheorem{thm}{\protect\theoremname}
\theoremstyle{definition}
\newtheorem{defn}[thm]{\protect\definitionname}
\theoremstyle{plain}
\newtheorem{cor}[thm]{\protect\corollaryname}
\theoremstyle{plain}
\newtheorem{lem}[thm]{\protect\lemmaname}

\@ifundefined{date}{}{\date{}}
\makeatother

\providecommand{\corollaryname}{Corollary}
\providecommand{\definitionname}{Definition}
\providecommand{\lemmaname}{Lemma}
\providecommand{\theoremname}{Theorem}

\begin{document}
\title{Graph Neural Networks: Theory for Estimation with Application on Network
Heterogeneity}
\author{Yike Wang\thanks{Department of Economics, London School of Economics and Political
Science, Houghton Street, London, WC2A 2AE, U.K. Email: y.wang379@lse.ac.uk.} \and Chris Gu\thanks{Scheller College of Business, Georgia Institute of Technology, 800
W. Peachtree St. NW Atlanta, GA 30308, U.S. Email: chris.gu@scheller.gatech.edu.} \and Taisuke Otsu\thanks{Department of Economics, London School of Economics, Houghton Street,
London, WC2A 2AE, U.K. Email: t.otsu@lse.ac.uk.}}
\maketitle
\begin{abstract}
This paper presents a novel application of graph neural networks for
modeling and estimating network heterogeneity. Network heterogeneity
is characterized by variations in unit's decisions or outcomes that
depend not only on its own attributes but also on the conditions of
its surrounding neighborhood. We delineate the convergence rate of
the graph neural networks estimator, as well as its applicability
in semiparametric causal inference with heterogeneous treatment effects.
The finite-sample performance of our estimator is evaluated through
Monte Carlo simulations. In an empirical setting related to microfinance
program participation, we apply the new estimator to examine the average
treatment effects and outcomes of counterfactual policies, and to
propose an enhanced strategy for selecting the initial recipients
of program information in social networks.
\end{abstract}
\begin{description}
\item [{{\small{}Keywords:}}] {\small{}Artificial neural nets, social networks,
causal inference, graph representation learning.}{\small\par}
\item [{{\small{}JEL}}] \textbf{\small{}Classification}{\small{}: C13,
C14, C45, C51\newpage}{\small\par}
\end{description}

\section{Introduction}

The world is inherently interconnected, whether it is the natural,
engineering, or social domain. For instance, fundamental particles
bond together to form larger molecules, which together form building
blocks of organic or inorganic objects. Human-engineered systems,
such as computers or the Internet, are made up of large amounts of
small components, connected together in well-designed manners in order
to perform more complex functionalities. Social animals, including
humans, form social ties and their decisions are influenced by who
they are connected with and the attributes of their peers. It is not
only the individual entities that are important, but the networks
surrounding them hold equal significance. As we collect ever larger
amounts of multi-faceted data, the opportunity and necessity to understand
the world in a way that respects such interconnection becomes increasingly
vital for expanding our understanding of these various fields.

In economics and social sciences, studying individual heterogeneity
is an essential topic for empirical research and policy making. The
inherent differences among individuals, which encompass a wide range
of factors such as personal characteristics, socioeconomic backgrounds,
and life circumstances, can significantly impact how individuals respond
to economic policies and play a fundamental role for policy design.
The literature has a long history of recognizing and accounting for
individual heterogeneity, with approaches such as the demand estimation
method for differentiated products (\citealp{berry1995automobile}),
the individual fixed-effects panel-data method (\citealp{wooldridge2010econometric}),
and the potential outcomes framework for heterogeneous treatment effects
(\citealp{rubin1974estimating}), among others. A more recent stream
of literature has started utilizing machine learning techniques to
incorporate individual heterogeneity. Examples include \citet{belloni2014inference}
using the Lasso, \citet{bonhomme2015grouped} using k-means, and \citet{wager2018estimation}
using random forests. In this trend, the study most closely related
to our work is \citet{farrell2021deep} who have demonstrated a promising
potential of using artificial neural networks to capture individual
heterogeneity.

The Multi-layer Perceptron (MLP) architecture, analyzed in \citet{farrell2021deep},
is a potent tool to flexibly represent the dependence of individual
heterogeneity on a large number of observed characteristics. However,
in various empirical settings, we often encounter rich network data
in conjunction with these observed characteristics. Adapting the canonical
neural network model to effectively capture the variations embedded
in network information across observations, which we refer to as network
heterogeneity, could be a compelling topic to explore in economic
research.

With this goal in mind, it is inspiring to reflect on the success
of Convolutional Neural Networks (CNNs) in the deep learning literature
(\citealp{lecun1998gradient}) and connect that to our goal of incorporating
network heterogeneity. CNNs have achieved significant success in computer
vision, with notable applications such as image classification, object
detection, and semantic segmentation. For example, a remarkable application
of CNNs is the renowned artificial intelligence program, AlphaGo,
where board positions in the game of Go are represented as visual
patterns and CNNs effectively detect and classify these patterns.
A prominent feature in the domain of computer vision is the importance
of small patterns within an image, which often matter more than the
whole image in many applications. And these small patterns often appear
repeatedly within and across images. The convolutional layer in CNNs
consists of filters that parsimoniously extracts only the local information
around a focal pixel. And the parameter sharing among these filters,
meaning the same parameters are used across filters, allows for the
detection of repeated patterns, contributing to the success of the
CNN architecture. 

Analogously, in the context of network heterogeneity, it is intuitive
that peers closely connected via networks would provide more useful
information than distant nodes. Also, the same aggregation function
could prove beneficial for extracting local information across different
nodes. Hence, it might be worthwhile to borrow insights from CNNs
when studying network heterogeneity. However, a direct application
is not viable, as traditional CNNs are designed to work on data residing
in regular, grid-like structures, rather than on graphs.

Graph Neural Networks (GNNs) offer a general modeling framework for
building neural networks on graphs, a topic thoroughly overviewed
by \citet{hamilton2020graph}. Intuitively, the network heterogeneity
of a focal node depends on the conditions of their connected peers,
which in turn depend on the conditions of their peers, and so on.
GNNs provide a flexible framework for incorporating peer information,
taking into account network topology. The canonical MLP is less suited
to this task, as each focal node has a different number of peers with
varying distances, while the MLP can only incorporate a fixed number
of inputs. GNNs already have many impressive real-world applications.
For instance, \citet{stokes2020deep} modeled chemical molecule structures
as graphs, based on which they predict their pathogen inhibitory properties
using GNNs. Their study successfully identified a new molecule named
halicin against Acinetobacter baumanni, which is one of the highest
priority pathogens the World Health Organization urgently seeks new
antibiotics for. Despite the impactful applications of GNNs in many
fields, to the best of our knowledge, no study has yet employed GNNs
to quantify network heterogeneity for causal inference and policy
recommendations in economics. Our study aims to contribute to this
area.

The primary theoretical contribution of our study lies in providing
the convergence rate of the GNN estimator. The novelty of our theoretical
development resides in bounding the complexity measure and approximation
error, both of which are specific to GNNs, and in adapting the localization
analysis (\citet{bartlett2005local}) to incorporate dependent data
using a dependency graph. These elements are crucial in establishing
the convergence rate of the GNN estimator.

We then demonstrate the application of the GNN estimator in semiparametric
causal inference. The primitive inputs for conducting robust inference
of causal effects are the conditional expectations of potential outcomes
and the propensity score. The GNN estimator, by effectively integrating
local neighborhood information, offers flexible estimations for both
the conditional expectations of potential outcomes and the propensity
score, thereby facilitating causal inference. The convergence rate
result of the GNN estimator elucidates its applicability in semiparametric
causal inference. In broader terms, our theoretical analysis sheds
light on a general procedure that future researchers could adopt to
study other artificial neural network architectures tailored to their
empirical needs, and subsequently conduct causal inference.

In our empirical application, we utilize the data from \citet{BanerjeeChandrasekharDufloJackson2013}
to analyze individual decisions on microfinance participation given
their neighborhood surroundings. Our treatment effects estimates,
based on first-stage GNN estimates, indicate that well-connected households,
as measured by network centrality measures, are less inclined to borrow
via microfinance. Given the empirical evidence that participating
households have a much higher chance to disseminate information to
their neighbors than non-participating households, our findings highlight
a tradeoff for information diffusion: while better-connected households
have more avenues to spread information, their lower participation
rates curb this potential. Using the participation probabilities predicted
by the GNN estimates, we find considerable potential for improvement
in microfinance information targeting by achieving a balance between
the participation rate and network centrality of targeted households.
This refined approach can be leveraged to facilitate information diffusion
through social networks.

We organize the paper as follows. In Section \ref{sec: Graph-neural-networks},
we introduce the modeling setup and GNN estimator. We then present
the rate of convergence for the GNN estimator in Section \ref{sec: Theoretical-property},
followed by semiparametric causal inference for treatment effects
estimators in Section \ref{sec:Inference}. Monte Carlo simulations
and empirical applications are provided in Sections \ref{sec:Monte-Carlo}
and \ref{sec:Empirical-Application}, respectively. Finally, we conclude
in Section \ref{sec:Conclusion}. All the theoretical proofs are contained
in the appendix.

We use the following notations throughout the paper. A bold capital
letter (e.g., $\mathbf{A}$) represents a matrix, a bold lowercase
letter (e.g., $\mathbf{a}$) signifies a vector or vector-valued function,
and an unbold letter (e.g., $A$ or $a$) denotes a scalar, unless
stated otherwise. For a real number $\alpha,$ $\text{sgn}\left(\alpha\right)$
takes the value $1$ if $\alpha\geq0$ and $0$ otherwise. $\log$
refers to the natural logarithm. For vectors $\mathbf{a}$ and $\boldsymbol{b}$,
$\mathbf{a}\cdot\boldsymbol{b}$ symbolizes the inner product. The
letter $C$ is designated to denote a fixed finite positive constant,
which does not depend on the sample size $n$. For concise presentation,
the specific value of $C$ may vary across contexts, even from line
to line throughout the proof. For any non-random non-negative sequences
$a_{n}$ and $b_{n}$, $a_{n}=O\left(b_{n}\right)$ or $a_{n}\lesssim b_{n}$
means there exists a fixed finite constant $C$ such that $a_{n}\leq Cb_{n}$
for all $n$, and $a_{n}\asymp b_{n}$ means $a_{n}=O\left(b_{n}\right)$
and $b_{n}=O\left(a_{n}\right)$. $\mathbb{Z}_{+}$ denotes the set
of positive integers. For any two real numbers $a$ and $b$, $a\vee b=\max\left\{ a,b\right\} $.
For any positive integer $n$, $\left[n\right]=\left\{ 1,...,n\right\} $.
And for any function $f$, $\left\Vert f\right\Vert _{\infty,\mathcal{A}}=\max_{x\in\mathcal{A}}\left|f\left(x\right)\right|$. 

\section{Graph neural networks \label{sec: Graph-neural-networks}}

In this section, we first introduce the definition of network heterogeneity,
and then present the GNN estimator for estimating the network heterogeneity.

\subsection{Network heterogeneity \label{subsec:Network-heterogeneity}}

For each node $i\in\left[n\right]$, let $y_{i}\in\mathbb{R}$ denote
an outcome variable and $\boldsymbol{x}_{i}\in\mathbb{R}^{d}$ a finite-dimensional
vector of covariates. The adjacency matrix, denoted as $\boldsymbol{D}=\left(d_{ij}\right)$,
is an $n\times n$ matrix indicating the connections between the nodes.
In this paper, we focus on binary adjacency matrices, which are applicable
to both undirected and directed graphs. For an undirected graph, the
adjacency matrix is symmetric, where $d_{ij}=d_{ji}=1$ if there is
an edge between nodes $i$ and $j$, and $d_{ij}=d_{ji}=0$ otherwise.
In contrast, for a directed graph, the adjacency matrix $\boldsymbol{D}$
is normally asymmetric, where $d_{ij}=1$ if there is a directed edge
from node $i$ to node $j$, and $d_{ij}=0$ if there is no such edge.
We can also extend the framework to non-binary adjacency matrices
to accommodate more complex information such as edge weights and labels,
albeit at the expense of additional notation. Using the binary adjacency
matrix, denote the set of $i$'s adjacent neighbors as $\mathcal{N}\left(i\right)=\left\{ j\in\left[n\right]:d_{ij}=1\right\} $.
Also, let $\left|\mathcal{N}\left(i\right)\right|$ be the number
of elements in set $\mathcal{N}\left(i\right)$, and $\left|\mathcal{N}\left(i\right)\right|=0$
if $\mathcal{N}\left(i\right)$ is empty. The researcher observes
$\left\{ y_{i}\right\} _{i\in\left[n\right]}$, $\left\{ \boldsymbol{x}_{i}\right\} _{i\in\left[n\right]}$,
and $\boldsymbol{D}$ for a realized sample of size $n$.

In the following, we first introduce the intuition of one-layer GNNs,
and then progress to multi-layer GNNs. The type of network heterogeneity
that one-layer GNNs aim to encapsulate can be expressed as

{\small{}
\[
z_{*i}\left(f_{*}\right)=f_{*}\left(\boldsymbol{x}_{i},\frac{1}{\left|\mathcal{N}\left(i\right)\right|}\sum_{j\in\mathcal{N}\left(i\right)}\boldsymbol{x}_{j}\right),
\]
}wherein the network heterogeneity for node $i$ depends on its own
attributes, $\boldsymbol{x}_{i}$, and the average of its neighbors'
attributes. The function $f_{*}$ represents an unrestricted mapping.
Hence, the popular linear-in-means model, in which the outcome variable
is a linear function of $\boldsymbol{x}_{i}$ and the average covariates
of neighbors, can be viewed as a specific example of this representation.
While our primary focus in this study is on the setup using neighbors'
average, alternative aggregation methods to combine neighbors' information,
such as taking the maximum, minimum, sum, or weighted sum given the
availability of weights, can be easily adopted to replace the function
of averaging. In this paper, we assume one of these simple aggregation
functions is valid, which is a plausible assumption in many empirical
contexts, and the exploration of more complex aggregation methods
is deferred to future research.

While the one-layer model is straightforward, a significant limitation
lies in the possible inadequacy of using the adjacent neighbors' attributes
alone to summarize a node's distinctiveness. Therefore, it could be
more fitting to consider a two-layer extension in which node $i$'s
network heterogeneity relies on some latent embeddings of itself and
its neighbors:

{\small{}
\[
z_{*i}\left(f_{*}\right)=f_{*}^{(2)}\left(\boldsymbol{h}_{*i},\frac{1}{\left|\mathcal{N}\left(i\right)\right|}\sum_{j\in\mathcal{N}\left(i\right)}\boldsymbol{h}_{*j}\right),
\]
}where $\boldsymbol{h}_{*i}\in\mathbb{R}^{d_{h*}}$ represent some
unobserved measures of node $i$'s uniqueness with an unknown dimension
$d_{h*}\in\mathbb{Z}_{+}$. This concept is depicted in Figure \ref{fig: GNN 2-layers architecture }.
The hidden embeddings of node $i$ consequently depend on the observed
covariates of both itself and its neighbors:

{\small{}
\[
\boldsymbol{h}_{*i}=\boldsymbol{f}_{*}^{(1)}\left(\boldsymbol{x}_{i},\frac{1}{\left|\mathcal{N}\left(i\right)\right|}\sum_{j\in\mathcal{N}\left(i\right)}\boldsymbol{x}_{j}\right).
\]
}{\small\par}

\begin{figure}[h]
\begin{centering}
\includegraphics[width=0.95\columnwidth]{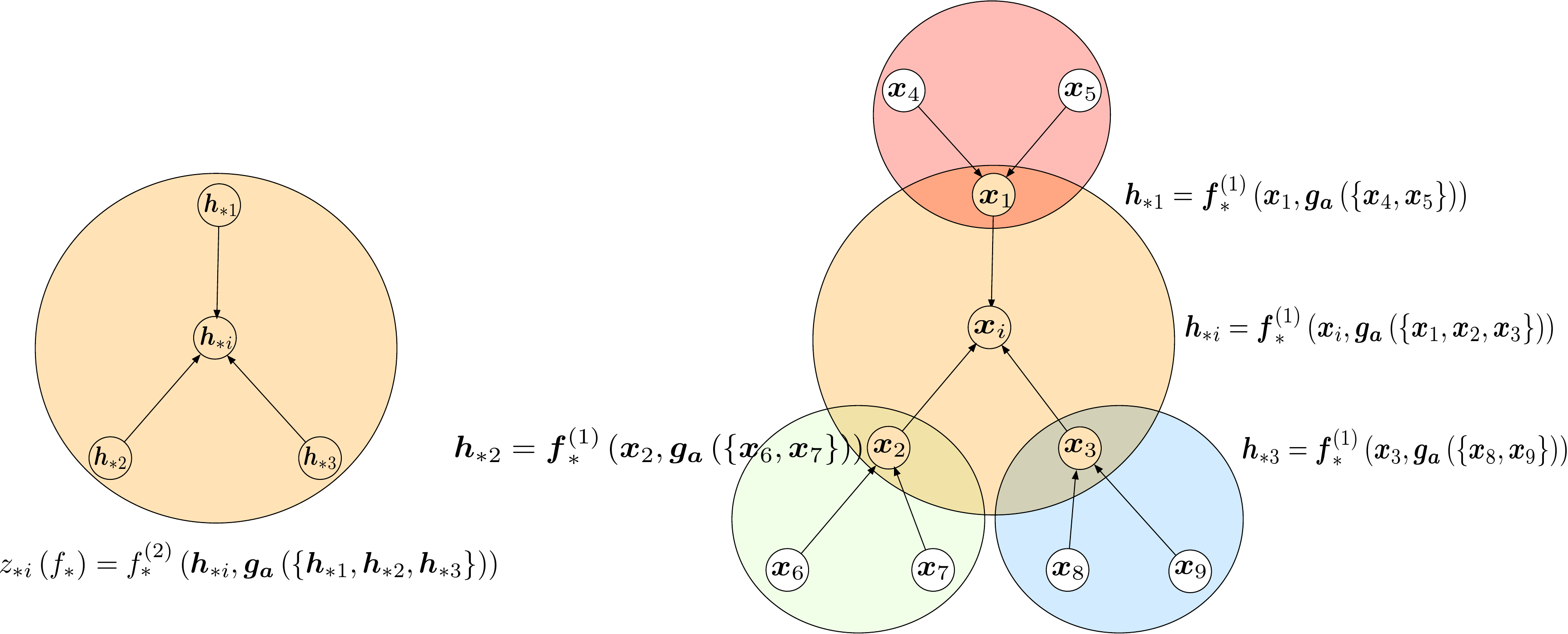}
\par\end{centering}
\caption{Network heterogeneity with two-layer GNNs.\label{fig: GNN 2-layers architecture }}

\medskip{}

\raggedright{}%
\noindent\begin{minipage}[t]{1\columnwidth}%
\begin{singlespace}
\begin{flushleft}
{\footnotesize{}Note: This figure provides a visual illustration of
the network heterogeneity of node $i$ using a two-layer GNN architecture.
In this example, node $i$ has three neighbors, namely nodes 1, 2
and 3. So the network heterogeneity of node $i$ is a flexible representation
that combines $i$'s own latent embeddings $\boldsymbol{h}_{*i}\in\mathbb{R}^{d_{h*}}$
and an aggregation of $i$'s neighbors' embeddings $\mathbf{g}_{a}\left(\left\{ \boldsymbol{h}_{*1},\boldsymbol{h}_{*2},\boldsymbol{h}_{*3}\right\} \right)\in\mathbb{R}^{d_{h*}}$
with a known aggregation function $\mathbf{g}_{a}\left(\cdot\right)$
(such as mean, max, min, etc). The latent embeddings of each node
in turn depend on their respective observed covariates and those of
their neighbors. For instance, node $1$ has two neighbors, nodes
4 and 5, so the latent embeddings of node $1$ depend flexibly on
$1$'s attributes, $\boldsymbol{x}_{1}$, and an aggregation of the
attributes of $1$'s neighbors, $\mathbf{g}_{a}\left(\left\{ \boldsymbol{x}_{4},\boldsymbol{x}_{5}\right\} \right)$.
With this two-layer GNN architecture, the network heterogeneity of
node $i$ depends on its own characteristics, those of its neighbors,
and the characteristics of its neighbors' neighbors, respecting the
topology of local networks. This concept can be expanded to describe
multi-layer GNN architectures.}
\par\end{flushleft}
\end{singlespace}
\end{minipage}
\end{figure}

In this two-layer model, the network heterogeneity of node $i$, denoted
as $z_{*i}\left(f_{*}\right)$, depends on its own attributes, its
neighbors' attributes, and the attributes of its neighbors' neighbors,
utilizing the local network structure. This is a natural extension
of the one-layer model setup. In many empirical contexts, an individual's
decision (for example, taking a loan) depends not only on the economic
status of her immediate connections, as their willingness to lend
may also depend on who they are connected with and could borrow from.
This concept can intuitively be extended to setups with multiple layers.
For instance, in a three-layer model, the third-degree connections
(i.e., neighbors of neighbors' neighbors) can also be relevant in
an individual's decision-making process.

To elucidate the idea more formally, we define an $L$-layer GNNs
model iteratively as $\forall l\in\left[L\right]$,{\small{}
\[
\boldsymbol{h}_{*i}^{\left(l\right)}\left(\boldsymbol{f}\right)=\boldsymbol{f}^{\left(l\right)}\left(\boldsymbol{h}_{*i}^{\left(l-1\right)}\left(\boldsymbol{f}\right),\overline{\boldsymbol{h}}_{*\mathcal{N}\left(i\right)}^{\left(l-1\right)}\left(\boldsymbol{f}\right)\right)\in\mathbb{R}^{d_{h*}^{\left(l\right)}},
\]
}where $\boldsymbol{h}_{*i}^{\left(0\right)}\left(\boldsymbol{f}\right)=\boldsymbol{x}_{i}$,
$\overline{\boldsymbol{h}}_{*\mathcal{N}\left(i\right)}^{\left(l-1\right)}\left(\boldsymbol{f}\right)=\frac{1}{\left|\mathcal{N}\left(i\right)\right|}\sum_{j\in\mathcal{N}\left(i\right)}\boldsymbol{h}_{*j}^{\left(l-1\right)}\left(\boldsymbol{f}\right)$,
and $\overline{\boldsymbol{h}}_{*\mathcal{N}\left(i\right)}^{\left(l-1\right)}\left(\boldsymbol{f}\right)=\boldsymbol{0}$
if $\mathcal{N}\left(i\right)$ is empty, with $\boldsymbol{f}$ being
the set of functions constructing the final layer output, i.e., $\boldsymbol{f}=\left\{ \boldsymbol{f}^{\left(l\right)}\right\} {}_{l\in\left[L\right]}$
and $\boldsymbol{f}^{\left(l\right)}=\left\{ f_{j}^{\left(l\right)}\right\} {}_{j\in\left[d_{h*}^{\left(l\right)}\right]}$.
The latent embedding dimension in each intermediate layer $l\in\left[L-1\right]$
is an unknown finite fixed integer $d_{h*}^{\left(l\right)}\in\mathbb{Z}_{+}$,
and the final layer output is a scalar ($d_{h*}^{\left(L\right)}=1$).
We denote the final layer output of the $L$-layer GNNs model as $z_{*i}\left(\boldsymbol{f}\right)\coloneqq\boldsymbol{h}_{*i}^{\left(L\right)}\left(\boldsymbol{f}\right)\in\mathbb{R}$.
Hence, the formulation of $z_{*i}\left(\boldsymbol{f}\right)$ depends
on the number of layers, $L$, and the dimensions of latent embeddings
in each intermediate layer, $\left\{ d_{h*}^{\left(l\right)}\right\} _{l\in\left[L-1\right]}$.
To simplify notation, we adopt the convention in this paper to omit
the reference to $L$ and $\left\{ d_{h*}^{\left(l\right)}\right\} _{l\in\left[L-1\right]}$
in $z_{*i}\left(\boldsymbol{f}\right)$. Also, define $\boldsymbol{f}_{*}$
as a set of functions yielding the smallest average loss function
value,

{\small{}
\begin{equation}
\boldsymbol{f}_{*}\in\arg\min_{\boldsymbol{f}}\mathbb{E}\left[\ell\left(y_{i},z_{*i}\left(\boldsymbol{f}\right)\right)\right].\label{eq: define_f_star}
\end{equation}
}Common examples of the loss function $\ell$ include the least squares
loss $\ell\left(y,z\right)=\frac{1}{2}\left(y-z\right)^{2}$ for real-valued
outcome variables, and the negative log-likelihood of logistic regression
$\ell\left(y,z\right)=-yz+\log\left(1+\exp\left(z\right)\right)$
for binary outcome variables.

We call the final layer output associated with an optimal $\boldsymbol{f}_{*}$,
namely $z_{*i}\left(\boldsymbol{f}_{*}\right)$, the network heterogeneity
of node $i$. On average, the network heterogeneity yields the smallest
loss function value for predicting the outcome variable $y_{i}$ given
the local neighborhood information. Standard derivations show that
$z_{*i}\left(\boldsymbol{f}_{*}\right)=\mathbb{E}\left[y_{i}\mid\boldsymbol{h}_{*i}^{\left(L-1\right)}\left(\boldsymbol{f}_{*}\right),\overline{\boldsymbol{h}}_{*\mathcal{N}\left(i\right)}^{\left(L-1\right)}\left(\boldsymbol{f}_{*}\right)\right]$
for the least squares loss, and $\frac{\exp\left(z_{*i}\left(\boldsymbol{f}_{*}\right)\right)}{1+\exp\left(z_{*i}\left(\boldsymbol{f}_{*}\right)\right)}=\mathbb{E}\left[y_{i}\mid\boldsymbol{h}_{*i}^{\left(L-1\right)}\left(\boldsymbol{f}_{*}\right),\overline{\boldsymbol{h}}_{*\mathcal{N}\left(i\right)}^{\left(L-1\right)}\left(\boldsymbol{f}_{*}\right)\right]$
for the logistic loss. In other words, the network heterogeneity
determines the conditional expectation of the outcome variable given
the local neighborhood information as summarized in the penultimate
layer latent embeddings. Without further restrictions, $z_{*i}\left(\boldsymbol{f}_{*}\right)$
may not be unique, while our main theoretical result in Theorem \ref{thm: rate_of_convergence}
is valid for any $z_{*i}\left(\boldsymbol{f}_{*}\right)$ that corresponds
to such an optimal forecasting rule.

The concept of network heterogeneity, represented as $z_{*i}\left(\boldsymbol{f}_{*}\right)$,
holds potential for a wide range of applications in empirical research.
This paper focuses on a significant application area: the semiparametric
causal inference of treatment effects. Central to identifying and
robustly inferring causal effects are the conditional expectations
of potential outcomes and the propensity score, both could depend
on local network surroundings. As detailed in Section \ref{sec:Inference},
under appropriate assumptions, the conditional expectation of the
potential outcome $y_{i}\left(t\right)$ and the propensity score
are known functions of respective network heterogeneity variables,
$z_{*i}\left(\boldsymbol{f}_{*}^{t}\right)$ and $z_{*i}\left(\boldsymbol{f}_{*}^{p}\right)$.
These network heterogeneities enable flexible characterization of
dependencies on local network surroundings. Utilizing GNN estimators,
we effectively estimate these network heterogeneities, leading to
accurate estimates for both the conditional expectations of potential
outcomes and the propensity score. These estimates are crucial for
conducting robust inference on various causal effects. In Section
\ref{sec:Inference}, we focus on a specific causal effect parameter:
the average effect of a counterfactual policy $\pi\left(s\right)=\mathbb{E}\left[s\left(\boldsymbol{\xi}_{i}\right)y_{i}\left(1\right)+\left(1-s\left(\boldsymbol{\xi}_{i}\right)\right)y_{i}\left(0\right)\right]$,
which plays an integral role in our empirical application.

Beyond the estimation of treatment effects, another noteworthy application
area of the concept of network heterogeneity, $z_{*i}\left(\boldsymbol{f}_{*}\right)$,
lies in improving structural estimations. Building upon and extending
the research by \citet{farrell2021individual}, the observation-specific
coefficients in structural models can be formulated as functions of
the observation-specific local network data via the concept of network
heterogeneity. This flexibility allows, for instance, the individual-specific
consumer demand elasticities to vary over consumers' social network
scenarios in demand estimations. As a result, the network heterogeneity
enriched structural models can be powerful to study key economic parameters,
such as elasticity and surplus, and to answer policy questions, such
as optimal pricing and targeting. The empirical exploration of these
diverse applications presents an opportunity for future research.

We now introduce some additional notations to describe network topology.
First, it would be convenient to transform the set $\mathcal{N}\left(i\right)$
into a tuple (i.e., an ordered list which may contain multiple occurrences
of the same element). Without loss of generality, we could fix the
labels of nodes in a sample using natural numbers. Then, let $\mathcal{D}\left(i\right)$
denote the tuple of nodes in the set $\mathcal{N}\left(i\right)$
arranged according to the order of natural numbers. For instance,
set $\mathcal{D}\left(i\right)=\left(1,2,3\right)$ if $\mathcal{N}\left(i\right)=\left\{ 1,2,3\right\} $
and $\mathcal{D}\left(i\right)=\emptyset$ if $\mathcal{N}\left(i\right)=\emptyset$.
Also, denote the concatenation of two tuples as $\left(a_{1},...,a_{m}\right)\oplus\left(b_{1},...,b_{n}\right)=\left(a_{1},...,a_{m},b_{1},...,b_{n}\right)$.\footnote{If there are empty tuples, set $\left(a_{1},...,a_{m}\right)\oplus\emptyset=\left(a_{1},...,a_{m}\right)$,
$\emptyset\oplus\left(b_{1},...,b_{n}\right)=\left(b_{1},...,b_{n}\right)$,
$\emptyset\oplus\emptyset=\emptyset$, and $\oplus_{j\in\emptyset}\mathcal{D}\left(j\right)=\emptyset$.} Then, let $\mathcal{D}_{l}\left(i\right)$ denote the tuple of nodes
whose distance from node $i$ is $l$ for $l\in\left\{ 0,1,...,L\right\} $.
In particular, define $\mathcal{D}_{l}\left(i\right)$ recursively
such that $\mathcal{D}_{0}\left(i\right)=\left(i\right)$ and $\mathcal{D}_{l}\left(i\right)=\oplus_{j\in\mathcal{\mathcal{D}}_{l-1}\left(i\right)}\mathcal{D}\left(j\right)$.
Using the example in Figure \ref{fig: GNN 2-layers architecture }
as an illustration, we have that $\mathcal{D}_{0}\left(i\right)=\left(i\right)$,
which includes node $i$ itself; $\mathcal{D}_{1}\left(i\right)=\mathcal{D}\left(i\right)=\left(1,2,3\right)$,
containing the immediate neighbors of node $i$; and $\mathcal{D}_{2}\left(i\right)=\oplus_{j\in\left(1,2,3\right)}\mathcal{D}\left(j\right)=\left(4,5,6,7,8,9\right)$,
including the neighbors of node $i$'s neighbors.

Moreover, let $\boldsymbol{\xi}_{i,l}$ denote the local network information
for node $i$ up to distance $l$ for $l\in\left\{ 0,1,...,L\right\} $.
Specifically, define $\boldsymbol{\xi}_{i,l}$ recursively such that
$\boldsymbol{\xi}_{i,0}=\left(\boldsymbol{x}_{i}\right)$ and $\boldsymbol{\xi}_{i,l}=\left(\boldsymbol{\xi}_{i,l-1},\left(\left(\boldsymbol{x}_{k}\right)_{k\in\mathcal{D}\left(j\right)}\right)_{j\in\mathcal{D}_{l-1}\left(i\right)}\right)$.\footnote{When encountering the empty tuple $\emptyset$, set $\left(\boldsymbol{x}_{k}\right)_{k\in\emptyset}=\emptyset$
and $\left(\left(\boldsymbol{x}_{k}\right)_{k\in\mathcal{D}\left(j\right)}\right)_{j\in\emptyset}=\left(\emptyset\right)$.} In the example in Figure \ref{fig: GNN 2-layers architecture },
it is easy to find that $\boldsymbol{\xi}_{i,0}=\left(\boldsymbol{x}_{i}\right)$,
$\boldsymbol{\xi}_{i,1}=\left(\boldsymbol{\xi}_{i,0},\left(\left(\boldsymbol{x}_{1},\boldsymbol{x}_{2},\boldsymbol{x}_{3}\right)\right)\right),$
and $\boldsymbol{\xi}_{i,2}=\left(\boldsymbol{\xi}_{i,1},\left(\left(\boldsymbol{x}_{4},\boldsymbol{x}_{5}\right),\left(\boldsymbol{x}_{6},\boldsymbol{x}_{7}\right),\left(\boldsymbol{x}_{8},\boldsymbol{x}_{9}\right)\right)\right).$
This definition of $\boldsymbol{\xi}_{i,l}$ can accommodate scenarios
where nodes share common friends and networks have cycles (see Figure
\ref{fig:Construction-of-xi_i} for illustrations). For simplicity
in notation, we denote the $L$-hop network information from the
perspective of node $i$ as $\boldsymbol{\xi}_{i}\coloneqq\boldsymbol{\xi}_{i,L}$
(omitting reference to $L$). With this definition of $\boldsymbol{\xi}_{i}$,
it is evident that $\boldsymbol{\xi}_{i}$ encapsulates all the needed
input data to formulate the network heterogeneity variable $z_{*i}\left(\boldsymbol{f}_{*}\right)$
for node $i$.

\begin{figure}
\begin{centering}
\includegraphics[scale=0.5]{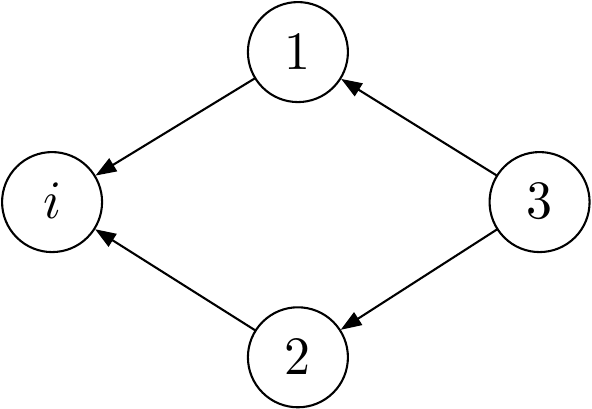} ~~~~~~~~~\includegraphics[scale=0.5]{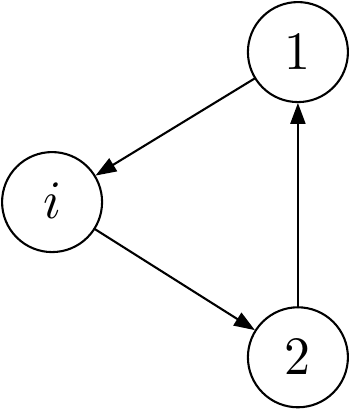}
\par\end{centering}
\caption{Construction of the local network data using the definition of $\boldsymbol{\xi}_{i,l}$\label{fig:Construction-of-xi_i}}
\medskip{}

\raggedright{}%
\noindent\begin{minipage}[t]{1\columnwidth}%
\begin{singlespace}
\begin{flushleft}
{\footnotesize{}Note: Figure \ref{fig:Construction-of-xi_i} illustrates
the construction of $\boldsymbol{\xi}_{i,l}$ when nodes share common
friends and networks have cycles. The figure on the left depicts a
local network around node $i$ which involves common friends (node
$3$ is the common friend of both nodes $1$ and $2$). In this case,
the definition of $\mathcal{D}_{l}\left(i\right)$ implies that $\mathcal{D}_{0}\left(i\right)=\left(i\right)$,
$\mathcal{D}_{1}\left(i\right)=\left(1,2\right)$, $\mathcal{D}_{2}\left(i\right)=\left(3,3\right)$,
and $\mathcal{D}_{3}\left(i\right)=\emptyset$. Also, the definition
$\boldsymbol{\xi}_{i,l}$ provides that $\boldsymbol{\xi}_{i,0}=\left(\boldsymbol{x}_{i}\right)$,
$\boldsymbol{\xi}_{i,1}=\left(\boldsymbol{\xi}_{i,0},\left(\left(\boldsymbol{x}_{1},\boldsymbol{x}_{2}\right)\right)\right)$,
$\boldsymbol{\xi}_{i,2}=\left(\boldsymbol{\xi}_{i,1},\left(\left(\boldsymbol{x}_{3}\right),\left(\boldsymbol{x}_{3}\right)\right)\right)$,
and $\boldsymbol{\xi}_{i,3}=\left(\boldsymbol{\xi}_{i,2},\left(\emptyset,\emptyset\right)\right)$.
In addition, the figure on the right depicts a local network with
a cycle. In this case, $\mathcal{D}_{0}\left(i\right)=\left(i\right)$,
$\mathcal{D}_{1}\left(i\right)=\left(1\right)$, $\mathcal{D}_{2}\left(i\right)=\left(2\right)$,
and $\mathcal{D}_{3}\left(i\right)=\left(i\right)$. Also, $\boldsymbol{\xi}_{i,0}=\left(\boldsymbol{x}_{i}\right)$,
$\boldsymbol{\xi}_{i,1}=\left(\boldsymbol{\xi}_{i,0},\left(\left(\boldsymbol{x}_{1}\right)\right)\right)$,
$\boldsymbol{\xi}_{i,2}=\left(\boldsymbol{\xi}_{i,1},\left(\left(\boldsymbol{x}_{2}\right)\right)\right)$,
and $\boldsymbol{\xi}_{i,3}=\left(\boldsymbol{\xi}_{i,2},\left(\left(\boldsymbol{x}_{i}\right)\right)\right)$.
The construction of $\boldsymbol{\xi}_{i,l}$ for $l>3$ can be obtained
analogously. These two examples demonstrate that the definition of
$\boldsymbol{\xi}_{i,l}$ can incorporate the settings when nodes
have mutual friends and networks have cycles.}
\par\end{flushleft}
\end{singlespace}
\end{minipage}
\end{figure}

\subsection{GNN estimator}

We now introduce the GNN estimator, which is utilized to estimate
the network heterogeneity variable $z_{*i}\left(\boldsymbol{f}_{*}\right)$.
In particular, for each layer $l\in\left[L\right]$, the latent feature
embeddings are updated recursively using shallow neural networks as
{\small{}
\[
\boldsymbol{h}_{i}^{\left(l\right)}=\mathbf{\bm{\sigma}}\left(\boldsymbol{A}^{\left(l\right)}\boldsymbol{h}_{i}^{\left(l-1\right)}+\boldsymbol{A}_{\mathcal{N}}^{\left(l\right)}\overline{\boldsymbol{h}}_{\mathcal{N}\left(i\right)}^{\left(l-1\right)}+\boldsymbol{b}^{\left(l\right)}\right)\in\mathbb{R}^{d_{h}^{\left(l\right)}},
\]
}where $\bm{\sigma}$ is an activation function applied element-wise,
and $\left(\boldsymbol{A}^{\left(l\right)},\boldsymbol{A}_{\mathcal{N}}^{\left(l\right)},\boldsymbol{b}^{\left(l\right)}\right)$
are parameters to be estimated. The initial feature embeddings are
defined by $\boldsymbol{h}_{i}^{(0)}=\boldsymbol{x}_{i}$, and the
mean feature embeddings of neighboring nodes are $\overline{\boldsymbol{h}}_{\mathcal{N}\left(i\right)}^{\left(l-1\right)}=\frac{1}{\left|\mathcal{N}\left(i\right)\right|}\sum_{j\in\mathcal{N}\left(i\right)}\boldsymbol{h}_{j}^{\left(l-1\right)}$
(if $\mathcal{N}\left(i\right)$ is empty, set $\overline{\boldsymbol{h}}_{\mathcal{N}\left(i\right)}^{\left(l-1\right)}=\boldsymbol{0}$).
Then, the feature of node $i$ is obtained through the transformation

{\small{}
\begin{equation}
z_{i}\left(\boldsymbol{\theta}\right)=\boldsymbol{a}\cdot\boldsymbol{h}_{i}^{\left(L\right)}+b\in\mathbb{R},\label{eq: z_i_theta}
\end{equation}
}where $\left(\boldsymbol{a},b\right)$ are additional parameters,
and $\boldsymbol{\theta}$ is the vectorization of all parameters
appeared in the construction of $z_{i}\left(\boldsymbol{\theta}\right)$.
Note that, given $\boldsymbol{\theta}$, $z_{i}\left(\boldsymbol{\theta}\right)$
can be fully constructed using the input data $\boldsymbol{\xi}_{i}$.

Denote the parameter space of $\boldsymbol{\theta}$ by $\Theta_{d_{h}}$,
which is indexed by the dimensions of the embeddings for each layer,
i.e., $d_{h}=\left\{ d_{h}^{\left(l\right)}\right\} _{l\in\left[L\right]}$.
In particular,{\small{}
\[
\Theta_{d_{h}}=\Theta_{\left\{ d_{h}^{\left(1\right)},...,d_{h}^{\left(L\right)}\right\} }=\left\{ \boldsymbol{a}\in\mathbb{R}^{d_{h}^{\left(L\right)}},b\in\mathbb{R},\left\{ \boldsymbol{A}^{\left(l\right)},\boldsymbol{A}_{\mathcal{N}}^{\left(l\right)}\in\mathbb{R}^{d_{h}^{\left(l\right)}\times d_{h}^{\left(l-1\right)}},\boldsymbol{b}^{\left(l\right)}\in\mathbb{R}^{d_{h}^{\left(l\right)}}\right\} _{l\in\left[L\right]}\right\} ,
\]
}with $d_{h}^{\left(0\right)}=d$ being the number of covariates.

Moreover, denote $\left\Vert z_{i}\left(\boldsymbol{\theta}\right)\right\Vert _{\infty}=\underset{\boldsymbol{\xi}_{i}}{\text{sup}}\left|z_{i}\left(\boldsymbol{\theta}\right)\right|$,
which, given $\boldsymbol{\theta}$, is the largest value of $\left|z_{i}\left(\boldsymbol{\theta}\right)\right|$
over all possible $L$-hop local networks $\boldsymbol{\xi}_{i}$.
The value of $\left\Vert z_{i}\left(\boldsymbol{\theta}\right)\right\Vert _{\infty}$
depends on $\boldsymbol{\theta}$ and is common across all nodes.
We further refine the set $\Theta_{d_{h}}$ so that $\left\Vert z_{i}\left(\boldsymbol{\theta}\right)\right\Vert _{\infty}$
is bounded from above, which results in a new set defined as{\small{}
\[
\Theta_{d_{h},\bar{z}}=\left\{ \boldsymbol{\theta}\in\Theta_{d_{h}}:\left\Vert z_{i}\left(\boldsymbol{\theta}\right)\right\Vert _{\infty}\leq\bar{z}\right\} ,
\]
}where $\bar{z}$ is a fixed constant used throughout the analysis.

To estimate the parameters $\boldsymbol{\theta}$, we consider the
optimization problem{\small{}
\begin{equation}
\hat{\boldsymbol{\theta}}\in\underset{\begin{array}{c}
\boldsymbol{\theta}\in\Theta_{d_{h},\bar{z}}\end{array}}{\arg\min}\frac{1}{n}\sum_{i=1}^{n}\ell\left(y_{i},z_{i}\left(\boldsymbol{\theta}\right)\right),\label{eq:loss}
\end{equation}
}with $\ell$ being the same loss function used in the formulation
of $z_{*i}\left(\boldsymbol{f}_{*}\right)$ in (\ref{eq: define_f_star}).
We refer to $z_{i}\left(\hat{\boldsymbol{\theta}}\right)$ as the
GNN estimator of the network heterogeneity $z_{*i}\left(\boldsymbol{f}_{*}\right)$.
The minimizer $\hat{\boldsymbol{\theta}}$ and hence $z_{i}\left(\hat{\boldsymbol{\theta}}\right)$
may not be unique, while our main theoretical result in Theorem \ref{thm: rate_of_convergence}
is valid for any minimizer $\hat{\boldsymbol{\theta}}$ that solves
the problem in (\ref{eq:loss}). Without loss of generality, we let
$\bar{z}>1.1M$ with $M$ being the upper bound of $\left|z_{*i}\left(\boldsymbol{f}_{*}\right)\right|$
as imposed in Assumption I 1 below (e.g., we could set $\bar{z}=2M$
as in \citealp{farrell2021deep}). In the following, we present the
theoretical properties of the GNN estimator $z_{i}\left(\hat{\boldsymbol{\theta}}\right)$
to estimate the target object $z_{*i}\left(\boldsymbol{f}_{*}\right)$.

\section{Theoretical property \label{sec: Theoretical-property}}

In this section, we present the convergence rate of the GNN estimator.
We first define the Sobolev space {\small{}
\[
\mathcal{W}_{\eta}^{\beta,\infty}\left(\Omega\right)=\left\{ f:\sum_{0\leq|\boldsymbol{k}|\leq\beta}\underset{\begin{array}{c}
\boldsymbol{x}\in\Omega\end{array}}{\ensuremath{\text{ess}\ \text{sup}}}\left|\frac{\partial^{|\mathbf{k}|}f}{\partial x_{1}^{k_{1}}...\partial x_{p}^{k_{p}}}\right|\leq\eta\right\} ,
\]
}where $\Omega\subseteq\mathbb{R}^{p}$, $\boldsymbol{k}\coloneqq\left(k_{1},...,k_{p}\right)$
with $k_{1},...,k_{p}$ being nonnegative integers, $\left|\boldsymbol{k}\right|=k_{1}+...+k_{p}$,
and $\boldsymbol{x}=\left(x_{1},...,x_{p}\right)$ being an argument
of $f$. Also, we define{\small{}
\[
\boldsymbol{\theta}_{*}\in\underset{\begin{array}{c}
\boldsymbol{\theta}\in\Theta_{d_{h},\bar{z}}\end{array}}{\ensuremath{\text{arg}\min}}\mathbb{E}\left[\ell\left(y_{i},z_{i}\left(\boldsymbol{\theta}\right)\right)\right].
\]
}{\small\par}

We now introduce the assumptions to demonstrate the convergence rate.

\subsection{Assumption I}
\begin{enumerate}
\item $f_{*}^{\left(L\right)}:\left[-1,1\right]^{2d_{h*}^{\left(L-1\right)}}\rightarrow\left[-M,M\right]$
belongs to the Sobolev space $\mathcal{W}_{\eta}^{\beta,\infty}\left(\left[-1,1\right]^{2d_{h*}^{\left(L-1\right)}}\right)$
and if $L\geq2$, $f_{*j_{l}}^{\left(l\right)}:\left[-1,1\right]^{2d_{h*}^{\left(l-1\right)}}\rightarrow\left[-1,1\right]$
belongs to the Sobolev space $\mathcal{W}_{\eta}^{\beta,\infty}\left(\left[-1,1\right]^{2d_{h*}^{\left(l-1\right)}}\right)$
for each $l\in\left[L-1\right]$ and $j_{l}\in\left[d_{h*}^{\left(l\right)}\right]$.
The smoothness parameter $\beta\ge1$ is a fixed positive integer.
The upper bound $\eta$, number of layers $L$, and number of observed
and hidden features $d,d_{h*}^{(1)},...,d_{h*}^{(L-1)}$ are finite
fixed constants. $\boldsymbol{x}_{i}\in\left[-1,1\right]^{d}$ for
every $i\in\left[n\right]$. 
\item The activation function, $\sigma:\mathbb{R}\rightarrow\mathbb{R}$,
is infinitely differentiable and non-polynomial, and this function
can be computed using a finite number of operations listed in Lemma
\ref{lem: Anthony and Bartlett Theorems 8.4 and 8.14}.
\item There exist finite positive fixed constants $c_{1}$, $c_{2}$, $c_{3}$,
and $c_{\ell}$ such that{\small{}
\begin{align}
c_{1}\mathbb{E}\left[\left(z_{i}\left(\boldsymbol{\theta}\right)-z_{*i}\left(\boldsymbol{f}_{*}\right)\right)^{2}\right] & \le\mathbb{E}\left[\ell\left(y_{i},z_{i}\left(\boldsymbol{\theta}\right)\right)-\ell\left(y_{i},z_{*i}\left(\boldsymbol{f}_{*}\right)\right)\right]\le c_{2}\mathbb{E}\left[\left(z_{i}\left(\boldsymbol{\theta}\right)-z_{*i}\left(\boldsymbol{f}_{*}\right)\right)^{2}\right],\label{eq: A1_3_curvature_1}\\
c_{3}\mathbb{E}\left[\left(z_{i}\left(\boldsymbol{\theta}\right)-z_{i}\left(\boldsymbol{\theta}_{*}\right)\right)^{2}\right] & \leq\mathbb{E}\left[\ell\left(y_{i},z_{i}\left(\boldsymbol{\theta}\right)\right)-\ell\left(y_{i},z_{i}\left(\boldsymbol{\theta}_{*}\right)\right)\right],\label{eq: A1_3_curvature_2}\\
\left|\ell\left(y,z_{1}\right)-\ell\left(y,z_{2}\right)\right| & \le c_{\ell}\left|z_{1}-z_{2}\right|,\label{eq: A1_3_lipschitz}
\end{align}
}for each $\boldsymbol{\theta}\in\Theta_{d_{h},\bar{z}}$, $y\in\mathcal{Y}$,
and $z_{1},z_{2}\in\left[-\bar{z},\bar{z}\right]$, where $\mathcal{Y}$
is the support of the outcome variable.
\item $\left(y_{i},\boldsymbol{\xi}_{i}\right)$ is identically distributed
over $i\in\left[n\right]$. There exists a sequence of constants $c_{n}$
such that $\max_{i\in\left[n\right]}\left|\mathcal{N}\left(i\right)\right|\leq c_{n}$
a.s.
\end{enumerate}
Assumption I 1 postulates that the functions composing network heterogeneity
are smooth with bounded input and output spaces. Given the bounded
nature of these input and output spaces, it is innocuous to normalize
the input and output spaces of each internal layer to be unit cubes.
Our theory can be adapted to allow the dimensions of hidden embeddings,
$d_{h*}^{\left(1\right)},...,d_{h*}^{\left(L-1\right)}$, to grow
slowly with the sample size, but we omit that complexity for concise
presentation. 

Assumption I 2 accommodates common activation functions such as the
sigmoid, $\sigma\left(x\right)=\frac{1}{1+\exp\left(-x\right)}$,
which takes four operations, and the tanh, $\sigma\left(x\right)=\frac{\exp\left(x\right)-\exp\left(-x\right)}{\exp\left(x\right)+\exp\left(-x\right)}$,
which takes six operations. As this paper focuses on the most prevalently
used GNN architectures, where each layer comprises a shallow neural
network, our setup does not account for activation functions that
are particularly effective for deep neural networks, such as the ReLU,
which is non-differentiable at zero. 

Assumption I 3 posits that the curvature of the loss function is bounded
from below and above at $z_{*i}\left(\boldsymbol{f}_{*}\right)$,
and bounded from below at $z_{i}\left(\boldsymbol{\theta}_{*}\right)$.
Furthermore, the loss function is Lipschitz continuous with respect
to the second argument. Regarding (\ref{eq: A1_3_curvature_1}) and
(\ref{eq: A1_3_lipschitz}), suppose $\mathbb{E}\left[y_{i}\mid\boldsymbol{h}_{*i}^{\left(L-1\right)}\left(\boldsymbol{f}_{*}\right),\overline{\boldsymbol{h}}_{*\mathcal{N}\left(i\right)}^{\left(L-1\right)}\left(\boldsymbol{f}_{*}\right)\right]=\mathbb{E}\left[y_{i}\mid\boldsymbol{\xi}_{i}\right]$,
which indicate that the penultimate layer latent embeddings sufficiently
summarize the local neighborhood information needed for the conditional
mean of the outcome variable. Then, \citet{farrell2021deep} shows
that (\ref{eq: A1_3_curvature_1}) and (\ref{eq: A1_3_lipschitz})
hold true for many commonly used loss functions, including the least
squares loss for bounded real-valued outcome variables, and logistic
loss for binary outcome variables. Condition (\ref{eq: A1_3_curvature_2})
is introduced in Section 5.2 of \citet{bartlett2005local} with valid
examples. If we could show that $\hat{\boldsymbol{\theta}}$ is a
consistent estimator of $\boldsymbol{\theta}_{*}$, we can relax this
condition by assuming it holds only for $\boldsymbol{\theta}\in\Theta_{d_{h},\bar{z}}$
in a small neighborhood around $\boldsymbol{\theta}_{*}$, and this
weaker condition can be easily verified for any loss function under
mild regularity conditions. Meanwhile, the consistency of the M-estimator
$\hat{\boldsymbol{\theta}}$ can be established using Theorem 5.7
in \Citet{Vaart1998}, where the uniform convergence of the criterion
function can be demonstrated. Yet, to simply the proof, we employ
the stronger condition in this paper as in \citet{bartlett2005local}.

Under Assumption I 4, $\left(y_{i},\boldsymbol{\xi}_{i}\right)$ is
identically distributed over $i\in\left[n\right]$, and each node
in the sample of size $n$ has at most $c_{n}$ neighbors. It would
be more precise to use the triangular array notation such that for
each $n$, $\left\{ \left(y_{n,i},\boldsymbol{\xi}_{n,i}\right)\right\} _{i\in\left[n\right]}$
denotes a sequence of identically distributed set of random variables
each distributed as $\left(y_{n,i},\boldsymbol{\xi}_{n,i}\right)\sim P_{n}$.
To streamline our notation, we suspend the use of subscript $n$ in
$\left(y_{n,i},\boldsymbol{\xi}_{n,i}\right)$ and in any object depending
on them ($c_{n}$ is an exception), although our analysis is valid
for the setting where the distribution of $\left(y_{n,i},\boldsymbol{\xi}_{n,i}\right)$
varies with $n$.

Assumption I 4 does not require the set of random variables $\left(y_{i},\boldsymbol{\xi}_{i}\right)$
to be independent across $i$. In our context, accounting for dependence
is important. Examples of the sources of dependence include (1) node
attributes $\boldsymbol{x}_{i}$ could be dependent across $i$, (2)
the $L$-hop local neighborhood $\boldsymbol{\xi}_{i}$ may overlap
across $i$, and (3) the edge formation process itself might induce
dependence. To account for dependence, we utilize the concepts of
\textit{dependency graph} and \textit{smallest proper cover}, as introduced
below. 
\begin{defn}
\label{def: Dependence-graph}{[}Dependency graph{]} A graph $G=\left(V,E\right)$
is the dependency graph associated to $\left\{ \left(y_{i},\boldsymbol{\xi}_{i}\right)\right\} _{i\in\left[n\right]}$
if (i) $V=\left[n\right]$ is the set of vertices, and (ii) $E\subseteq V\times V$
is the set of edges where $\left(i,j\right)\in E$ if and only if
$\left(y_{i},\boldsymbol{\xi}_{i}\right)$ and $\left(y_{j},\boldsymbol{\xi}_{j}\right)$
are dependent.
\end{defn}

\begin{defn}
\label{def: Smallest-proper-cover}{[}Smallest proper cover{]} Let
$G=\left(V,E\right)$ be a dependency graph. For some positive integer
$J$, $\mathcal{C}=\left\{ \mathcal{C}_{j}\right\} _{j\in\left[J\right]}$
is a proper cover of $G$ if (i) $\mathcal{C}_{j}$'s are disjoint
subsets of $V$ and $\cup_{j\in\left[J\right]}\mathcal{C}_{j}=V$,
and (ii) $\forall j\in\left[J\right]$, there are no connections based
on $E$ between vertices within $\mathcal{C}_{j}$ (i.e., the nodes
in $\mathcal{C}_{j}$ are independent). A smallest proper cover of
$G$ is a proper cover of $G$ with the smallest size $J$.
\end{defn}

Given the dependency graph of $\left\{ \left(y_{i},\boldsymbol{\xi}_{i}\right)\right\} _{i\in\left[n\right]}$,
we could construct the smallest proper cover based on the dependency
graph denoted by $\left\{ \mathcal{C}_{j}\right\} _{j\in\left[J\right]}$.\footnote{If the smallest proper cover is not unique, selecting any one of them
as $\left\{ \mathcal{C}_{j}\right\} _{j\in\left[J\right]}$ would
suffice.} By construction, $\left(y_{i},\boldsymbol{\xi}_{i}\right)$'s are
independent \textit{within} each cover $\mathcal{C}_{j}$ but can
be arbitrarily correlated \textit{across} covers. Both the number
of covers $J$ and the set of covers $\left\{ \mathcal{C}_{j}\right\} _{j\in\left[J\right]}$
depend on the sample size $n$, although we leave out the subscript
$n$ to simplify notations. In our analysis, we treat $\left\{ \mathcal{C}_{j}\right\} _{j\in\left[J\right]}$
as fixed for each sample size $n$, otherwise the analysis becomes
overly complicated.

The construction of the smallest proper cover based on dependency
graph can be better understood with the help of graph coloring, as
discussed in \citet{ralaivola2015entropy}. The idea of graph coloring
is to assign different colors to nodes in a graph so that no two adjacent
nodes share the same color. Using the \textit{dependency graph} for
graph coloring, nodes assigned to the same color are necessarily independent,
thereby enabling the construction of a proper cover in which each
cover $\mathcal{C}_{j}$ includes only the nodes of the same color
$j$. Note that this graph coloring is based on the \textit{dependency
graph}, rather than the observed network. Subsequently, a smallest
proper cover can be obtained through graph coloring using the \textit{minimal}
number of colors needed.\footnote{While the concept of graph coloring helps to understand the theoretical
results, in practice there is no necessity for empirical researchers
to construct the smallest proper cover for inference, as demonstrated
in Section \ref{sec:Empirical-Application} for the semiparametric
inference.}

\textbf{Remark} {[}Motivating example{]}: To illustrate, consider
the setting of our empirical application, where the dataset includes
many villages. In this setting, we suppose that households \textit{within}
the same village can form links and have correlated attributes and
outcome variables. However, households \textit{across} different villages
are unable to form links and their attributes and outcome variables
are considered independent. This leads to a specific pattern in our
dependency graph for the set $\left\{ \left(y_{i},\boldsymbol{\xi}_{i}\right)\right\} _{i\in\left[n\right]}$:
households within each village are interconnected, forming a complete
graph, while there are no connections between households from different
villages.

Given that households in the same village are completely connected
in the dependency graph, the graph coloring principle dictates that
each household in the same village must be assigned a unique color.
Hence, the \textit{number of covers} of a smallest proper cover ($J$),
which equates to the \textit{smallest number of colors} needed, corresponds
to the \textit{size of the largest village}. Moreover, the largest
cover size $\left(\max_{j\in\left[J\right]}\left|\mathcal{C}_{j}\right|\right)$
cannot exceed the number of villages. With appropriate coloring, we
could set $\left|\mathcal{C}_{j}\right|\asymp\frac{n}{J}$ for every
$j\in\left[J\right]$, ensuring that each cover is similar in size.

As discussed previously, our study treats the smallest proper cover
$\left\{ \mathcal{C}_{j}\right\} _{j\in\left[J\right]}$ as fixed
for each sample size $n$. Additionally, we assume that $\left(y_{i},\boldsymbol{\xi}_{i}\right)$
is identically distributed over $i\in\left[n\right]$, and the maximum
degree in $\boldsymbol{D}$ is bounded by $c_{n}$. Now we delve into
additional details to justify this modeling choice.

First, suppose that nature assigns each household to a village. Both
the number of villages and the specific assignment of households to
these villages are considered fixed for each sample size, conditional
on which we form our statistical analysis. Then, given the village
assignment, households randomly draw their vector of observed covariates
$\boldsymbol{x}_{i}$ and unobserved attributes $\boldsymbol{\epsilon}_{i}$
such that (1) $\left(\boldsymbol{x}_{i},\boldsymbol{\epsilon}_{i}\right)$'s
are identically distributed across $i\in\left[n\right]$, (2) $\left(\boldsymbol{x}_{i},\boldsymbol{\epsilon}_{i}\right)$'s
are independent for households in different villages, and (3) $\left(\boldsymbol{x}_{i},\boldsymbol{\epsilon}_{i}\right)$'s
are exchangeable over $i$ within the same village. Next, households
in the same village form network connections. This process could follow
various commonly used network formation models. Examples include the
Erd\H{o}s-R\'enyi model, in which each pair of nodes is connected
with a constant probability, the Barab\'asi-Albert model, which introduces
a preferential attachment process, and a generalized Erd\H{o}s-R\'enyi
model, where the probability of edge formation between two nodes $i$
and $j$ depends on their respective attributes $\left(\boldsymbol{x}_{i},\boldsymbol{\epsilon}_{i}\right)$
and $\left(\boldsymbol{x}_{j},\boldsymbol{\epsilon}_{j}\right)$,
and so on. An additional modification to the network formation model
can be applied to fulfill the requirement that each household has
at most $c_{n}$ connections, as long as this adjustment is identical
across $i\in\left[n\right]$. Finally, the outcome variable $y_{i}$
is formed based on $\boldsymbol{\xi}_{i}$, and possibly on the unobserved
attributes $\boldsymbol{\epsilon}_{i}$.

Following this procedure, if we are under the scenario where the sizes
of villages are either identical or approximately similar (to the
extent that the discrepancy in village size is negligible when $n$
is large), the assumption that $\left(y_{i},\boldsymbol{\xi}_{i}\right)$
is identically distributed over $i\in\left[n\right]$ holds true,
or at least serves as a reasonable approximation. We defer the relaxation
of the identical-distribution assumption to future research. $\square$

\subsection{Rate of convergence}

Utilizing the concept of the smallest proper cover $\left\{ \mathcal{C}_{j}\right\} _{j\in\left[J\right]}$,
we now present the population and empirical $L_{2}$ convergence rates
of the GNN estimator in Theorem \ref{thm: rate_of_convergence}. The
proof of this theorem can be found in Appendix \ref{sec:Proof-of-Theorem 1}.
\begin{thm}
\label{thm: rate_of_convergence}Under Assumptions I 1-4, for $k=2,4$
or $6$ and $s=L$ or $2L-1$ ($k$ and $s$ depend on the activation
function $\sigma\left(\cdot\right)$ and the number of layers $L$),
with probability at least $1-\text{\ensuremath{\exp\left(-\rho\right),}}${\small{}
\[
\mathbb{E}\left[\left(z_{i}\left(\hat{\boldsymbol{\theta}}\right)-z_{*i}\left(\boldsymbol{f}_{*}\right)\right)^{2}\right]\leq C\cdot\left(\left(\frac{1}{n}\sum_{j=1}^{J}\left(1+\log\left|\mathcal{C}_{j}\right|\right)\cdot\left(c_{n}\right)^{s}\right)^{\frac{\beta}{\beta+kd_{h*}}}+\frac{J\log J+J\rho}{n}\right),
\]
}and with probability at least $1-\text{\ensuremath{2\exp\left(-\rho\right)},}${\small{}
\[
\frac{1}{n}\sum_{i=1}^{n}\left(z_{i}\left(\hat{\boldsymbol{\theta}}\right)-z_{*i}\left(\boldsymbol{f}_{*}\right)\right)^{2}\leq C\left(\left(\frac{1}{n}\sum_{j=1}^{J}\left(1+\log\left|\mathcal{C}_{j}\right|\right)\cdot\left(c_{n}\right)^{s}\right)^{\frac{\beta}{\beta+kd_{h*}}}+\frac{J\log J+J\rho}{n}\right),
\]
}where $\beta$ is the smoothness parameter, $d_{h*}=d\vee d_{h*}^{(1)}\vee...\vee d_{h*}^{(L)}$
with $d$ being the number of covariates and $d_{h*}^{\left(l\right)}$
being the number of latent embeddings in each layer, $J$ is the number
of covers in the smallest proper cover as previously discussed, $\left|\mathcal{C}_{j}\right|$
is the size of each cover, and $c_{n}$ is the largest number of peers
a node could have. In particular, when $\sigma\left(\cdot\right)$
does not involve taking an exponential, we have $s=L$, and $k=2$
if $L=1$ and $k=4$ if $L\geq2$. When $\sigma\left(\cdot\right)$
involves taking an exponential, we have $s=2L-1$, and $k=4$ if $L=1$
and $k=6$ if $L\geq2$. The above results hold by setting $d_{h}^{\left(l\right)}\asymp\left(\frac{1}{n}\sum_{j=1}^{J}\left(1+\log\left|\mathcal{C}_{j}\right|\right)\cdot\left(c_{n}\right)^{s}\right)^{-\frac{d_{h*}}{\beta+kd_{h*}}}$
for every $l\in\left[L\right]$.
\end{thm}

Theorem \ref{thm: rate_of_convergence} implies that the GNN estimator
is consistent under relatively mild conditions that $J$ and $c_{n}$
do not grow too fast with $n$. In particular, if the growth rates
of $J$ and $c_{n}$ satisfy $J\log J=o\left(n\right)$ and $c_{n}=o\left(\left(\frac{n/J}{\log\left(n/J\right)}\right)^{1/s}\right)$,
then by Jensen's inequality the GNN estimator is consistent in terms
of the population and empirical $L_{2}$ norms.

Theorem \ref{thm: rate_of_convergence} also sheds light on the usage
of the GNN estimator for semiparametric inference. As will be shown
in Section \ref{sec:Empirical-Application}, the GNN estimator can
be used to conduct inference on treatment effects if its empirical
$L_{2}$ convergence rate is $o_{p}(n^{-\frac{1}{4}})$. To achieve
this rate, Theorem \ref{thm: rate_of_convergence} ensures that it
is sufficient if $J$ and $c_{n}$ do not grow too fast with $n$
such that $J\log J=o\left(n^{1/2}\right)$ and {\small{}
\begin{equation}
\frac{\log\left(n/J\right)}{n/J}\left(c_{n}\right)^{s}=o\left(n^{-\frac{\beta+kd_{h*}}{2\beta}}\right).\label{eq: comment_GNN_L2_bound}
\end{equation}
}Suppose that the growth rates of $J$ and $c_{n}$ satisfy $J=O\left(n^{\alpha}\right)$
and $c_{n}=O\left(n^{\gamma}\right)$. Then Condition (\ref{eq: comment_GNN_L2_bound})
is satisfied if $\frac{kd_{h*}}{\beta}<0.999-2\alpha-2\gamma s$,
where $k$ and $s$ are specified as in Theorem \ref{thm: rate_of_convergence}.
This condition is achievable if the growth rates of $J$ and $c_{n}$,
and the dimensions of the input variables and latent embeddings (as
captured by $d_{h*}$) are not too large, and the target function
is sufficiently smooth. 

It is worthwhile to compare the convergence rate of the GNN with that
of the MLP in \citet{farrell2021deep}. Given that the MLP approach
in \citet{farrell2021deep} is not designed for dependent data in
networks, it is appropriate to set $J=O\left(1\right)$ and $c_{n}=O\left(1\right)$
for comparison. As our setting allows observations to be arbitrarily
correlated across covers, setting $J=O\left(1\right)$ ensures that
the \textit{effective} number of independent observations in our setting
is proportional to the sample size $n$, rendering the GNN and MLP
setups comparable. Also, let $k=2$ to consider the GNN model with
one layer, where the activation function does not involve taking exponentials,
mirroring the setting in \citet{farrell2021deep}, and let $d_{h*}=d$
for simplicity. Then, Condition (\ref{eq: comment_GNN_L2_bound})
is satisfied if $\frac{2}{1-\delta}d<\beta$, where $\delta$ can
be any arbitrarily small constant. As a comparison, for the MLP with
$2d$ input variables, which include an observation's own characteristics
and the average characteristics of adjacent nodes, Theorem 3 in \citet{farrell2021deep}
requires that $\beta>2d$ for semiparametric inference. Hence, the
requirements on the smoothness parameter $\beta$ for the GNN and
MLP frameworks are comparable. Also, with more layers ($L\geq2$)
and other settings the same, the smoothness condition becomes $\frac{4}{1-\delta}d<\beta$,
which does not further change with $L$. This suggests that although
the GNN approach could incorporate the characteristics of neighbors
over multiple distances, it does not necessarily lead to a curse of
dimensionality issue with larger $L$ when $L\geq2$.

In our empirical setting with many villages, let $n_{v}$ denote the
number of observations in village $v$, and $\bar{v}$ the number
of villages. Suppose that the village sizes are uniformly proportional
to $J$ (i.e., $\min_{v\in\left[\bar{v}\right]}n_{v}\asymp J=\max_{v\in\left[\bar{v}\right]}n_{v}$),
then Condition (\ref{eq: comment_GNN_L2_bound}) is equivalent to
$\frac{\log\left(\bar{v}\right)}{\bar{v}}\left(c_{n}\right)^{s}=o\left(n^{-\frac{\beta+kd_{h*}}{2\beta}}\right)$.
If the growth rates of $\bar{v}$ and $c_{n}$ satisfy $n^{\eta}\lesssim\bar{v}$
and $c_{n}=O\left(n^{\gamma}\right)$, then Condition (\ref{eq: comment_GNN_L2_bound})
holds when $\frac{kd_{h*}}{\beta}<2\eta-2\gamma s-1.001$, which
is achievable if the number of villages grows sufficiently fast, both
the growth rate of $c_{n}$ and the dimensions of the input variables
and latent embeddings are not too large, and the target function is
sufficiently smooth.

\section{Semiparametric inference \label{sec:Inference}}

In this section, we show that the theoretical guarantee of the GNN
estimator in Theorem \ref{thm: rate_of_convergence} can be applied
to conduct causal inference on parameters with heterogeneous treatment
effects. This section closely follows \citet{farrell2021deep} and
hence our exposition is brief. More detailed explanations of the setup
can be found in \citet{farrell2021deep} and references therein. The
main distinction in this section is our need to adapt the inference
procedure to accommodate dependent data with networks.

\subsection{Setup}

Consider potential outcomes $\left(y_{i}\left(1\right),y_{i}\left(0\right)\right)$
for a binary treatment $t_{i}\in\left\{ 0,1\right\} $. Then the observable
outcome can be expressed as $y_{i}=t_{i}y_{i}\left(1\right)+\left(1-t_{i}\right)y_{i}\left(0\right)$.

To illustrate the idea, we follow \citet{farrell2021deep} and focus
on the \textit{average outcome for a counterfactual treatment policy}
$\pi\left(s\right)=\mathbb{E}\left[s\left(\boldsymbol{\xi}_{i}\right)y_{i}\left(1\right)+\left(1-s\left(\boldsymbol{\xi}_{i}\right)\right)y_{i}\left(0\right)\right]$,
where $s\left(\boldsymbol{\xi}_{i}\right)\in\left\{ 0,1\right\} $
is a known deterministic counterfactual policy that assigns treatment
status to each observation based on $\boldsymbol{\xi}_{i}$. For convenience,
we denote $s_{1}\left(\boldsymbol{\xi}_{i}\right)=s\left(\boldsymbol{\xi}_{i}\right)$
and $s_{0}\left(\boldsymbol{\xi}_{i}\right)=1-s\left(\boldsymbol{\xi}_{i}\right)$.
Similar analysis directly applies to other popular parameters of interest,
such as the average treatment effect (ATE) $\tau=\mathbb{E}\left[y_{i}\left(1\right)-y_{i}\left(0\right)\right]$
and the average treatment effect on the treated (ATT) $\mathbb{E}\left[y_{i}\left(1\right)-y_{i}\left(0\right)\mid t_{i}=1\right]$.

For $t\in\left\{ 0,1\right\} $, define 

{\small{}
\begin{align}
\boldsymbol{f}_{*}^{t} & \in\arg\min_{\boldsymbol{f}}\mathbb{E}\left[\ell_{y}\left(y_{i}\left(t\right),z_{*i}\left(\boldsymbol{f}\right)\right)\right],\ \boldsymbol{f}_{*}^{p}\in\arg\min_{\boldsymbol{f}}\mathbb{E}\left[\ell_{p}\left(t_{i},z_{*i}\left(\boldsymbol{f}\right)\right)\right],\label{eq: def_f_star_t_p}
\end{align}
}where the loss function is the least squares loss for real-valued
dependent variables and logistic loss for binary dependent variables.
Then, $z_{*i}\left(\boldsymbol{f}_{*}^{t}\right)$ and $z_{*i}\left(\boldsymbol{f}_{*}^{p}\right)$
are the individual network heterogeneities of node $i$ that affect
the potential outcomes and treatment assignment, respectively. The
construction of the network heterogeneities $z_{*i}\left(\boldsymbol{f}_{*}^{t}\right)$
and $z_{*i}\left(\boldsymbol{f}_{*}^{p}\right)$ allows their dependence
on node $i$'s local neighborhood surroundings. The optimization problems
in (\ref{eq: def_f_star_t_p}) are specific instances of the one in
(\ref{eq: define_f_star}), where we replace the generic outcome variable
by the potential outcomes and treatment variable.

Then, the optimization problems in (\ref{eq: def_f_star_t_p}) imply
that 
\begin{equation}
\mu_{t}\left(\boldsymbol{\xi}_{i}\right)\coloneqq\mathbb{E}\left[y_{i}\left(t\right)\mid\boldsymbol{h}_{*i,t}^{\left(L-1\right)}\left(\boldsymbol{f}_{*}^{t}\right),\overline{\boldsymbol{h}}_{*\mathcal{N}\left(i\right),t}^{\left(L-1\right)}\left(\boldsymbol{f}_{*}^{t}\right)\right]=\left\{ \begin{array}{c}
z_{*i}\left(\boldsymbol{f}_{*}^{t}\right)\ \text{if \ensuremath{y_{i}\left(t\right)\in\ }\ensuremath{\mathbb{R}}}\\
\frac{\exp\left(z_{*i}\left(\boldsymbol{f}_{*}^{t}\right)\right)}{1+\exp\left(z_{*i}\left(\boldsymbol{f}_{*}^{t}\right)\right)}\ \text{if \ensuremath{y_{i}\left(t\right)\in\left\{  0,1\right\} } }
\end{array}\right.\label{eq: def_mu_t}
\end{equation}
and 
\begin{equation}
p\left(\boldsymbol{\xi}_{i}\right)\coloneqq\Pr\left(t_{i}=1\mid\boldsymbol{h}_{*i,p}^{\left(L-1\right)}\left(\boldsymbol{f}_{*}^{p}\right),\overline{\boldsymbol{h}}_{*\mathcal{N}\left(i\right),p}^{\left(L-1\right)}\left(\boldsymbol{f}_{*}^{p}\right)\right)=\frac{\exp\left(z_{*i}\left(\boldsymbol{f}_{*}^{p}\right)\right)}{1+\exp\left(z_{*i}\left(\boldsymbol{f}_{*}^{p}\right)\right)},\label{eq: def_p_xi_i}
\end{equation}
where, as usual, $\boldsymbol{h}_{*i,t}^{\left(L-1\right)}\left(\boldsymbol{f}_{*}^{t}\right)$
and $\boldsymbol{h}_{*i,p}^{\left(L-1\right)}\left(\boldsymbol{f}_{*}^{p}\right)$
are the penultimate layer latent embeddings in the construction of
$z_{*i}\left(\boldsymbol{f}_{*}^{t}\right)$ and $z_{*i}\left(\boldsymbol{f}_{*}^{p}\right)$,
respectively, and $\overline{\boldsymbol{h}}_{*\mathcal{N}\left(i\right),t}^{\left(L-1\right)}\left(\boldsymbol{f}_{*}^{t}\right)$
and $\overline{\boldsymbol{h}}_{*\mathcal{N}\left(i\right),p}^{\left(L-1\right)}\left(\boldsymbol{f}_{*}^{p}\right)$
are the corresponding average latent embeddings of the adjacent neighbors.

For convenience, we denote the conditional expectation of the potential
outcome $y_{i}\left(t\right)$ and of the treatment status $t_{i}$
as $\mu_{t}\left(\boldsymbol{\xi}_{i}\right)$ and $p\left(\boldsymbol{\xi}_{i}\right),$
as indicated in (\ref{eq: def_mu_t}) and (\ref{eq: def_p_xi_i}).
As such, $\mu_{t}\left(\boldsymbol{\xi}_{i}\right)$ and $p\left(\boldsymbol{\xi}_{i}\right)$
can be expressed by the network heterogeneity variables, $z_{*i}\left(\boldsymbol{f}_{*}^{t}\right)$
and $z_{*i}\left(\boldsymbol{f}_{*}^{p}\right)$. For simplicity,
we denote $\mu_{t}\left(\boldsymbol{\xi}_{i}\right)=h_{\mu}\left(z_{*i}\left(\boldsymbol{f}_{*}^{t}\right)\right)$
and $p\left(\boldsymbol{\xi}_{i}\right)=h_{p}\left(z_{*i}\left(\boldsymbol{f}_{*}^{p}\right)\right)$,
where $h_{\mu}\left(\cdot\right)$ and $h_{p}\left(\cdot\right)$
are some known Lipschitz continuous functions with finite Lipschitz
constants. Also, we denote $p_{1}\left(\boldsymbol{\xi}_{i}\right)=p\left(\boldsymbol{\xi}_{i}\right)$
and $p_{0}\left(\boldsymbol{\xi}_{i}\right)=1-p\left(\boldsymbol{\xi}_{i}\right)$.

Given that $\boldsymbol{f}_{*}^{t}$ and $\boldsymbol{f}_{*}^{p}$
may not be uniquely determined in (\ref{eq: def_f_star_t_p}), the
definitions of $\mu_{t}\left(\boldsymbol{\xi}_{i}\right)$ and $p\left(\boldsymbol{\xi}_{i}\right)$
may also depend on the specific choice of $\boldsymbol{f}_{*}^{t}$
and $\boldsymbol{f}_{*}^{p}$. We omit such dependence in the notations
for simplicity, as only the optimal $\boldsymbol{f}_{*}^{t}$'s and
$\boldsymbol{f}_{*}^{p}$'s that satisfy Assumption II 6 below are
relevant for our discussion and proof.

Then, Theorem \ref{thm: rate_of_convergence}, coupled with the conditions
on $d_{h*}$, $\beta$, $J$ and $c_{n}$, implies that $\mu_{t}\left(\boldsymbol{\xi}_{i}\right)$
and $p\left(\boldsymbol{\xi}_{i}\right)$ can be estimated with sufficient
accuracy for inference. Denote $\hat{\mu}_{t}\left(\boldsymbol{\xi}_{i}\right)=h_{\mu}\left(z_{i}\left(\hat{\boldsymbol{\theta}}_{t}\right)\right)$,
$\hat{p}_{1}\left(\boldsymbol{\xi}_{i}\right)=h_{p}\left(z_{i}\left(\hat{\boldsymbol{\theta}}_{p}\right)\right)$,
and $\hat{p}_{0}\left(\boldsymbol{\xi}_{i}\right)=1-\hat{p}_{1}\left(\boldsymbol{\xi}_{i}\right)$
as the GNN estimators for $\mu_{t}\left(\boldsymbol{\xi}_{i}\right)$,
$p_{1}\left(\boldsymbol{\xi}_{i}\right)$, and $p_{0}\left(\boldsymbol{\xi}_{i}\right)$,
respectively, where $\hat{\boldsymbol{\theta}}_{t}$ and $\hat{\boldsymbol{\theta}}_{p}$
are estimated using the following GNN architectures

{\small{}
\begin{align}
\hat{\boldsymbol{\theta}}_{t} & \in\underset{\begin{array}{c}
\boldsymbol{\theta}\in\Theta_{d_{h},\bar{z}}\end{array}}{\arg\min}\sum_{i:t_{i}=t}\ell_{y}\left(y_{i},z_{i}\left(\boldsymbol{\theta}\right)\right),\ \hat{\boldsymbol{\theta}}_{p}\in\underset{\begin{array}{c}
\boldsymbol{\theta}\in\Theta_{d_{h},\bar{z}}\end{array}}{\arg\min}\sum_{i=1}^{n}\ell_{p}\left(t_{i},z_{i}\left(\boldsymbol{\theta}\right)\right).\label{eq: gnn_estimator_theta_t_p}
\end{align}
}The estimators $\hat{\boldsymbol{\theta}}_{t}$ and $\hat{\boldsymbol{\theta}}_{p}$
may not be unique, while our discussion holds for any $\hat{\boldsymbol{\theta}}_{t}$
and $\hat{\boldsymbol{\theta}}_{p}$ that solve the minimization problems
in (\ref{eq: gnn_estimator_theta_t_p}). Moreover, denote $\boldsymbol{\upsilon}{}_{i}=\left(y_{i},t_{i},\boldsymbol{\xi}_{i}\right)$,
and we further introduce the following notations{\small{}
\begin{eqnarray*}
 &  & \varphi_{t}\left(\boldsymbol{\upsilon}{}_{i}\right)=\frac{\mathbb{I}\left\{ t_{i}=t\right\} }{p_{t}\left(\boldsymbol{\xi}_{i}\right)}\left(y_{i}-\mu_{t}\left(\boldsymbol{\xi}_{i}\right)\right)+\mu_{t}\left(\boldsymbol{\xi}_{i}\right),\ \hat{\varphi}_{t}\left(\boldsymbol{\upsilon}{}_{i}\right)=\frac{\mathbb{I}\left\{ t_{i}=t\right\} }{\hat{p}_{t}\left(\boldsymbol{\xi}_{i}\right)}\left(y_{i}-\hat{\mu}_{t}\left(\boldsymbol{\xi}_{i}\right)\right)+\hat{\mu}_{t}\left(\boldsymbol{\xi}_{i}\right),\\
 &  & \zeta_{i}=s_{1}\left(\boldsymbol{\xi}_{i}\right)\varphi_{1}\left(\boldsymbol{\upsilon}{}_{i}\right)+s_{0}\left(\boldsymbol{\xi}_{i}\right)\varphi_{0}\left(\boldsymbol{\upsilon}{}_{i}\right),\ \hat{\zeta}_{i}=s_{1}\left(\boldsymbol{\xi}_{i}\right)\hat{\varphi}_{1}\left(\boldsymbol{\upsilon}{}_{i}\right)+s_{0}\left(\boldsymbol{\xi}_{i}\right)\hat{\varphi}_{0}\left(\boldsymbol{\upsilon}{}_{i}\right),\\
 &  & u_{i}\left(t\right)=y_{i}\left(t\right)-\mu_{t}\left(\boldsymbol{\xi}_{i}\right),\ \Sigma_{n}=\text{Var}\left(\frac{1}{\sqrt{n}}\sum_{i=1}^{n}\zeta_{i}\right).
\end{eqnarray*}
}Then the counterfactual policy effect can be identified as $\pi\left(s\right)=\mathbb{E}\left[\zeta_{i}\right]$
under proper conditions, and our GNN-based estimator of $\pi\left(s\right)$
is constructed as{\small{}
\[
\hat{\pi}\left(s\right)=\frac{1}{n}\sum_{i=1}^{n}\hat{\zeta}_{i}.
\]
}In the next subsections, we study the asymptotic distribution of
this estimator.

\textbf{Remark} {[}SUTVA assumption{]}: It is important to note that
the setup described above relies on the stable unit treatment value
assumption (SUTVA), which assumes that an individual's potential outcome
only depends on their own treatment assignment. However, a natural
extension of this setup allows an individual's potential outcome to
also depend on the treatment assignments of other people. This has
been extensively studied in the network interference literature (e.g.,
\citet{leung2022causal} and \citet{savje2021average}), where the
potential outcome is modeled as $y_{i}(\left\{ t_{j}\right\} _{j\in\left[n\right]})$,
instead of $y_{i}\left(t_{i}\right)$.

Although interference is prevalent in many contexts, studying network
heterogeneity under SUTVA also holds significant empirical value.
In particular, our empirical application studies a substantive policy
question about selecting a group of `seed' information recipients,
who will then diffuse information throughout social networks via word-of-mouth
(\citet{Rogers2003diffusion} and \citet{banerjee2019using}). To
address this question, it is useful for policymakers to understand
the average potential outcome \textit{among the chosen seed information
recipients} under different selection rules.\footnote{This average potential outcome can be learnt once $\pi\left(s\right)$
is known. Specifically, suppose the policymaker considers a selection
rule which picks observation $i$ as an information recipient (so
$s\left(\boldsymbol{\xi}_{i}\right)=1$) if their $\boldsymbol{\xi}_{i}$
meet certain criteria. Also, for simplicity, assume that $y_{i}\left(0\right)=0$
and the selection probability $Pr\left(s\left(\boldsymbol{\xi}_{i}\right)=1\right)$
is known, as in our empirical application. Then, it follows that $\mathbb{E}\left[y_{i}\left(1\right)\mid s\left(\boldsymbol{\xi}_{i}\right)=1\right]=\frac{\pi\left(s\right)}{Pr\left(s\left(\boldsymbol{\xi}_{i}\right)=1\right)}$,
showing that the estimation of $\pi\left(s\right)$ can help the policymaker
obtain the estimation of the average potential outcome among the seed
information recipients.} This average potential outcome substantially affects the success
of the information diffusion process. In the context of \citet{BanerjeeChandrasekharDufloJackson2013},
for instance, the average potential outcome equates to the probability
that a seed household participates in a microfinance program, which
largely affects the likelihood that the household subsequently informs
its peers about the program. Nevertheless, interference is not necessarily
an issue in addressing this policy question. This is because the seed
information recipients make their participation decisions \textit{before}
their peers are informed about the program (this is consistent with
the first-step model for estimating the characteristic coefficients
$\beta$ using the leaders' participation decisions in \citet{BanerjeeChandrasekharDufloJackson2013}).

Our GNN approach has the potential to contribute to the literature
of network interference as well. By including the treatment variable
$t_{i}$ as a component of the node attribute vector $\boldsymbol{x}_{i}$,
the GNN framework allows an individual's outcome to depend on the
treatment assignments of neighbors up to a distance of $L$. Nonetheless,
as the empirical focus of our paper is not on network interference,
we leave the detailed discussion to future research. $\square$

\subsection{Assumptions}

Assumption II provides the conditions for deriving the asymptotic
distribution of the estimator $\hat{\pi}\left(s\right)$ for the average
effect of a counterfactual policy $\pi\left(s\right)$.

\subsubsection{Assumption II}
\begin{enumerate}
\item $\left\{ \left(y_{i}\left(1\right),y_{i}\left(0\right),t_{i},\boldsymbol{\xi}_{i}\right)\right\} _{i\in\left[n\right]}$
is a sequence of identically distributed set of random variables.
\item $\max_{i}\left|z_{*i}\left(\boldsymbol{f}_{*}^{t}\right)\right|<c_{z}<\infty$
a.s. with some fixed constant $c_{z}$ for $t\in\left\{ 0,1\right\} $.
\item The first-stage GNN estimators $\hat{\mu}_{t}\left(\boldsymbol{\xi}_{i}\right)$
and $\hat{p}_{t}\left(\boldsymbol{\xi}_{i}\right)$ satisfy that for
$t\in\left\{ 0,1\right\} $, 
\begin{enumerate}
\item $\frac{1}{n}\sum_{i=1}^{n}\left(\hat{p}_{t}\left(\boldsymbol{\xi}_{i}\right)-p_{t}\left(\boldsymbol{\xi}_{i}\right)\right)^{2}=o_{p}\left(1\right)$
and $\frac{1}{n}\sum_{i=1}^{n}\left(\hat{\mu}_{t}\left(\boldsymbol{\xi}_{i}\right)-\mu_{t}\left(\boldsymbol{\xi}_{i}\right)\right)^{2}=o_{p}\left(1\right)$,
\item $\left(\frac{1}{n}\sum_{i=1}^{n}\left(\hat{\mu}_{t}\left(\boldsymbol{\xi}_{i}\right)-\mu_{t}\left(\boldsymbol{\xi}_{i}\right)\right)^{2}\right)^{1/2}\left(\frac{1}{n}\sum_{i=1}^{n}\left(\hat{p}_{t}\left(\boldsymbol{\xi}_{i}\right)-p_{t}\left(\boldsymbol{\xi}_{i}\right)\right)^{2}\right)^{1/2}=o_{p}\left(n^{-1/2}\right)$,
\item $\frac{1}{n}\sum_{i=1}^{n}s_{t}\left(\boldsymbol{\xi}_{i}\right)\left(\hat{\mu}_{t}\left(\boldsymbol{\xi}_{i}\right)-\mu_{t}\left(\boldsymbol{\xi}_{i}\right)\right)\left(1-\frac{\mathbb{I}\left\{ t_{i}=t\right\} }{p_{t}\left(\boldsymbol{\xi}_{i}\right)}\right)=o_{p}\left(n^{-1/2}\right).$
\end{enumerate}
\item $p_{\min}\leq p_{1}\left(\boldsymbol{\xi}_{i}\right)\leq1-p_{\min}$
and $p_{\min}\leq\hat{p}_{1}\left(\boldsymbol{\xi}_{i}\right)\leq1-p_{\min}$
a.s. for every $i$ and some fixed constant $0<p_{\min}<\frac{1}{2}$. 
\item $\mathbb{E}\left[y_{i}\left(t\right)\mid\boldsymbol{\xi}_{i},t_{i}\right]=\mathbb{E}\left[y_{i}\left(t\right)\mid\boldsymbol{\xi}_{i}\right]$
a.s. for every $i$ and $t\in\left\{ 0,1\right\} $.
\item $\mathbb{E}\left[y_{i}\left(t\right)\mid\boldsymbol{h}_{*i,t}^{\left(L-1\right)}\left(\boldsymbol{f}_{*}^{t}\right),\overline{\boldsymbol{h}}_{*\mathcal{N}\left(i\right),t}^{\left(L-1\right)}\left(\boldsymbol{f}_{*}^{t}\right)\right]=\mathbb{E}\left[y_{i}\left(t\right)\mid\boldsymbol{\xi}_{i}\right]$
and\\
$\Pr\left(t_{i}=1\mid\boldsymbol{h}_{*i,p}^{\left(L-1\right)}\left(\boldsymbol{f}_{*}^{p}\right),\overline{\boldsymbol{h}}_{*\mathcal{N}\left(i\right),p}^{\left(L-1\right)}\left(\boldsymbol{f}_{*}^{p}\right)\right)=\Pr\left(t_{i}=1\mid\boldsymbol{\xi}_{i}\right)$
a.s. for every $i\in\left[n\right]$, $t\in\left\{ 0,1\right\} $,
and some $\boldsymbol{f}_{*}^{t}$ and $\boldsymbol{f}_{*}^{p}$.
\item $\mathbb{E}\left[u_{i}\left(t\right)u_{j}\left(t\right)\mid\left\{ \boldsymbol{\xi}_{i},t_{i}\right\} _{i=1,...,n}\right]=0$
a.s. for every $i,j\in\left[n\right]$ such that $i\neq j$, $t\in\left\{ 0,1\right\} $,
and all $n$.
\item $\max_{i\in\left[n\right]}\mathbb{E}\left[\left|y_{i}\left(t\right)\right|^{2}\mid\left\{ \boldsymbol{\xi}_{i},t_{i}\right\} _{i=1,...,n}\right]\leq c_{y}<\infty$
a.s. for $t\in\left\{ 0,1\right\} $ and all $n$ with some fixed
constant $c_{y}>0$. 
\item $\underset{n\rightarrow\infty}{\lim\inf}\ \Sigma_{n}>c_{\sigma}$
for some fixed constant $c_{\sigma}>0$. 
\item Let $\omega_{n}$ be the maximal degree of the dependency graph of
$\left\{ \boldsymbol{\upsilon}{}_{1},...,\boldsymbol{\upsilon}{}_{n}\right\} $,
and set $\omega_{n}=1$ if the dependency graph has no edges. $\frac{\omega_{n}\sum_{i=1}^{n}\mathbb{E}\left[\zeta_{i}^{2}\mathbb{I}\left\{ \left|\zeta_{i}\right|>a_{n}\right\} \right]}{n\Sigma_{n}}\rightarrow0$
and $\frac{\left(n\right)^{\frac{1}{m}}\left(\omega_{n}\right)^{\frac{m-1}{m}}a_{n}}{\sqrt{n\Sigma_{n}}}\rightarrow0$
for some sequence of real numbers $a_{n}$ and integer $m$. 
\end{enumerate}
In particular, the first condition assumes the data is identically
distributed, and the second condition requires the individual network
heterogeneity is uniformly bounded. Both conditions have been imposed
in the derivation of Theorem \ref{thm: rate_of_convergence}.

Conditions 3 (a)-(c) assume that the first-stage GNN estimators are
well-behaved. This condition can be justified by Theorem \ref{thm: rate_of_convergence}
assuming that the prerequisites are fulfilled to ensure the empirical
$L_{2}$ convergence rate of the GNN estimators is $o_{p}\left(n^{-1/4}\right)$,
as discussed following Theorem \ref{thm: rate_of_convergence}. More
details for verifying Condition 3 (c) are provided in Appendix \ref{sec: Verification-of-Assumption II3(c)}.

Conditions 4 and 5 include the unconfoundedness and overlap assumptions,
which are standard identification conditions in the treatment effects
literature. And the estimated treatment probability is assumed to
be bounded inside the interval $\left(0,1\right)$. We assume that
the number of layers $L$ in $\boldsymbol{\xi}_{i}$ to fulfill the
unconfoundedness condition is known to the researcher, and this $L$
is used throughout this section to construct $\boldsymbol{\xi}_{i}$,
to define the network heterogeneities $z_{*i}\left(\boldsymbol{f}_{*}^{t}\right)$
and $z_{*i}\left(\boldsymbol{f}_{*}^{p}\right)$, and to formulate
the GNN estimators $z_{i}(\hat{\boldsymbol{\theta}}_{t})$ and $z_{i}(\hat{\boldsymbol{\theta}}_{p})$.
We utilize the same $L$ throughout our analysis in this section to
simplify our presentation, which however is not a binding constraint.\footnote{We could achieve the same asymptotic distribution in Corollary \ref{cor: asymptotic_distribution}
by incorporating different numbers of layers across object constructions.
In particular, we could allow the counterfactual policy $s\left(\boldsymbol{\xi}_{i,L_{1}}\right)$
to depend on the $L_{1}$-hop neighborhood. And we could define the
network heterogeneities $z_{*i}\left(\boldsymbol{f}_{*,L_{2}^{1}}^{1}\right)$,
$z_{*i}\left(\boldsymbol{f}_{*,L_{2}^{0}}^{0}\right)$, and $z_{*i}\left(\boldsymbol{f}_{*,L_{3}}^{p}\right)$,
as well as the GNN estimators $z_{i}\left(\hat{\boldsymbol{\theta}}_{1,L_{2}^{1}}\right)$,
$z_{i}\left(\hat{\boldsymbol{\theta}}_{0,L_{2}^{0}}\right)$, and
$z_{i}\left(\hat{\boldsymbol{\theta}}_{p,L_{3}}\right)$, with each
depending on the respective $L_{2}^{1}$, $L_{2}^{0}$, and $L_{3}$-hop
neighborhoods. Moreover, we could introduce the assumptions that the
unconfoundedness condition holds when conditioning on at least the
$L_{4}^{t}$-hop neighborhood (i.e., $\mathbb{E}\left[y_{i}\left(t\right)\mid\boldsymbol{\xi}_{i,L},t_{i}\right]=\mathbb{E}\left[y_{i}\left(t\right)\mid\boldsymbol{\xi}_{i,L}\right]$
for $L\geq L_{4}^{t}$ and $t\in\left\{ 0,1\right\} $); the expected
potential outcomes depend on at most the $L_{2*}^{t}$-hop neighborhood
(i.e., $\mathbb{E}\left[y_{i}\left(t\right)\mid\boldsymbol{\xi}_{i,L}\right]=\mathbb{E}\left[y_{i}\left(t\right)\mid\boldsymbol{\xi}_{i,L_{2*}^{t}}\right]$
for $L\geq L_{2*}^{t}$ and $t\in\left\{ 0,1\right\} $); and the
treatment probability depends on at most the $L_{3*}$-hop neighborhood
(i.e., $\Pr\left(t_{i}=1\mid\boldsymbol{\xi}_{i,L}\right)=\Pr\left(t_{i}=1\mid\boldsymbol{\xi}_{i,L_{3*}}\right)$
for $L\geq L_{3*}$). Then, the result in Corollary \ref{cor: asymptotic_distribution}
can be established if the numbers of layers in the GNN estimations
satisfy $L_{2}^{t}\geq L_{2*}^{t}$ for $t\in\left\{ 0,1\right\} $
and $L_{3}\geq L_{3*}$. In practice, $L_{2*}^{t}$ and $L_{3*}$
may be unknown, while the authors are currently working on a follow-up
paper to formally study the procedure for selecting $L_{2}^{t}$ and
$L_{3}$.}

Condition 6 posits that the penultimate layer latent embeddings effectively
capture the relevant local neighborhood information in $\boldsymbol{\xi}_{i}$
required to determine the conditional expectations of potential outcomes
and the propensity score. Condition 7 assumes that the potential outcome
residuals are uncorrelated across observations, conditional on the
local neighborhoods and treatment decisions. Condition 8 is a mild
regularity condition, assuming that the potential outcomes have finite
second moments conditional on the local neighborhoods and treatment
decisions. 

Condition 9 assumes that $\Sigma_{n}$ is bounded away from zero in
the limit. If the data were i.i.d., this condition holds, as $\Sigma_{n}=\text{Var}\left(\zeta_{i}\right)>0$.
And $\Sigma_{n}$ would be larger when there is a positive correlation
among observations, which is commonly observed in network settings
where individuals share common friends. Hence, this condition is easily
satisfied. 

Condition 10 includes the primitive conditions needed to obtain the
central limit theorem for dependent data, which allows $\zeta_{i}$
to be unbounded. If we additionally assume $\zeta_{i}$ is uniformly
bounded (i.e., $\max_{i}\left|\zeta_{i}\right|<c_{\zeta}<\infty$
for some fixed constant $c_{\zeta}$), which naturally holds for binary
outcome variables, Condition 10 can be replaced with a simpler condition
that $\omega_{n}$ grows slower than $\sqrt{n}$ such that $\frac{\left(\omega_{n}\right)^{1-\frac{1}{m}}}{\left(n\right)^{\frac{1}{2}-\frac{1}{m}}}\to0$
for some integer $m$, as discussed in Theorem 2 of \citet{janson1988normal}.
Essentially, Condition 10 requires that the tails of the distribution
of $\zeta_{i}$ do not decay too slowly and that $\omega_{n}$ does
not grow too fast. In the context of clustered data which we will
discuss next, Condition 10 holds true under Assumptions II 2 and 4,
along with Assumptions III 2, 3 and 4. 

Under Assumption II, we can derive the asymptotic distribution of
the estimator $\hat{\pi}\left(s\right)$, which depends on the unknown
parameter $\Sigma_{n}$, as presented in Corollary \ref{cor: asymptotic_distribution}.
Then for feasible inference, we introduce some additional conditions,
which are suitable for many empirical settings using clustered data,
including our empirical context where the data is separated by villages.
These conditions are listed under Assumption III.

\subsubsection{Assumption III}
\begin{enumerate}
\item A random sample of size $n$, $\left\{ \zeta_{i}:i=1,\ldots,n\right\} $,
can be reorganized as clustered data $\left\{ \zeta_{cj}:j=1,...,n_{c}\right\} _{c\in\left[\bar{c}\right]}$,
where there are $\bar{c}$ mutually-exclusive clusters and $n_{c}$
observations in each cluster $c$. The $\zeta_{cj}$'s are independent
across clusters but can be arbitrarily correlated within clusters.
And $\bar{c}\rightarrow\infty$ as $n\rightarrow\infty$.
\item $\underset{c\in\left[\bar{c}\right]}{\max}\ n_{c}=O\left(n/\bar{c}\right)$
as $n\rightarrow\infty$.
\item $\underset{i}{\max}\ \mathbb{E}\left[\left|y_{i}\left(t\right)\right|^{2+\delta}\right]\leq c_{y}<\infty$
for $t\in\left\{ 0,1\right\} $ with some fixed constants $\delta>0$
and $c_{y}>0$. 
\item $\underset{n\rightarrow\infty}{\lim\inf}\ \frac{1}{n/\bar{c}}\Sigma_{n}>c_{\sigma}$
for some fixed constant $c_{\sigma}>0$. 
\end{enumerate}
Condition 1 assumes a clustered structure on the data. We treat the
clusters as fixed for each sample size, otherwise the analysis would
be overly complicated. So the number of clusters $\bar{c}$ and cluster
sizes $\left\{ n_{c}\right\} _{c\in\left[\bar{c}\right]}$ depend
on $n$ deterministically. Condition 1 requires that the number of
clusters grows with the sample size. This condition suits many empirical
datasets with inherent grouping structures, such as students in different
schools, employees in various companies, or patients in different
hospitals, among others. In our empirical setting, where households
are separated by villages, it is also reasonable to allow for free
intra-village dependencies while assuming away the possibility of
inter-village dependencies. For inference, we apply the same partitions
that are used to cluster $\left\{ \zeta_{i}:i=1,\ldots,n\right\} $
to partition their estimators $\left\{ \hat{\zeta}_{i}:i=1,\ldots,n\right\} $,
which yields $\left\{ \hat{\zeta}_{cj}:j=1,...,n_{c}\right\} _{c\in\left[\bar{c}\right]}$
used in Corollary \ref{cor: asymptotic_distribution}.

Condition 2 assumes that each cluster size grows uniformly no faster
than the average cluster size. This condition allows the cluster sizes
to stay finite. Condition 3 is another primitive condition, which
is slightly stronger than assuming that the potential outcomes uniformly
have finite second moments. Condition 4 is equivalent to the assumption
that $\underset{n\rightarrow\infty}{\lim\inf}\ensuremath{\frac{1}{\bar{c}}}\sum_{c\in\left[\bar{c}\right]}\ensuremath{\left(\frac{n_{c}}{n/\bar{c}}\right)^{2}\mathbb{E}\left[\left(\bar{\zeta}_{c}-\mu\right)^{2}\right]>c_{\sigma}>0}$,
where $\bar{\zeta}_{c}\coloneqq\frac{1}{n_{c}}\sum_{j\in\left[n_{c}\right]}\zeta_{cj}$
and $\mu\coloneqq\mathbb{E}\left[\zeta_{i}\right]$, which essentially
requires that there is no weak dependence among observations within
each cluster.

\subsection{Asymptotic distribution}

In Corollary \ref{cor: asymptotic_distribution}, we present the asymptotic
distribution of the counterfactual policy effect estimator $\hat{\pi}\left(s\right)$
and a feasible estimator of the asymptotic variance for valid inference.
The proof of the corollary can be found in Appendix \ref{sec:Proof-of-Corollary}.
\begin{cor}
\label{cor: asymptotic_distribution}Under Assumption II, as $n\rightarrow\infty$,{\small{}
\[
\left(\Sigma_{n}\right)^{-1/2}\sqrt{n}\left(\hat{\pi}\left(s\right)-\pi\left(s\right)\right)\overset{d}{\to}N\left(0,1\right).
\]
}Define $\hat{\Sigma}_{n}=\ensuremath{\frac{1}{\bar{c}}}\ensuremath{\sum_{c\in\left[\bar{c}\right]}}\left(n_{c}^{2}\frac{\bar{c}}{n}\right)\left(\bar{\hat{\zeta}}_{c}-\bar{\hat{\zeta}}\right)^{2}$,
where $\bar{\hat{\zeta}}=\frac{1}{n}\sum_{i=1}^{n}\hat{\zeta}_{i}$
and $\bar{\hat{\zeta}}_{c}=\frac{1}{n_{c}}\sum_{j\in\left[n_{c}\right]}\hat{\zeta}_{cj}$.
Under Assumptions II and III, as $n\rightarrow\infty$,{\small{}
\[
\left(\hat{\Sigma}_{n}\right)^{-1/2}\sqrt{n}\left(\hat{\pi}\left(s\right)-\pi\left(s\right)\right)\overset{d}{\to}N\left(0,1\right).
\]
}{\small\par}
\end{cor}

Based on the result above, we could construct confidence intervals
and conduct hypothesis testing as usual.

\section{Monte Carlo simulations \label{sec:Monte-Carlo}}

In this section, we examine the finite sample performance of the average
treatment effect ($\tau$) estimator using simulations. The exercise
on $\pi\left(s\right)$, the average effect of a counterfactual policy,
is very similar and hence omitted in the presentation. In particular,
we generate network data, and simulate treatment and outcome variables
on the network data. Then, we estimate the GNN models, construct the
ATE estimator, and evaluate the coverage of the 95\% confidence interval
of the ATE estimator. The simulation setup is designed to reflect
the conditions of our empirical application.

Specifically, we simulate $\bar{v}=50$ separated networks that we
refer to as `villages', with each village containing $n_{v}=200$
observations. For any given node $i$ in a village, the probability
of forming a directed edge from another node $j$ within the same
village to node $i$ is denoted as $p_{e}$ (i.e., $\text{Pr}\left(d_{ij}=1\right)=p_{e}$).
This probability is set to maintain an average of $5$ adjacent neighbors
per node, computed as $p_{e}\times(n_{v}-1)=5$. After the initial
edge formation, we impose a limit, $c_{n}=10$, on the maximum number
of adjacent neighbors a node can have. If a node has more than $c_{n}$
incoming edges (i.e., $\left|\mathcal{N}\left(i\right)\right|\geq c_{n}$),
only $c_{n}$ of them will be retained at random.\footnote{We have examined other settings, including those where each node has
on average $10$ adjacent neighbors but no more than $40$, which
provide similar results and hence are omitted in the presentation.} No edge can be formed across villages. In addition, the covariates,
$\boldsymbol{x}_{i}$, are generated i.i.d. uniformly from the $d$-dimensional
cube $\left[-1,1\right]^{d}$, where $d$, the number of covariates,
is set to $4$. As before, we use $\boldsymbol{\xi}_{i}$ to denote
the $L$-hop local neighborhood around node $i$.

We then construct a treatment model and an outcome model. In our empirical
application, the treatment decision pertains to whether a household
is informed by the experimenter about the microfinance loan, and the
outcome decision corresponds to whether a household chooses to participate
in the microfinance program if being informed. To emulate the conditions
of our empirical application, both the treatment and outcome decisions
are binary variables in our simulations.

For the outcome variable, we consider the setting $L=2$, indicating
that an observation's decision depends on the covariates of their
neighbors and neighbors' neighbors. In particular, we have{\small{}
\begin{align}
\boldsymbol{h}_{*i} & =\boldsymbol{f}_{*}^{(1)}\left(\boldsymbol{x}_{i},\overline{\boldsymbol{x}}_{*\mathcal{N}(i)}\right)\in\mathbb{R}^{d_{h*}},\label{eq: monte-carlo-model1}\\
z_{*i}\left(\boldsymbol{f}_{*}\right) & =f_{*}^{(2)}\left(\boldsymbol{h}_{*i},\overline{\boldsymbol{h}}_{*\mathcal{N}(i)}\right)\in\mathbb{R},\label{eq: monte-carlo-model2}
\end{align}
}where we set $d_{h*}=d=4$, $\overline{\boldsymbol{x}}_{*\mathcal{N}\left(i\right)}=\frac{1}{\left|\mathcal{N}\left(i\right)\right|}\sum_{j\in\mathcal{N}\left(i\right)}\boldsymbol{x}_{j}$,
and $\overline{\boldsymbol{h}}_{*\mathcal{N}\left(i\right)}=\frac{1}{\left|\mathcal{N}\left(i\right)\right|}\sum_{j\in\mathcal{N}\left(i\right)}\boldsymbol{h}_{*j}$
unless $\mathcal{N}\left(i\right)=\textrm{Ø}$, otherwise we set $\overline{\boldsymbol{x}}_{*\mathcal{N}\left(i\right)}=\overline{\boldsymbol{h}}_{*\mathcal{N}\left(i\right)}=\mathbf{0}$.
We use quadratic models to construct $\boldsymbol{f}_{*}^{(1)}$ and
$f_{*}^{(2)}$. Hence, for $l\in\left\{ 1,2\right\} $, we express{\small{}
\begin{equation}
\boldsymbol{f}_{*}^{(l)}\left(\boldsymbol{x},\boldsymbol{y}\right)=\boldsymbol{A}^{(l)}\boldsymbol{x}+\boldsymbol{B}^{(l)}\boldsymbol{y}+\boldsymbol{C}^{(l)}\tilde{\boldsymbol{x}}+\boldsymbol{D}^{(l)}\tilde{\boldsymbol{y}}+\boldsymbol{c}^{(l)},\label{eq: monte-carlo-model3}
\end{equation}
}with $\tilde{\boldsymbol{x}}\in\mathbb{R}^{d^{2}}$ and $\tilde{\boldsymbol{y}}\in\mathbb{R}^{d^{2}}$being
the second order interactions of $\boldsymbol{x}$ and $\boldsymbol{y}$,
respectively. The outcome model depends on the parameters for the
linear terms $\boldsymbol{A}^{(1)},\boldsymbol{B}^{(1)}\in\mathbb{R}^{4\times4}$
and $\boldsymbol{A}^{(2)},\boldsymbol{B}^{(2)}\in\mathbb{R}^{1\times4}$,
for which we randomly draw each element from the uniform distribution
$\mathcal{U}\left[0.3,0.7\right]$. Also, we randomly draw each element
from $\mathcal{U}\left[-0.1,0.1\right]$ for the quadratic-term parameters
$\boldsymbol{C}^{(1)},\boldsymbol{D}^{(1)}\in\mathbb{R}^{4\times16}$
and $\boldsymbol{C}^{(2)},\boldsymbol{D}^{(2)}\in\mathbb{R}^{1\times16}$,
and from $\mathcal{U}\left[0.3,0.7\right]$ for the constant-term
parameters $\boldsymbol{c}^{(1)}\in\mathbb{R}^{4}$, and we set $c^{(2)}=-0.1$.
Then, the outcome decision if being treated is generated as $y_{i}\left(1\right)=1\left[z_{*i}\left(\boldsymbol{f}_{*}\right)>\epsilon_{1,i}\right]$,
where $\epsilon_{1,i}\sim Logistic\left(0,1\right)$. The outcome
decision if not being treated is set as $y_{i}\left(0\right)=0$,
which aligns with our empirical application where a household cannot
take the microfinance loan if not being informed. 

For the treatment model, we consider two scenarios. The first scenario
is about random treatment with a constant treatment probability, where
the propensity score is set to be $\Pr\left(t_{i}=1\mid\boldsymbol{\xi}_{i}\right)=0.5$,
and the treatment decision is generated as $t_{i}=1\left[\mathcal{U}\left[0,1\right]>0.5\right]$.
In the second scenario, the treatment decision depends on neighbors
up to distance $L=2$, using the same models as for the outcome decision
(as shown in (\ref{eq: monte-carlo-model1})-(\ref{eq: monte-carlo-model3})),
except that the parameters are different. Specifically, we randomly
draw each element of the parameters for linear components, $\boldsymbol{A}^{(l)},\boldsymbol{B}^{(l)}$
for $l\in\left\{ 1,2\right\} $, from $\mathcal{U}\left[0.1,0.2\right]$,
for quadratic components, $\boldsymbol{C}^{(l)},\boldsymbol{D}^{(l)}$
for $l\in\left\{ 1,2\right\} $, from $\mathcal{U}\left[-0.1,0.1\right]$,
and for constant terms $\boldsymbol{c}^{(1)}$ from $\mathcal{U}\left[0.1,0.2\right]$,
and we set $c^{(2)}=-0.1$. Then, the treatment decision is determined
as $t_{i}=1\left[z_{*i}\left(\boldsymbol{f}_{*}^{p}\right)>\epsilon_{2,i}\right],$
where $\epsilon_{2,i}\sim Logistic\left(0,1\right)$. After generating
the parameters for the treatment and outcome models, these parameters
are held fixed across all simulation replications, while the networks,
covariates, outcomes, and treatment decisions are regenerated for
each simulation replication.

In the simulations, our parameter of interest is the average treatment
effect \\
$\tau=\mathbb{E}\left[y_{i}\left(1\right)-y_{i}\left(0\right)\right]=\mathbb{E}\left[y_{i}\left(1\right)\right]$.
With the simulated data, we construct the ATE estimator as follows.
For the outcome model, we estimate the model with a two-layer GNN
using treated individuals only, which provides $\hat{\mu}_{1}\left(\boldsymbol{\xi}_{i}\right)$.
And for the treatment model, we use the sample mean of $t_{i}$ for
the random scenario, and a two-layer GNN with all individuals for
the non-random scenario, which gives $\hat{p}_{1}\left(\boldsymbol{\xi}_{i}\right)$.
Then, we compute the ATE estimator as $\hat{\tau}_{n}=\frac{1}{n}\sum_{i\in\left[n\right]}\hat{\zeta}_{i}$,
where $\hat{\zeta}_{i}=\frac{\mathbb{I}\left\{ t_{i}=1\right\} }{\hat{p}_{1}\left(\boldsymbol{\xi}_{i}\right)}\left(y_{i}-\hat{\mu}_{1}\left(\boldsymbol{\xi}_{i}\right)\right)+\hat{\mu}_{1}\left(\boldsymbol{\xi}_{i}\right)$,
and the estimator of its variance as $\hat{\Sigma}_{n}=\frac{1}{\bar{v}}\ensuremath{\sum_{v\in\left[\bar{v}\right]}}n_{v}(\bar{\hat{\zeta}}_{v}-\bar{\hat{\zeta}})^{2}$,
with $\bar{\hat{\zeta}}_{v}$ being the sample mean of $\hat{\zeta}_{i}$
in each village $v$ and $\bar{\hat{\zeta}}$ being the sample mean
of $\hat{\zeta}_{i}$ over the full sample.

In terms of the GNN estimation for the outcome and treatment models,
the number of hidden neurons for each layer is set to be 8, 16, or
32, which results in nine different architectures for a two-layer
GNN estimator. We set the learning rate to be 0.001 and batch size
5. To keep the presentation succinct, we only present results using
the same GNN architecture for both the outcome and treatment models.
However, we have also experimented with various combinations of architectures,
different learning rates (such as 0.005 and 0.01), and different batch
sizes (such as 10). These additional results yield very similar findings
and hence are omitted from the presentation.

Based on Corollary \ref{cor: asymptotic_distribution}, we could construct
the 95\% confidence interval for each simulated sample (i.e., $\hat{\tau}\pm\frac{1.96}{\sqrt{n}}(\hat{\Sigma}_{n})^{1/2}$)
and examine if $\tau$ falls within it. We repeat this process for
1000 replications, and report the biases and coverage probabilities
of the estimator in Table \ref{tab:simulation_bias_coverage}. Our
findings indicate that the ATE estimator exhibits very small biases,
and the coverage probabilities are close to the nominal level, for
both random and non-random treatment scenarios and across all GNN
architectures. The simulation results suggest that the asymptotic
distribution in Corollary \ref{cor: asymptotic_distribution} is a
reasonable approximation to the finite sample dispersion of the estimator.

{\small{}}
\begin{table}[h]
\begin{singlespace}
\begin{centering}
{\footnotesize{}}%
\begin{tabular}{cccccc}
\toprule 
 & \multicolumn{2}{c}{{\footnotesize{}Random Treatment}} &  & \multicolumn{2}{c}{{\footnotesize{}GNN Treatment}}\tabularnewline
\cmidrule{2-3} \cmidrule{3-3} \cmidrule{5-6} \cmidrule{6-6} 
{\footnotesize{}Architecture ($d_{h}^{\left(1\right)},d_{h}^{\left(2\right)}$)} & {\footnotesize{}Bias} & {\footnotesize{}Coverage} &  & {\footnotesize{}Bias} & {\footnotesize{}Coverage}\tabularnewline
\midrule
\midrule 
{\footnotesize{}(8, 8)} & {\footnotesize{}0.000108} & {\footnotesize{}0.946} &  & {\footnotesize{}0.000075} & {\footnotesize{}0.940}\tabularnewline
\midrule 
{\footnotesize{}(8, 16)} & {\footnotesize{}0.000116} & {\footnotesize{}0.943} &  & {\footnotesize{}0.000089} & {\footnotesize{}0.943}\tabularnewline
\midrule 
{\footnotesize{}(8, 32)} & {\footnotesize{}0.000106} & {\footnotesize{}0.944} &  & {\footnotesize{}0.000111} & {\footnotesize{}0.942}\tabularnewline
\midrule 
{\footnotesize{}(16, 8)} & {\footnotesize{}0.000106} & {\footnotesize{}0.944} &  & {\footnotesize{}0.000129} & {\footnotesize{}0.945}\tabularnewline
\midrule 
{\footnotesize{}(16, 16)} & {\footnotesize{}0.000106} & {\footnotesize{}0.946} &  & {\footnotesize{}0.000070} & {\footnotesize{}0.941}\tabularnewline
\midrule 
{\footnotesize{}(16, 32)} & {\footnotesize{}0.000119} & {\footnotesize{}0.946} &  & {\footnotesize{}0.000106} & {\footnotesize{}0.942}\tabularnewline
\midrule 
{\footnotesize{}(32, 8)} & {\footnotesize{}0.000104} & {\footnotesize{}0.942} &  & {\footnotesize{}0.000064} & {\footnotesize{}0.942}\tabularnewline
\midrule 
{\footnotesize{}(32, 16)} & {\footnotesize{}0.000113} & {\footnotesize{}0.945} &  & {\footnotesize{}0.000092} & {\footnotesize{}0.942}\tabularnewline
\midrule 
{\footnotesize{}(32, 32)} & {\footnotesize{}0.000118} & {\footnotesize{}0.943} &  & {\footnotesize{}0.000043} & {\footnotesize{}0.943}\tabularnewline
\bottomrule
\end{tabular}{\footnotesize\par}
\par\end{centering}
\end{singlespace}
\begin{centering}
{\small{}\caption{Biases and coverage probabilities for the ATE estimator\label{tab:simulation_bias_coverage}}
\medskip{}
}{\small\par}
\par\end{centering}
\raggedright{}{\small{}}%
\noindent\begin{minipage}[t]{1\columnwidth}%
\begin{singlespace}
\begin{flushleft}
{\footnotesize{}Note: This table presents the biases and coverage
probabilities of the ATE estimator, based on $1000$ replications.
For random treatment, the treatment probability is constant and we
estimate it using the sample counterpart. For GNN treatment, the treatment
assignment depends on neighbors up to two hops away, and we estimate
the treatment model using a two-layer GNN model. The outcome model
is always estimated using a two-layer GNN model. In the GNN estimation,
the number of hidden neurons for each layer is configured to be $8$,
$16$, or $32$, resulting in a total of nine distinct architectures.
The learning rate and batch size are set at $0.001$ and $5$, respectively.
We use the same GNN architecture for both the outcome and treatment
models. The results obtained using other learning rates, batch sizes,
and combinations of architectures are similar and hence omitted from
the table. The simulation results suggest that the ATE estimator has
small biases, and the coverage probabilities are close to the nominal
level.}
\par\end{flushleft}
\end{singlespace}
\end{minipage}{\small\par}
\end{table}
{\small\par}

\section{Empirical application \label{sec:Empirical-Application}}

The estimation of heterogeneous treatment effects is of essential
policy relevance. In many empirical contexts, an individual's response
to treatment depends not only on their own characteristics but also
on their social ties and the characteristics of their neighbors. In
this section, we study the empirical context of seeking external financing,
which is an important life decision that significantly affects an
individual's well-being. In particular, we investigate whether people
choose to borrow through a microfinance program when they are informed
about this opportunity. From the perspective of policymakers, understanding
the probabilities of individual borrowing based on neighborhood information
can assist in determining to whom the relevant information should
be sent when network data is available.

We utilize the dataset from the paper titled \textquotedbl The Diffusion
of Microfinance\textquotedbl{} by \citet{BanerjeeChandrasekharDufloJackson2013}.
The researchers in this study observe households in 49 isolated Indian
villages, with an average village size of 216.69 households. The dataset
includes information on household demographics, social networks, and
household microfinance participation decisions. In their study, information
on microfinance is initially shared with a subset of households in
each village, referred to as `leaders'. In total, there are 1262 leaders,
with an average of 25.75 leaders per village. The researchers then
observe the borrowing decisions made by these leaders after receiving
the information. We utilize this dataset to study the average treatment
effect of providing loan information, as well as the average outcome
of counterfactual policies. 

In the realm of microfinance, leaders' borrowing decisions can be
influenced not only by their own economic conditions but also by the
availability of alternative borrowing options, such as borrowing from
their neighbors. Meanwhile, the lending decisions of these neighbors
may hinge not only on their own economic conditions, but also on the
conditions of others they can in turn borrow from (i.e., the conditions
of leaders' neighbors' neighbors). Hence, when modeling leaders' \textit{borrowing}
decisions, we assume that these decisions depend on social ties extending
two hops away, and we adopt a two-layer GNN model for estimation.
In this estimation, interference is not a primary concern, as leaders
are the first to be informed about the loan when others do not have
access to such information (this aligns with the model in \citet{BanerjeeChandrasekharDufloJackson2013},
as we discussed earlier). 

For the \textsl{treatment} decisions of the experimenter on disseminating
microfinance information, we consider two possibilities: a random
assignment with a constant treatment probability and a decision that
depends on neighbors up to two hops away, where for the latter we
again use a two-layer GNN model for estimation. The authors are working
on a subsequent paper that aims to formally determine the relevant
neighborhood size for GNN estimations.

In our empirical estimation, we employ the `union relation' from the
dataset to represent network edge, which encompasses any social interaction
between two households, including activities like visiting each other's
homes, going to a temple together, among others. Additionally, we
follow the study by \citet{BanerjeeChandrasekharDufloJackson2013}
and use the following four variables as covariates: number of rooms,
number of beds, whether electricity is privately or publicly supplied,
and the presence or absence of a latrine. We also include closeness
centrality and betweenness centrality as covariates. Our estimation
results are presented in the following subsection.

\subsection{Inference of treatment effects}

In Table \ref{tab:Inference-results}, we present the estimation of
the average treatment effect along with the 95\% confidence intervals.
As can be seen from the table, the average borrowing probability if
being informed about the microfinance program is approximately 24\%
to 25\%, with a relatively tight confidence interval. The ATE estimates
are similar across GNN architectures, as well as between the random
treatment estimation and the GNN estimation for the treatment decision.

Beyond the ATEs, policymakers may also be interested in understanding
the average outcomes of disseminating information to villagers with
different network statistics. For example, they may ask: what would
the average participation rate be if information is sent to those
whose network centrality measure ranks in the top 25th percentile,
versus those in the bottom 25th percentile? Specifically, by including
household $i$'s centrality measure in $\boldsymbol{\xi}_{i}$, policymakers
may compare $\pi\left(s\right)$ where $s\left(\boldsymbol{\xi}_{i}\right)=1$
if $i$'s centrality is above the 75th percentile, with $\pi\left(s\right)$
where $s\left(\boldsymbol{\xi}_{i}\right)=1$ if $i$'s centrality
is below the 25th percentile. We focus on two specific centrality
measures: closeness centrality and betweenness centrality. Closeness
centrality measures the reciprocal of the average of the shortest
path distances from one node to all other nodes within the same village,
effectively reflecting how quickly a node can reach others. Betweenness
centrality, on the other hand, quantifies the number of shortest paths
between all pairs of nodes within the same village that pass through
a specific node, indicating its role as a `bridge' in the network.
Both of these measures are valuable in assessing a node's importance
in terms of information dissemination within network structures.

As shown in Table \ref{tab:Inference-results}, households with higher
centrality tend to exhibit a lower borrowing probability compared
to those with lower centrality. The average borrowing probability
is approximately $28.0\%$ ($0.070\times4$) to $31.6\%$ ($0.079\times4$)
among households within the lowest 25th percentile of closeness centrality,
compared to $19.6\%$ ($0.049\times4$) to $20.0\%$ ($0.050\times4$)
among households within the top 25th percentile of the same measure.
Similar patterns are found for betweenness centrality. We have also
experimented with degree centrality and eigenvector centrality, both
of which produced similar findings and are hence omitted for brevity.
A plausible explanation for these findings could be that households
with higher centrality tend to have more avenues for borrowing within
their social networks, thereby reducing their need to participate
in the microfinance program.

According to the study by \citet{BanerjeeChandrasekharDufloJackson2013},
a selected group of households (i.e., leaders) is informed about the
microfinance program and subsequently decides whether to participate.
These leaders then have a certain probability of informing their acquaintances
about the program, thereby facilitating the spread of knowledge about
the microfinance program. Importantly, the likelihood of leaders disseminating
this information depends on their own participation decisions. As
demonstrated empirically in Table 1 of the paper by \citet{BanerjeeChandrasekharDufloJackson2013},
the probability of a \textit{participating} leader disseminating the
information is 7-10 times higher than that of a \textit{non-participating}
leader.

Given budget constraints, policymakers often face the challenge of
deciding who should initially receive the program information, as
it would be prohibitively costly to inform everyone. The aim is to
distribute this information as quickly and as widely as possible through
social networks to enhance overall program participation. However,
solely targeting households based on their centrality measure might
not be the most effective approach, as the participation probability
of the targeted households also plays a significant role in information
diffusion. This consideration prompts our exploration of more balanced
targeting strategies in the following subsection, aiming to identify
and prioritize individuals who exhibit not only high centrality but
also a strong willingness to participate. 
\begin{center}
\begin{sidewaystable}
\begin{centering}
{\scriptsize{}}%
\begin{tabular}{ccccccccccccccccc}
\toprule 
 &  &  & \multicolumn{2}{c}{{\scriptsize{}Average Treatment Effect}} &  & \multicolumn{2}{c}{{\scriptsize{}Closeness Low}} &  & \multicolumn{2}{c}{{\scriptsize{}Closeness High}} &  & \multicolumn{2}{c}{{\scriptsize{}Betweenness Low}} &  & \multicolumn{2}{c}{{\scriptsize{}Betweenness High}}\tabularnewline
\cmidrule{4-5} \cmidrule{5-5} \cmidrule{7-8} \cmidrule{8-8} \cmidrule{10-11} \cmidrule{11-11} \cmidrule{13-14} \cmidrule{14-14} \cmidrule{16-17} \cmidrule{17-17} 
 & {\scriptsize{}($d_{h}^{(1)}$,$d_{h}^{(2)}$)} &  & {\scriptsize{}$\hat{\tau}$} & {\scriptsize{}95\% CI} &  & {\scriptsize{}$\hat{\pi}\left(s\right)$} & {\scriptsize{}95\% CI} &  & {\scriptsize{}$\hat{\pi}\left(s\right)$} & {\scriptsize{}95\% CI} &  & {\scriptsize{}$\hat{\pi}\left(s\right)$} & {\scriptsize{}95\% CI} &  & {\scriptsize{}$\hat{\pi}\left(s\right)$} & {\scriptsize{}95\% CI}\tabularnewline
\midrule
\midrule 
\multirow{9}{*}{{\scriptsize{}Random Treatment }} & {\scriptsize{}(6, 6)} &  & {\scriptsize{}0.239} & {\scriptsize{}{[}0.206, 0.273{]}} &  & {\scriptsize{}0.070} & {\scriptsize{}{[}0.054, 0.086{]}} &  & {\scriptsize{}0.048} & {\scriptsize{}{[}0.027,0.069{]}} &  & {\scriptsize{}0.066} & {\scriptsize{}{[}0.058, 0.074{]}} &  & {\scriptsize{}0.051} & {\scriptsize{}{[}0.036, 0.066{]}}\tabularnewline
\cmidrule{2-17} \cmidrule{3-17} \cmidrule{4-17} \cmidrule{5-17} \cmidrule{6-17} \cmidrule{7-17} \cmidrule{8-17} \cmidrule{9-17} \cmidrule{10-17} \cmidrule{11-17} \cmidrule{12-17} \cmidrule{13-17} \cmidrule{14-17} \cmidrule{15-17} \cmidrule{16-17} \cmidrule{17-17} 
 & {\scriptsize{}(6, 12)} &  & {\scriptsize{}0.238} & {\scriptsize{}{[}0.204, 0.271{]}} &  & {\scriptsize{}0.070} & {\scriptsize{}{[}0.054, 0.086{]}} &  & {\scriptsize{}0.048} & {\scriptsize{}{[}0.027,0.069{]}} &  & {\scriptsize{}0.066} & {\scriptsize{}{[}0.057, 0.074{]}} &  & {\scriptsize{}0.051} & {\scriptsize{}{[}0.036, 0.066{]}}\tabularnewline
\cmidrule{2-17} \cmidrule{3-17} \cmidrule{4-17} \cmidrule{5-17} \cmidrule{6-17} \cmidrule{7-17} \cmidrule{8-17} \cmidrule{9-17} \cmidrule{10-17} \cmidrule{11-17} \cmidrule{12-17} \cmidrule{13-17} \cmidrule{14-17} \cmidrule{15-17} \cmidrule{16-17} \cmidrule{17-17} 
 & {\scriptsize{}(6, 24)} &  & {\scriptsize{}0.242} & {\scriptsize{}{[}0.209, 0.276{]}} &  & {\scriptsize{}0.071} & {\scriptsize{}{[}0.055, 0.087{]}} &  & {\scriptsize{}0.049} & {\scriptsize{}{[}0.028,0.070{]}} &  & {\scriptsize{}0.067} & {\scriptsize{}{[}0.059, 0.075{]}} &  & {\scriptsize{}0.052} & {\scriptsize{}{[}0.037, 0.068{]}}\tabularnewline
\cmidrule{2-17} \cmidrule{3-17} \cmidrule{4-17} \cmidrule{5-17} \cmidrule{6-17} \cmidrule{7-17} \cmidrule{8-17} \cmidrule{9-17} \cmidrule{10-17} \cmidrule{11-17} \cmidrule{12-17} \cmidrule{13-17} \cmidrule{14-17} \cmidrule{15-17} \cmidrule{16-17} \cmidrule{17-17} 
 & {\scriptsize{}(12, 6)} &  & {\scriptsize{}0.244} & {\scriptsize{}{[}0.210, 0.277{]}} &  & {\scriptsize{}0.071} & {\scriptsize{}{[}0.055, 0.088{]}} &  & {\scriptsize{}0.050} & {\scriptsize{}{[}0.029,0.071{]}} &  & {\scriptsize{}0.067} & {\scriptsize{}{[}0.059, 0.076{]}} &  & {\scriptsize{}0.052} & {\scriptsize{}{[}0.037, 0.068{]}}\tabularnewline
\cmidrule{2-17} \cmidrule{3-17} \cmidrule{4-17} \cmidrule{5-17} \cmidrule{6-17} \cmidrule{7-17} \cmidrule{8-17} \cmidrule{9-17} \cmidrule{10-17} \cmidrule{11-17} \cmidrule{12-17} \cmidrule{13-17} \cmidrule{14-17} \cmidrule{15-17} \cmidrule{16-17} \cmidrule{17-17} 
 & {\scriptsize{}(12, 12)} &  & {\scriptsize{}0.242} & {\scriptsize{}{[}0.209, 0.276{]}} &  & {\scriptsize{}0.071} & {\scriptsize{}{[}0.055, 0.087{]}} &  & {\scriptsize{}0.049} & {\scriptsize{}{[}0.028,0.070{]}} &  & {\scriptsize{}0.067} & {\scriptsize{}{[}0.059, 0.076{]}} &  & {\scriptsize{}0.052} & {\scriptsize{}{[}0.036, 0.067{]}}\tabularnewline
\cmidrule{2-17} \cmidrule{3-17} \cmidrule{4-17} \cmidrule{5-17} \cmidrule{6-17} \cmidrule{7-17} \cmidrule{8-17} \cmidrule{9-17} \cmidrule{10-17} \cmidrule{11-17} \cmidrule{12-17} \cmidrule{13-17} \cmidrule{14-17} \cmidrule{15-17} \cmidrule{16-17} \cmidrule{17-17} 
 & {\scriptsize{}(12, 24)} &  & {\scriptsize{}0.243} & {\scriptsize{}{[}0.209, 0.276{]}} &  & {\scriptsize{}0.071} & {\scriptsize{}{[}0.055, 0.087{]}} &  & {\scriptsize{}0.050} & {\scriptsize{}{[}0.029,0.071{]}} &  & {\scriptsize{}0.067} & {\scriptsize{}{[}0.058, 0.075{]}} &  & {\scriptsize{}0.052} & {\scriptsize{}{[}0.037, 0.068{]}}\tabularnewline
\cmidrule{2-17} \cmidrule{3-17} \cmidrule{4-17} \cmidrule{5-17} \cmidrule{6-17} \cmidrule{7-17} \cmidrule{8-17} \cmidrule{9-17} \cmidrule{10-17} \cmidrule{11-17} \cmidrule{12-17} \cmidrule{13-17} \cmidrule{14-17} \cmidrule{15-17} \cmidrule{16-17} \cmidrule{17-17} 
 & {\scriptsize{}(24, 6)} &  & {\scriptsize{}0.242} & {\scriptsize{}{[}0.208, 0.275{]}} &  & {\scriptsize{}0.072} & {\scriptsize{}{[}0.056, 0.088{]}} &  & {\scriptsize{}0.049} & {\scriptsize{}{[}0.028,0.070{]}} &  & {\scriptsize{}0.068} & {\scriptsize{}{[}0.059, 0.076{]}} &  & {\scriptsize{}0.052} & {\scriptsize{}{[}0.036, 0.067{]}}\tabularnewline
\cmidrule{2-17} \cmidrule{3-17} \cmidrule{4-17} \cmidrule{5-17} \cmidrule{6-17} \cmidrule{7-17} \cmidrule{8-17} \cmidrule{9-17} \cmidrule{10-17} \cmidrule{11-17} \cmidrule{12-17} \cmidrule{13-17} \cmidrule{14-17} \cmidrule{15-17} \cmidrule{16-17} \cmidrule{17-17} 
 & {\scriptsize{}(24, 12)} &  & {\scriptsize{}0.242} & {\scriptsize{}{[}0.209, 0.276{]}} &  & {\scriptsize{}0.071} & {\scriptsize{}{[}0.055, 0.088{]}} &  & {\scriptsize{}0.049} & {\scriptsize{}{[}0.028,0.070{]}} &  & {\scriptsize{}0.067} & {\scriptsize{}{[}0.059, 0.076{]}} &  & {\scriptsize{}0.052} & {\scriptsize{}{[}0.037, 0.067{]}}\tabularnewline
\cmidrule{2-17} \cmidrule{3-17} \cmidrule{4-17} \cmidrule{5-17} \cmidrule{6-17} \cmidrule{7-17} \cmidrule{8-17} \cmidrule{9-17} \cmidrule{10-17} \cmidrule{11-17} \cmidrule{12-17} \cmidrule{13-17} \cmidrule{14-17} \cmidrule{15-17} \cmidrule{16-17} \cmidrule{17-17} 
 & {\scriptsize{}(24, 24)} &  & {\scriptsize{}0.242} & {\scriptsize{}{[}0.208, 0.276{]}} &  & {\scriptsize{}0.072} & {\scriptsize{}{[}0.055, 0.088{]}} &  & {\scriptsize{}0.049} & {\scriptsize{}{[}0.028,0.070{]}} &  & {\scriptsize{}0.067} & {\scriptsize{}{[}0.059, 0.076{]}} &  & {\scriptsize{}0.052} & {\scriptsize{}{[}0.037, 0.067{]}}\tabularnewline
\midrule 
\multirow{9}{*}{{\scriptsize{}GNN Treatment }} & {\scriptsize{}(6, 6)} &  & {\scriptsize{}0.253} & {\scriptsize{}{[}0.217, 0.289{]}} &  & {\scriptsize{}0.079} & {\scriptsize{}{[}0.056, 0.102{]}} &  & {\scriptsize{}0.050} & {\scriptsize{}{[}0.031, 0.069{]}} &  & {\scriptsize{}0.072} & {\scriptsize{}{[}0.056, 0.087{]}} &  & {\scriptsize{}0.053} & {\scriptsize{}{[}0.042, 0.065{]}}\tabularnewline
\cmidrule{2-17} \cmidrule{3-17} \cmidrule{4-17} \cmidrule{5-17} \cmidrule{6-17} \cmidrule{7-17} \cmidrule{8-17} \cmidrule{9-17} \cmidrule{10-17} \cmidrule{11-17} \cmidrule{12-17} \cmidrule{13-17} \cmidrule{14-17} \cmidrule{15-17} \cmidrule{16-17} \cmidrule{17-17} 
 & {\scriptsize{}(6, 12)} &  & {\scriptsize{}0.249} & {\scriptsize{}{[}0.213, 0.285{]}} &  & {\scriptsize{}0.077} & {\scriptsize{}{[}0.055, 0.100{]}} &  & {\scriptsize{}0.050} & {\scriptsize{}{[}0.031, 0.069{]}} &  & {\scriptsize{}0.069} & {\scriptsize{}{[}0.055, 0.084{]}} &  & {\scriptsize{}0.054} & {\scriptsize{}{[}0.042, 0.066{]}}\tabularnewline
\cmidrule{2-17} \cmidrule{3-17} \cmidrule{4-17} \cmidrule{5-17} \cmidrule{6-17} \cmidrule{7-17} \cmidrule{8-17} \cmidrule{9-17} \cmidrule{10-17} \cmidrule{11-17} \cmidrule{12-17} \cmidrule{13-17} \cmidrule{14-17} \cmidrule{15-17} \cmidrule{16-17} \cmidrule{17-17} 
 & {\scriptsize{}(6, 24)} &  & {\scriptsize{}0.253} & {\scriptsize{}{[}0.218, 0.289{]}} &  & {\scriptsize{}0.078} & {\scriptsize{}{[}0.056, 0.099{]}} &  & {\scriptsize{}0.050} & {\scriptsize{}{[}0.030, 0.069{]}} &  & {\scriptsize{}0.071} & {\scriptsize{}{[}0.057, 0.084{]}} &  & {\scriptsize{}0.054} & {\scriptsize{}{[}0.041, 0.067{]}}\tabularnewline
\cmidrule{2-17} \cmidrule{3-17} \cmidrule{4-17} \cmidrule{5-17} \cmidrule{6-17} \cmidrule{7-17} \cmidrule{8-17} \cmidrule{9-17} \cmidrule{10-17} \cmidrule{11-17} \cmidrule{12-17} \cmidrule{13-17} \cmidrule{14-17} \cmidrule{15-17} \cmidrule{16-17} \cmidrule{17-17} 
 & {\scriptsize{}(12, 6)} &  & {\scriptsize{}0.252} & {\scriptsize{}{[}0.215, 0.289{]}} &  & {\scriptsize{}0.077} & {\scriptsize{}{[}0.055, 0.100{]}} &  & {\scriptsize{}0.050} & {\scriptsize{}{[}0.030, 0.069{]}} &  & {\scriptsize{}0.072} & {\scriptsize{}{[}0.056, 0.087{]}} &  & {\scriptsize{}0.053} & {\scriptsize{}{[}0.041, 0.065{]}}\tabularnewline
\cmidrule{2-17} \cmidrule{3-17} \cmidrule{4-17} \cmidrule{5-17} \cmidrule{6-17} \cmidrule{7-17} \cmidrule{8-17} \cmidrule{9-17} \cmidrule{10-17} \cmidrule{11-17} \cmidrule{12-17} \cmidrule{13-17} \cmidrule{14-17} \cmidrule{15-17} \cmidrule{16-17} \cmidrule{17-17} 
 & {\scriptsize{}(12, 12)} &  & {\scriptsize{}0.249} & {\scriptsize{}{[}0.212, 0.286{]}} &  & {\scriptsize{}0.076} & {\scriptsize{}{[}0.054, 0.097{]}} &  & {\scriptsize{}0.050} & {\scriptsize{}{[}0.030, 0.070{]}} &  & {\scriptsize{}0.070} & {\scriptsize{}{[}0.055, 0.085{]}} &  & {\scriptsize{}0.053} & {\scriptsize{}{[}0.041, 0.065{]}}\tabularnewline
\cmidrule{2-17} \cmidrule{3-17} \cmidrule{4-17} \cmidrule{5-17} \cmidrule{6-17} \cmidrule{7-17} \cmidrule{8-17} \cmidrule{9-17} \cmidrule{10-17} \cmidrule{11-17} \cmidrule{12-17} \cmidrule{13-17} \cmidrule{14-17} \cmidrule{15-17} \cmidrule{16-17} \cmidrule{17-17} 
 & {\scriptsize{}(12, 24)} &  & {\scriptsize{}0.252} & {\scriptsize{}{[}0.216, 0.289{]}} &  & {\scriptsize{}0.079} & {\scriptsize{}{[}0.057, 0.101{]}} &  & {\scriptsize{}0.049} & {\scriptsize{}{[}0.030, 0.069{]}} &  & {\scriptsize{}0.072} & {\scriptsize{}{[}0.057, 0.088{]}} &  & {\scriptsize{}0.052} & {\scriptsize{}{[}0.040, 0.065{]}}\tabularnewline
\cmidrule{2-17} \cmidrule{3-17} \cmidrule{4-17} \cmidrule{5-17} \cmidrule{6-17} \cmidrule{7-17} \cmidrule{8-17} \cmidrule{9-17} \cmidrule{10-17} \cmidrule{11-17} \cmidrule{12-17} \cmidrule{13-17} \cmidrule{14-17} \cmidrule{15-17} \cmidrule{16-17} \cmidrule{17-17} 
 & {\scriptsize{}(24, 6)} &  & {\scriptsize{}0.251} & {\scriptsize{}{[}0.215, 0.286{]}} &  & {\scriptsize{}0.078} & {\scriptsize{}{[}0.057, 0.098{]}} &  & {\scriptsize{}0.050} & {\scriptsize{}{[}0.030, 0.069{]}} &  & {\scriptsize{}0.072} & {\scriptsize{}{[}0.058, 0.086{]}} &  & {\scriptsize{}0.053} & {\scriptsize{}{[}0.041, 0.066{]}}\tabularnewline
\cmidrule{2-17} \cmidrule{3-17} \cmidrule{4-17} \cmidrule{5-17} \cmidrule{6-17} \cmidrule{7-17} \cmidrule{8-17} \cmidrule{9-17} \cmidrule{10-17} \cmidrule{11-17} \cmidrule{12-17} \cmidrule{13-17} \cmidrule{14-17} \cmidrule{15-17} \cmidrule{16-17} \cmidrule{17-17} 
 & {\scriptsize{}(24, 12)} &  & {\scriptsize{}0.252} & {\scriptsize{}{[}0.215, 0.289{]}} &  & {\scriptsize{}0.078} & {\scriptsize{}{[}0.056, 0.101{]}} &  & {\scriptsize{}0.050} & {\scriptsize{}{[}0.030, 0.070{]}} &  & {\scriptsize{}0.072} & {\scriptsize{}{[}0.057, 0.087{]}} &  & {\scriptsize{}0.054} & {\scriptsize{}{[}0.041, 0.066{]}}\tabularnewline
\cmidrule{2-17} \cmidrule{3-17} \cmidrule{4-17} \cmidrule{5-17} \cmidrule{6-17} \cmidrule{7-17} \cmidrule{8-17} \cmidrule{9-17} \cmidrule{10-17} \cmidrule{11-17} \cmidrule{12-17} \cmidrule{13-17} \cmidrule{14-17} \cmidrule{15-17} \cmidrule{16-17} \cmidrule{17-17} 
 & {\scriptsize{}(24, 24)} &  & {\scriptsize{}0.253} & {\scriptsize{}{[}0.217, 0.289{]}} &  & {\scriptsize{}0.079} & {\scriptsize{}{[}0.057, 0.101{]}} &  & {\scriptsize{}0.050} & {\scriptsize{}{[}0.031, 0.069{]}} &  & {\scriptsize{}0.073} & {\scriptsize{}{[}0.058, 0.089{]}} &  & {\scriptsize{}0.054} & {\scriptsize{}{[}0.042, 0.066{]}}\tabularnewline
\bottomrule
\end{tabular}{\scriptsize\par}
\par\end{centering}
\caption{Inference on Average Treatment Effects and Average Outcomes of Counterfactual
Policies\label{tab:Inference-results}}
\medskip{}

\raggedright{}%
\noindent\begin{minipage}[t]{1\columnwidth}%
\begin{singlespace}
\begin{flushleft}
{\scriptsize{}Note: This table presents the estimations and confidence
intervals for the average treatment effects and outcomes of counterfactual
policies based on two centrality measures. For both closeness and
betweenness centrality, we analyze two counterfactual policies: one
assigning treatments to households within the top 25th percentile
of the centrality measure, and the other to those within the bottom
25th percentile. Leaders' participation decisions are estimated using
two-layer GNNs. For treatment decisions, we first assume a random
assignment with a constant treatment probability, for which the treatment
probability is estimated using the sample counterpart. We next allow
treatment decisions to depend on neighbors up to two hops away, for
which the treatment model is estimated using two-layer GNNs. The GNN
model is designed with three possible numbers of hidden neurons per
layer: $6$, $12$, and $24$, resulting in nine different architectures.
Data is divided into training and validation sets in an $80:20$ ratio,
preserving the network structure by splitting at the village level.
The Adam optimizer is used for model training with early stopping
based on validation loss. The results presented in this table are
obtained using a batch size of $5$ and a learning rate of $0.001$,
and the same architecture is used for the two GNN models, one for
treatment and one for participation. Our explorations with varying
batch sizes (such as $10$ or the entire data set) and learning rates
(such as $0.005$ or $0.01$), as well as different architecture combinations,
yield similar results. Hence, those specifics are omitted for brevity.}
\par\end{flushleft}
\end{singlespace}
\end{minipage}
\end{sidewaystable}
\par\end{center}

\subsection{Individual targeting on networks}

As discussed previously, the success of the program hinges on a judicious
selection of leaders. The process must balance the leaders' likelihood
of participating with the strength of their social ties. The program
is likely to flourish if the chosen leaders not only are inclined
to participate but also maintain strong connections with uninformed
households in their social networks. This is an essential crux of
network diffusion problems. To address this, we propose an improved
approach to leader selection that outperforms the current method.

We use the GNN estimation to predict the participation probabilities
of households based on their local neighborhoods if they were informed
about the program. We denote the predicted participation probability
for household $i$ as $\hat{p}_{i}$. To measure the strength of social
ties for each household $i$ in the network, denoted as $c_{i}$,
we use either closeness centrality or betweenness centrality separately.
We have also explored degree centrality and eigenvector centrality,
which yielded similar results and hence are omitted in the presentation.

To select a group of households that balance the participation probability
$\hat{p}_{i}$ and centrality $c_{i}$, we use a relative weight $\omega\in\left[0,1\right]$
to capture the balance between these two aspects, and we compute a
weighted score for each household as follows{\small{}
\[
s_{i}=\left(1-\omega\right)\hat{p}_{i}+\omega c_{i}.
\]
}As $\omega$ approaches 0, the score places more emphasis on participation
probability. Conversely, when $\omega$ approaches 1, the score assigns
greater emphasis to centrality. Fixing a relative weight $\omega$,
we can select the top $k$ households based on the weighted score.
This gives us a group of the $k$ best leaders that balance the leader
participation and centrality at this particular weight $\omega$.
For our exercises, we set $k$ to equal 1262, which is the actual
number of leaders selected by the researchers in \citet{BanerjeeChandrasekharDufloJackson2013}.

\begin{figure}[h]
\begin{centering}
\includegraphics[scale=0.5]{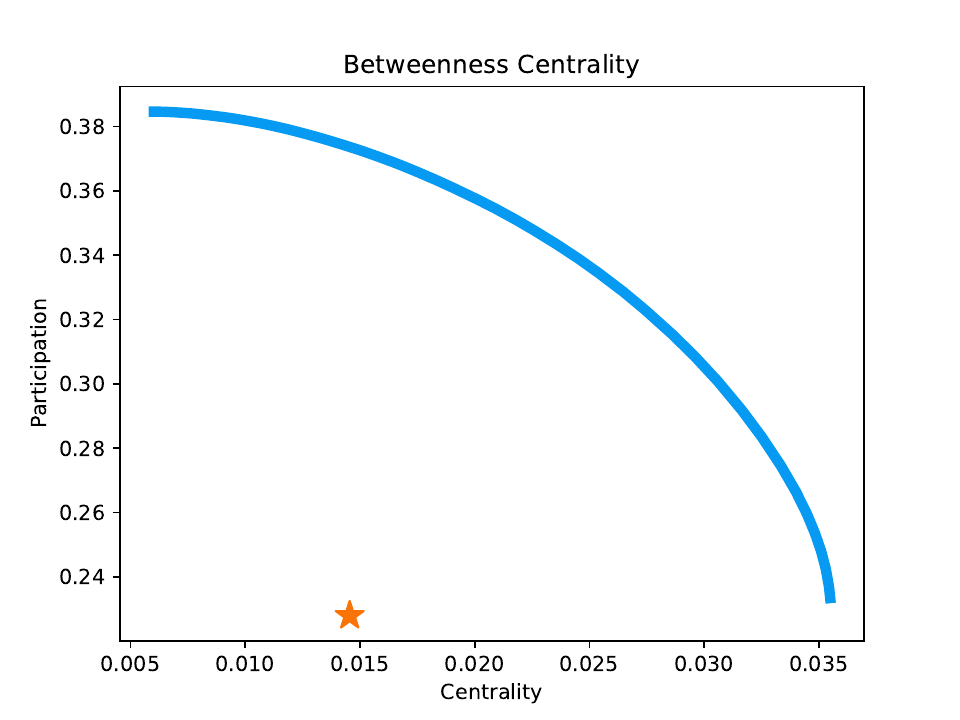}\includegraphics[scale=0.5]{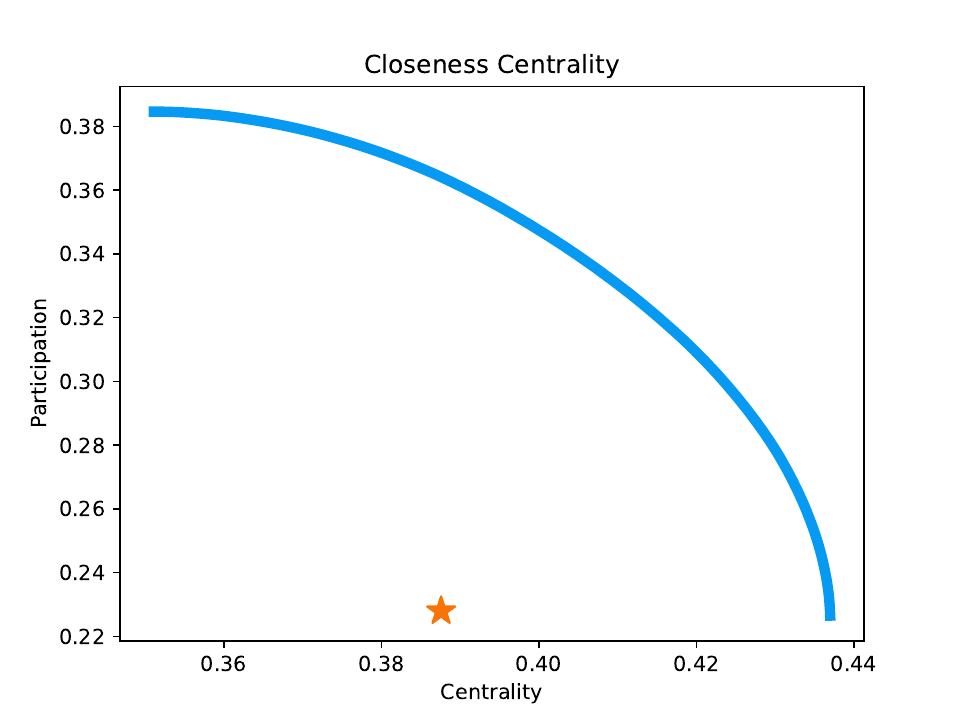}
\par\end{centering}
\caption{Pareto Frontiers between Participation Rates and Centrality Measures
in Leader Selection\label{fig:centrality-vs-participation}}
\medskip{}

\raggedright{}%
\noindent\begin{minipage}[t]{1\columnwidth}%
\begin{singlespace}
\begin{flushleft}
{\footnotesize{}Note: in the figure, the blue curves indicate the
Pareto frontiers between the average participation rates and average
centralities among selected leaders, while the orange star corresponds
to the selected leaders under the current strategy, which selects
teachers and shopkeepers, among others. We predict each household's
participation probability if being informed about the program using
the two-layer GNN estimates with the architecture of $(24,12)$-hidden
neurons. The figure suggests considerable potential for Pareto improvements
in leader selection, indicating options with both higher average participation
probabilities and greater centralities compared to the current leader
selection.}
\par\end{flushleft}
\end{singlespace}
\end{minipage}
\end{figure}

By varying $\omega$, we can find different groups of $k$ best leaders
and compute their average participation probabilities and centralities.
This enables us to trace out the Pareto frontier that the program
can achieve between leader participation rate and centrality, as indicated
by the blue curves in Figure \ref{fig:centrality-vs-participation}.
The Pareto frontier illustrates potential improvements compared to
the selected leaders under the current strategy, as marked by the
orange star in Figure \ref{fig:centrality-vs-participation}. Regarding
centrality and participation probability, it is possible to realize
improvements equivalent to 1.5 standard deviations in betweenness
centrality, 0.94 standard deviations in closeness centrality, or 1.83
standard deviations in leader participation probability alone. Also,
there is a large room for joint Pareto improvements in both participation
probability and network centrality, relative to the current leader
selection rule. Our findings present valuable insights on the enhancement
of information diffusion through social networks.

\section{Conclusion \label{sec:Conclusion}}

In this paper, we present a novel application of graph neural networks
for modeling and estimating network heterogeneity, focusing on the
empirical context where individual outcomes or decisions depend on
local neighborhood surroundings. We provide theoretical justifications
for the GNN estimator and demonstrate its usage in causal inference
with heterogeneous treatment effects. Empirically, in the context
of a microfinance program, we show that the estimator can be applied
to evaluate various treatment effects and to enhance information targeting
on networks. We believe the GNN estimator holds substantial value
in economics and social sciences, leveraging the potency of versatile
neural network architectures and abundant network data. Adapting the
GNN model to accommodate diverse empirical needs in different economic
settings opens up an exciting avenue for future research.

\appendix
\newpage{}

\section*{Appendix}

\section{Definitions and supporting lemmas \label{sec:Definitions-and-Supporting-Lemmas}}

\begin{defn}
\label{def: sub-root function}{[}Sub-root function{]} A function
$\psi:[0,\infty)\rightarrow[0,\infty)$ is sub-root if it is nonnegative,
nondecreasing, and if $r\mapsto\psi\left(r\right)/\sqrt{r}$ is nonincreasing
for $r>0$.
\end{defn}

\begin{lem}
\label{lem: Bartlett Lemma 3.2}{[}\citet[Lemma 3.2]{bartlett2005local}{]}
If $\psi:\left[0,\infty\right)\rightarrow\left[0,\infty\right)$ is
a nontrivial sub-root function (a nontrivial sub-root function is
not the constant function $\psi\coloneqq0$), then the equation $\psi\left(r\right)=r$
has a unique positive solution. We call the unique positive solution
of $\psi\left(r\right)=r$ the fixed point of $\psi$.
\end{lem}

\begin{defn}
\label{def: Rademacher variables}{[}Rademacher variables{]} $\eta_{1},...,\eta_{n}$
are $n$ independent Rademacher random variables if each $\eta_{i}$
is an i.i.d. draw taking on the value of $1$ or $-1$ with probability
$1/2$.
\end{defn}

\begin{lem}
\label{lem: Bartlett Theorem 3.3}{[}\citet[Theorem 3.3]{bartlett2005local}{]}
Let $\left(\mathcal{X},P\right)$ be a probability space. Let $\mathcal{G}$
be a class of measurable functions from $\mathcal{X}$ to $\left[a,b\right]$.
Let $\left\{ \boldsymbol{x}_{i}\right\} _{i\in\left[n\right]}$ be
independent random variables distributed according to $P$. Assume
that there are some functional $T:\mathcal{G}\rightarrow\mathbb{R}_{+}$
and some constant $B>0$ such that for every $g\in\mathcal{G}$, $\text{Var}\left(g\right)\leq T\left(g\right)\leq B\mathbb{E}\left[g\right]$.
Assume that $\psi$ satisfies, for every $r\geq r^{*}$,{\small{}
\[
\psi\left(r\right)\geq B\mathbb{E}\left[\sup_{g\in\mathcal{G}:T\left(g\right)\leq r}\frac{1}{n}\sum_{i=1}^{n}\eta_{i}g\left(\boldsymbol{x}_{i}\right)\right],
\]
}where $\left\{ \eta_{i}\right\} _{i\in\left[n\right]}$ are independent
Rademacher random variables, and the expectation is taken over the
randomness of $\left\{ \left(\eta_{i},\boldsymbol{x}_{i}\right)\right\} _{i\in\left[n\right]}.$

Then, for every $t>0$, with probability at least $1-\exp\left(-t\right)$,{\small{}
\[
\forall g\in\mathcal{G},\quad\mathbb{E}\left[g\left(\boldsymbol{x}_{i}\right)\right]\leq\max\left\{ \frac{1}{n}\sum_{i=1}^{n}g\left(\boldsymbol{x}_{i}\right),\frac{2}{n}\sum_{i=1}^{n}g\left(\boldsymbol{x}_{i}\right)\right\} +Cr^{*}+C\frac{t}{n},
\]
}and also with probability at least $1-\exp\left(-t\right)$,{\small{}
\[
\forall g\in\mathcal{G},\quad\frac{1}{n}\sum_{i=1}^{n}g\left(\boldsymbol{x}_{i}\right)\leq\frac{3}{2}\mathbb{E}\left[g\left(\boldsymbol{x}_{i}\right)\right]+Cr^{*}+C\frac{t}{n},
\]
}where $C$ is a fixed constant.
\end{lem}

\begin{lem}
\label{lem: contraction}{[}\citet[Lemma 5.7]{mohri2018foundations}{]}
Let $\mathcal{G}$ be a class of functions mapping $\mathcal{X}$
to $\mathbb{R}$, and let $\phi_{1},\ldots,\phi_{n}$ be $L$-Lipschitz
functions for some constant $L>0$. Then for any fixed sample of $n$
points $\boldsymbol{x}_{1},\ldots,\boldsymbol{x}_{n}\in\mathcal{X}$,
it holds{\small{}
\[
\frac{1}{n}\mathbb{E}_{\eta}\left[\sup_{g\in\mathcal{G}}\sum_{i=1}^{n}\eta_{i}\left(\phi_{i}\circ g\right)\left(\boldsymbol{x}_{i}\right)\right]\leq L\frac{1}{n}\mathbb{E}_{\eta}\left[\sup_{g\in\mathcal{G}}\sum_{i=1}^{n}\eta_{i}g\left(\boldsymbol{x}_{i}\right)\right],
\]
}where $\left\{ \eta_{i}\right\} _{i\in[n]}$ are independent Rademacher
random variables, and the expectation $\mathbb{E}_{\eta}\left[\cdot\right]$
is taken over the randomness of $\left\{ \eta_{i}\right\} _{i\in\left[n\right]}$.
\end{lem}

\begin{lem}
\label{lem: Bartlett Corollary 2.2}{[}\citet[Corollary 2.2]{bartlett2005local}{]}
Let $\mathcal{G}$ be a class of functions that map $\mathcal{X}$
to $\left[-b,b\right]$ with $b>0$, and $\boldsymbol{x}_{1},\ldots,\boldsymbol{x}_{n}\in\mathcal{X}$
be independent and identically distributed. For every constants $t>0$
and $r$ that satisfy{\small{}
\[
r\geq10b\mathbb{E}\left[\sup_{g\in\mathcal{G}:\mathbb{E}\left[g^{2}\right]\leq r}\frac{1}{n}\sum_{i=1}^{n}\eta_{i}g\left(\boldsymbol{x}_{i}\right)\right]+\frac{11b^{2}t}{n},
\]
}then with probability at least $1-\exp\left(-t\right)$, {\small{}
\[
\left\{ g\in\mathcal{G}:\mathbb{E}\left[g^{2}\right]\leq r\right\} \subseteq\left\{ g\in\mathcal{G}:\frac{1}{n}\sum_{i=1}^{n}g\left(\boldsymbol{x}_{i}\right)^{2}\leq2r\right\} .
\]
}{\small\par}
\end{lem}

\begin{lem}
\label{lem: bounded expectation} Let $x_{1}$ and $x_{2}$ be two
random variables and $a,b>0$ be positive constants. Suppose $x_{1},x_{2}\in\left[0,a\right]$
and $\Pr\left(x_{1}\leq x_{2}\right)\geq1-b.$ Then {\small{}
\[
\mathbb{E}\left[x_{1}-x_{2}\right]\leq ab.
\]
}{\small\par}
\end{lem}

\begin{proof}
Note that{\small{}
\begin{align*}
\mathbb{E}\left[x_{1}-x_{2}\right] & =\mathbb{E}\left[x_{1}-x_{2}\mid x_{1}>x_{2}\right]\Pr\left(x_{1}>x_{2}\right)+\mathbb{E}\left[x_{1}-x_{2}\mid x_{1}\leq x_{2}\right]\Pr\left(x_{1}\leq x_{2}\right)\\
 & \leq\mathbb{E}\left[x_{1}-x_{2}\mid x_{1}>x_{2}\right]b\leq ab.
\end{align*}
}{\small\par}
\end{proof}
\begin{defn}
{[}Metric entropy{]} \label{def: metric entropy }Let $\mathcal{G}$
be a class of functions mapping $\mathcal{X}$ to $\mathbb{R}$ and
$\boldsymbol{x}_{1},\ldots,\boldsymbol{x}_{n}\in\mathcal{X}$ be a
fixed sample of $n$ points. $\mathcal{V}=\left\{ g_{1},...,g_{m}\right\} $
is an $\infty$-norm cover of $\mathcal{G}$ on $\left\{ \boldsymbol{x}_{i}\right\} _{i\in\left[n\right]}$
at scale $\epsilon$ if for every $g\in\mathcal{G}$, there exists
$g_{j}\in\mathcal{V}$ such that {\small{}
\[
\max_{i\in\left[n\right]}\left|g\left(\boldsymbol{x}_{i}\right)-g_{j}\left(\boldsymbol{x}_{i}\right)\right|\leq\epsilon.
\]
}Define the $\infty$-norm covering number of $\mathcal{G}$ on $\left\{ \boldsymbol{x}_{i}\right\} _{i\in\left[n\right]}$
with covering radius $\epsilon$ as the minimum size of $\mathcal{V}$:{\small{}
\[
\mathcal{N}_{\infty}\left(\epsilon,\mathcal{G},\left\{ \boldsymbol{x}_{i}\right\} _{i\in\left[n\right]}\right)=\min\left\{ m:\mathcal{V}=\left\{ g_{1},...,g_{m}\right\} \ \text{is an \ensuremath{\infty}-norm cover of \ensuremath{\mathcal{G}} on }\left\{ \boldsymbol{x}_{i}\right\} _{i\in\left[n\right]}\text{ at scale }\epsilon\right\} .
\]
}And the logarithm of the covering number, $\log\mathcal{N}_{\infty}\left(\epsilon,\mathcal{G},\left\{ \boldsymbol{x}_{i}\right\} _{i\in\left[n\right]}\right)$,
is called metric entropy.
\end{defn}

\begin{lem}
\label{lem: Dudley's chaining}{[}\citet[Lemma 3]{farrell2021deep}{]}
Let $\mathcal{G}$ be a class of functions mapping $\mathcal{X}$
to $\mathbb{R}$ and $\boldsymbol{x}_{1},\ldots,\boldsymbol{x}_{n}\in\mathcal{X}$
be a fixed sample of $n$ points. Then,

{\footnotesize{}
\[
\mathbb{E}_{\eta}\left[\sup_{g\in\mathcal{G}:\sqrt{\frac{1}{n}\sum_{i=1}^{n}g\left(\boldsymbol{x}_{i}\right)^{2}}\leq r}\frac{1}{n}\sum_{i=1}^{n}\eta_{i}g\left(\boldsymbol{x}_{i}\right)\right]\leq\inf_{0\leq\lambda\leq r}\left\{ 4\lambda+\frac{12}{\sqrt{n}}\int_{\lambda}^{r}\sqrt{\log\mathcal{N}_{\infty}\left(\epsilon,\mathcal{G},\left\{ \boldsymbol{x}_{i}\right\} _{i=1,...,n}\right)}d\epsilon\right\} ,
\]
}where $\log\mathcal{N}_{\infty}\left(\epsilon,\mathcal{G},\left\{ \boldsymbol{x}_{i}\right\} _{i\in\left[n\right]}\right)$
is the metric entropy as defined in Definition \ref{def: metric entropy },
$\left\{ \eta_{i}\right\} _{i\in\left[n\right]}$ are independent
Rademacher random variables, and the expectation $\mathbb{E}_{\eta}\left[\cdot\right]$
is taken over the randomness of $\left\{ \eta_{i}\right\} _{i\in\left[n\right]}$. 
\end{lem}

\begin{defn}
{[}Pseudo-dimension{]}\label{def: pseudo-dimension} Let $\mathcal{G}$
be a set of functions mapping from $\mathcal{X}$ to $\mathbb{R}$
and suppose that $S=\left\{ \boldsymbol{x}_{1},\ldots,\boldsymbol{x}_{m}\right\} \subseteq\mathcal{X}$
is a set of $m$ points in the input space. Then, $S$ is pseudo-shattered
by $\mathcal{G}$ if there are real numbers $r_{1},...,r_{m}$ such
that for each $b\in\{0,1\}^{m}$ there is a function $g_{b}\in\mathcal{G}$
with $\text{sgn}\left(g_{b}\left(\boldsymbol{x}_{i}\right)-r_{i}\right)=b_{i}$
for $i\in\left[m\right]$. Then $\mathcal{G}$ has pseudo-dimension
$p$ if $p$ is the maximum cardinality of a subset $S$ of $\mathcal{X}$
that is pseudo-shattered by $\mathcal{G}$. The pseudo-dimension of
$\mathcal{G}$ is denoted by $\text{Pdim}\left(\mathcal{G}\right)$. 
\end{defn}

\begin{lem}
{[}\citet[Theorem 12.2]{anthony1999neural}{]}\label{lem: Anthony Theorem 12.2}
Let $\mathcal{G}$ be a set of functions mapping from $\mathcal{X}$
to $\left[-b,b\right]$. Let $\epsilon>0$ and suppose $\text{Pdim}\left(\mathcal{G}\right)\geq1$.
Then for any fixed sample of $n$ points $\boldsymbol{x}_{1},\ldots,\boldsymbol{x}_{n}\in\mathcal{X}$,
{\small{}
\[
\mathcal{\log N}_{\infty}\left(\epsilon,\mathcal{G},\left\{ \boldsymbol{x}_{i}\right\} _{i\in\left[n\right]}\right)\leq\text{Pdim}\left(\mathcal{G}\right)\cdot\log\left(\max\left\{ \frac{2b}{\epsilon},1\right\} \cdot e\cdot n\right),
\]
}where $\mathcal{\log N}_{\infty}\left(\epsilon,\mathcal{G},\left\{ \boldsymbol{x}_{i}\right\} _{i\in\left[n\right]}\right)$
is the metric entropy defined in Definition \ref{def: pseudo-dimension},
and $e$ is the exponential constant.
\end{lem}

\begin{proof}
Following \citet[Theorem 12.2]{anthony1999neural}, we have the upper
bound of the covering number as follows{\small{}
\[
\mathcal{N}_{\infty}\left(\epsilon,\mathcal{G},\left\{ \boldsymbol{x}_{i}\right\} _{i\in[n]}\right)\leq\sum_{k=1}^{\text{Pdim}\left(\mathcal{G}\right)}\left(\begin{array}{c}
n\\
k
\end{array}\right)\left(\frac{2b}{\epsilon}\right)^{k}\le\left(\max\left\{ \frac{2b}{\epsilon},1\right\} \right)^{\text{Pdim}\left(\mathcal{G}\right)}\cdot\sum_{k=1}^{\text{Pdim}\left(\mathcal{G}\right)}\left(\begin{array}{c}
n\\
k
\end{array}\right).
\]
}Then, if $n<\text{Pdim}\left(\mathcal{G}\right)$, {\small{}
\[
\sum_{k=1}^{\text{Pdim}\left(\mathcal{G}\right)}\left(\begin{array}{c}
n\\
k
\end{array}\right)=\sum_{k=1}^{n}\left(\begin{array}{c}
n\\
k
\end{array}\right)\leq2^{n}\leq2^{\text{Pdim}\left(\mathcal{G}\right)}.
\]
}And if $n\geq\text{Pdim}\left(\mathcal{G}\right)$, \citet[Theorem 3.7]{anthony1999neural}
implies{\small{}
\[
\sum_{k=1}^{\text{Pdim}\left(\mathcal{G}\right)}\left(\begin{array}{c}
n\\
k
\end{array}\right)\leq\left(\frac{e\cdot n}{\text{Pdim}\left(\mathcal{G}\right)}\right)^{\text{Pdim}\left(\mathcal{G}\right)}.
\]
}Therefore, for any positive integers $n$ and $\text{Pdim}\left(\mathcal{G}\right)$,{\small{}
\[
\sum_{k=1}^{\text{Pdim}(\mathcal{G})}\left(\begin{array}{c}
n\\
k
\end{array}\right)\leq\left(\max\left\{ \frac{e\cdot n}{\text{Pdim}\left(\mathcal{G}\right)},2\right\} \right)^{\text{Pdim}\left(\mathcal{G}\right)}\leq\left(e\cdot n\right)^{\text{Pdim}\left(\mathcal{G}\right)},
\]
}which completes the proof of the lemma.
\end{proof}
\begin{defn}
{[}VC-dimension{]}\label{def: VC-dimension } Let $\mathcal{G}$ be
a set of binary-valued functions with input space $\mathcal{X}$.
The growth function for the function class $\mathcal{G}$ is defined
by{\small{}
\[
\Pi_{\mathcal{G}}\left(m\right)=\max_{\left\{ \boldsymbol{x}_{1},...,\boldsymbol{x}_{m}\right\} \subseteq\mathcal{X}}\left|\left\{ \left(g\left(\boldsymbol{x}_{1}\right),...,g\left(\boldsymbol{x}_{m}\right)\right):g\in\mathcal{\mathcal{G}}\right\} \right|,
\]
}for $m\in\mathbb{N}$. The VC-dimension of the set $\mathcal{G}$
is largest value of $m$ such that $\Pi_{\mathcal{G}}\left(m\right)=2^{m}$.
\end{defn}

\begin{lem}
{[}\citet[Theorems 8.4 and 8.14]{anthony1999neural}{]}\label{lem: Anthony and Bartlett Theorems 8.4 and 8.14}
Consider the parameterized class{\small{}
\[
\mathcal{G}=\left\{ x\mapsto g\left(x;\theta\right):\theta\in\mathbb{R}^{d_{\theta}}\right\} ,
\]
}for some $\left\{ 0,1\right\} $-valued function $g$. Suppose that,
for each input $x\in\mathbb{R}^{d_{x}}$, there is an algorithm that
computes $g\left(x;\theta\right)$ and this computation takes no more
than $t$ operations of the following types:
\end{lem}

\begin{itemize}
\item \textit{the arithmetic operations $+$, $-$, $\times$, and $/$
on real numbers,}
\item \textit{jumps conditioned on $>$, $\geq$, $<$, $\leq$, $=$, and
$\neq$ comparisons of real numbers, and}
\item \textit{output $0$ or $1$.}
\end{itemize}
\textit{Then, $\text{VCdim}\left(\mathcal{G}\right)=O\left(d_{\theta}\cdot t\right)$.
If besides the operations mentioned above, the types of operations
also include}
\begin{itemize}
\item \textit{the exponential function $\alpha\rightarrow\text{exp}\left(\alpha\right)$
on real numbers,}
\end{itemize}
\textit{and if the $t$ steps include no more than $q$ in which the
exponential function is evaluated and $q\geq1$, then $\text{VCdim}(\mathcal{G})=O\left(d_{\theta}^{2}\cdot q^{2}+d_{\theta}\cdot q\cdot t\right)$.
}
\begin{lem}
\label{lem: approximation_infinitely_differentiable}For every $f_{*}\in\mathcal{W}_{\gamma}^{\beta,\infty}\left(\left[-1,1\right]^{d}\right)$,
where $\gamma<\infty$ is a fixed constant, and $\left\Vert f_{*}\right\Vert _{\infty,\left[-1,1\right]^{d}}\leq1.1M$,
there exists a sequence of shallow neural networks $\left\{ g_{r}\right\} $
such that, as $r$ tends to infinity,{\small{}
\[
\sup_{\boldsymbol{x}\in\left[-1,1\right]^{d}}\left|f_{*}\left(\boldsymbol{x}\right)-g_{r}\left(\boldsymbol{x}\right)\right|\lesssim r^{-\beta/d},
\]
}where $g_{r}$ with $r$ neurons is of the form{\small{}
\[
g_{r}\in\left\{ g_{r}:\left[-1,1\right]^{d}\rightarrow\mathbb{R},g_{r}\left(\boldsymbol{x}\right)=\boldsymbol{a}\cdot\mathbf{\bm{\sigma}}\left(\boldsymbol{C}\boldsymbol{x}+\boldsymbol{c}\right)+b,\boldsymbol{C}\in\mathbb{R}^{r\times d},\boldsymbol{a},\boldsymbol{c}\in\mathbb{R}^{r},b\in\mathbb{R},\left\Vert g_{r}\right\Vert _{\infty,\left[-1,1\right]^{d}}\leq\bar{z}\right\} ,
\]
}and the activation function $\mathbf{\bm{\sigma}}:\mathbb{R}^{r}\rightarrow\mathbb{R}^{r}$
applies $\sigma:\mathbb{R}\rightarrow\mathbb{R}$ elementwise, where
$\sigma\left(\cdot\right)$ is infinitely differentiable and non-polynomial.
\end{lem}

\begin{proof}
This Lemma is a minor adjustment of Theorem 1 in \citet{poggio2017and}.
Theorem 1 in \citet{poggio2017and} implies that for every $f_{*}\in\mathcal{W}_{\gamma}^{\beta,\infty}\left(\left[-1,1\right]^{d}\right)$,
there exists $g_{r}\in\mathcal{G}_{r}$ such that{\small{}
\[
\sup_{\boldsymbol{x}\in[-1,1]^{d}}\left|\frac{1}{\gamma}f_{*}\left(\boldsymbol{x}\right)-g_{r}\left(\boldsymbol{x}\right)\right|\lesssim r^{-\beta/d},
\]
}where $\mathcal{G}_{r}\coloneqq\left\{ g_{r}:\left[-1,1\right]^{d}\rightarrow\mathbb{R},g_{r}\left(\boldsymbol{x}\right)=\boldsymbol{a}\cdot\mathbf{\bm{\sigma}}\left(\boldsymbol{C}\boldsymbol{x}+\boldsymbol{c}\right)+b,\boldsymbol{C}\in\mathbb{R}^{r\times d},\boldsymbol{a},\boldsymbol{c}\in\mathbb{R}^{r},b\in\mathbb{R}\right\} .$
As $\gamma<\infty$, this provides, for some $g_{r}\in\mathcal{G}_{r}$
and a fixed constant $B<\infty$, that{\small{}
\begin{equation}
\sup_{\boldsymbol{x}\in\left[-1,1\right]^{d}}\left|f_{*}\left(\boldsymbol{x}\right)-g_{r}\left(\boldsymbol{x}\right)\right|\leq Br^{-\beta/d}.\label{eq: poggio2017_2}
\end{equation}
}To complete the proof, we show that $g_{r}\in\mathcal{G}_{r}$ can
be restricted to have a bounded sup norm such that $\left\Vert g_{r}\right\Vert _{\infty,\left[-1,1\right]^{d}}\leq\bar{z}$.
Let $g_{r*}\in\mathcal{G}_{r}$ satisfy \ref{eq: poggio2017_2}. As
$\left\Vert f_{*}\right\Vert _{\infty,\left[-1,1\right]^{d}}\leq1.1M$,
it holds{\small{}
\[
\left\Vert g_{r*}\right\Vert _{\infty,\left[-1,1\right]^{d}}\leq Br^{-\beta/d}+\left\Vert f_{*}\right\Vert _{\infty,\left[-1,1\right]^{d}}\leq\bar{z},
\]
}when $r$ is large enough such that $r\geq\left(B/\left(\bar{z}-1.1M\right)\right){}^{d/\beta}$.
This completes the proof.
\end{proof}
\begin{lem}
{[}\citet[Theorem 2]{janson1988normal}{]}\label{lem: Janson Theorem 2}
Let $\left\{ x_{1},...,x_{n}\right\} $ be a sequence of random variables
and $\omega_{n}$ be the maximal degree of the dependency graph of
the sequence. Set $\omega_{n}=1$ if the dependency graph has no edges.
Denote $\sigma_{n}=\left(\text{Var}\left(\sum_{i\in\left[n\right]}x_{i}\right)\right)^{1/2}$.
If there exists a sequence of real numbers $a_{n}$ and an integer
$m$ such that as $n\rightarrow\infty$,{\small{}
\[
\left(\sigma_{n}\right)^{-2}\omega_{n}\sum_{i=1}^{n}\mathbb{E}\left[x_{i}^{2}\mathbb{I}\left\{ \left|x_{i}\right|>a_{n}\right\} \right]\rightarrow0,\ \text{and}\ \left(\sigma_{n}\right)^{-1}\left(n\right)^{1/m}\left(\omega_{n}\right)^{\frac{m-1}{m}}a_{n}\rightarrow0,
\]
}then,{\small{}
\[
\left(\sigma_{n}\right)^{-1}\sum_{i=1}^{n}\left(x_{i}-\mathbb{E}\left[x_{i}\right]\right)\overset{d}{\to}N\left(0,1\right).
\]
}{\small\par}
\end{lem}

\begin{lem}
\label{lem: Hansen 2.6.2}Let $\left\{ z_{n,i}\right\} _{i\in\left[n\right]}$
be a row-wise triangular array in which variables in the same row
are mutually independent and non-identically distributed. Suppose
$\mathbb{E}\left[z_{n,i}\right]=\mu_{n,i}$ and $\mathbb{E}\left|z_{n,i}\right|^{1+\delta}<C<\infty$
for all $n$ and $i$, where $\delta>0$ and $C>0$ are some fixed
constants. Then, as $n\rightarrow\infty$, {\small{}
\[
\frac{1}{n}\sum_{i=1}^{n}\left(z_{n,i}-\mu_{n,i}\right)\overset{p}{\rightarrow}0.
\]
}{\small\par}
\end{lem}

\begin{proof}
This lemma can be directly justified by the proof of \citet[Lemma 2.6.2]{hansen2004inference}.
\end{proof}
\begin{lem}
{[}\citet[Theorem 2.1]{bartlett2005local}{]}\label{lem: Bartlett Theorem 2.1}
Let $\left(\mathcal{X},P\right)$ be a probability space. Let $\mathcal{G}$
be a class of measurable functions from $\mathcal{X}$ to $\left[a,b\right]$.
Let $\left\{ \boldsymbol{x}_{i}\right\} _{i\in\left[n\right]}$ be
independent random variables distributed according to $P$. Assume
that there is some $r>0$ such that for every $g\in\mathcal{G}$,
$\text{Var}\left(g\left(\boldsymbol{x}_{i}\right)\right)\leq r$.
Then for every $t>0$, with probability at least $1-\exp\left(-t\right)$,{\footnotesize{}
\[
\sup_{g\in\mathcal{G}}\left(\mathbb{E}\left[g\left(\boldsymbol{x}_{i}\right)\right]-\frac{1}{n}\sum_{i=1}^{n}g\left(\boldsymbol{x}_{i}\right)\right)\leq\underset{\alpha>0}{\inf}\left\{ 2\left(1+\alpha\right)\mathbb{E}\left[\sup_{g\in\mathcal{G}}\frac{1}{n}\sum_{i=1}^{n}\eta_{i}g\left(\boldsymbol{x}_{i}\right)\right]+\sqrt{\frac{2rt}{n}}+\left(b-a\right)\left(\frac{1}{3}+\frac{1}{\alpha}\right)\frac{t}{n}\right\} ,
\]
}where $\left\{ \eta_{i}\right\} _{i\in\left[n\right]}$ are independent
Rademacher random variables. And the same result holds for $\sup_{g\in\mathcal{G}}\left(\frac{1}{n}\sum_{i=1}^{n}g\left(\boldsymbol{x}_{i}\right)-\mathbb{E}\left[g\left(\boldsymbol{x}_{i}\right)\right]\right)$.
\end{lem}

\section{Proof of Theorem \ref{thm: rate_of_convergence} \label{sec:Proof-of-Theorem 1}}

The proof of Theorem \ref{thm: rate_of_convergence} takes four steps,
as we show in this section.

\subsection{Step 1: main decomposition}

Our proof is based on the decomposition {\small{}
\begin{align*}
c_{1}\mathbb{E}\left[\left(z_{i}\left(\hat{\boldsymbol{\theta}}\right)-z_{*i}\left(\boldsymbol{f}_{*}\right)\right)^{2}\right]\leq & \mathbb{E}\left[\ell\left(y_{i},z_{i}\left(\hat{\boldsymbol{\theta}}\right)\right)-\ell\left(y_{i},z_{*i}\left(\boldsymbol{f}_{*}\right)\right)\right]\\
= & \mathbb{E}\left[\ell\left(y_{i},z_{i}\left(\hat{\boldsymbol{\theta}}\right)\right)-\ell\left(y_{i},z_{i}\left(\boldsymbol{\theta}_{*}\right)\right)\right]+\min_{\boldsymbol{\theta}\in\Theta_{d_{h},\bar{z}}}\mathbb{E}\left[\ell\left(y_{i},z_{i}\left(\boldsymbol{\theta}\right)\right)-\ell\left(y_{i},z_{*i}\left(\boldsymbol{f}_{*}\right)\right)\right]\\
\leq & \mathbb{E}\left[\ell\left(y_{i},z_{i}\left(\hat{\boldsymbol{\theta}}\right)\right)-\ell\left(y_{i},z_{i}\left(\boldsymbol{\theta}_{*}\right)\right)\right]+c_{2}\min_{\boldsymbol{\theta}\in\Theta_{d_{h},\bar{z}}}\mathbb{E}\left[\left(z_{i}\left(\boldsymbol{\theta}\right)-z_{*i}\left(\boldsymbol{f}_{*}\right)\right)^{2}\right],
\end{align*}
}where the first and second inequalities hold under Assumption I 3,
and the equality holds by the definition of $\boldsymbol{\theta}_{*}$.
We call $T_{1}=\mathbb{E}\left[\ell\left(y_{i},z_{i}\left(\hat{\boldsymbol{\theta}}\right)\right)-\ell\left(y_{i},z_{i}\left(\boldsymbol{\theta}_{*}\right)\right)\right]$
the estimation error and $T_{2}=\min_{\boldsymbol{\theta}\in\Theta_{d_{h},\bar{z}}}\mathbb{E}\left[\left(z_{i}\left(\boldsymbol{\theta}\right)-z_{*i}\left(\boldsymbol{f}_{*}\right)\right)^{2}\right]$
the approximation error. The expectation in $T_{1}$ is taken over
the randomness of $\left\{ y_{i},\boldsymbol{\xi}_{i}\right\} $ and
that in $T_{2}$ is taken over $\boldsymbol{\xi}_{i}$.

\subsection{Step 2: bound the estimation error $T_{1}$}

To derive an upper bound of the estimation error, we follow \citet{bartlett2005local}
to apply the localization analysis. The strategy of using localization
analysis to derive convergence rate has been previously adopted in
\citet{farrell2021deep} and references therein, and we pursue a similar
approach. In our context, we need to apply the localization analysis
to accommodate dependent data, and to examine the complexity measure
(pseudo-dimension) specific to the GNN architecture.

In particular, we first apply Lemma \ref{lem: Bartlett Theorem 3.3}
to the set of functions\\
${\cal G}=\left\{ \ell\left(y_{i},z_{i}\left(\boldsymbol{\theta}\right)\right)-\ell\left(y_{i},z_{i}\left(\boldsymbol{\theta}_{*}\right)\right):\boldsymbol{\theta}\in\Theta_{d_{h},\bar{z}}\right\} $.
By Assumption I 3, it holds for every $\boldsymbol{\theta}\in\Theta_{d_{h},\bar{z}}$,
{\small{}
\[
\left|\ell\left(y_{i},z_{i}\left(\boldsymbol{\theta}\right)\right)-\ell\left(y_{i},z_{i}\left(\boldsymbol{\theta}_{*}\right)\right)\right|\leq c_{\ell}\left|z_{i}\left(\boldsymbol{\theta}\right)-z_{i}\left(\boldsymbol{\theta}_{*}\right)\right|\leq2c_{\ell}\bar{z}.
\]
}So the range of any function $g$ in ${\cal {\cal G}}$ is $\left[-2c_{\ell}\bar{z},2c_{\ell}\bar{z}\right]$.

Provided that the sample $\left\{ \left(y_{i},\boldsymbol{\xi}_{i}\right)\right\} _{i\in\mathcal{C}_{j}}$
is i.i.d. within each cover $\mathcal{C}_{j}$, we could apply Lemma
\ref{lem: Bartlett Theorem 3.3} to each cover $\mathcal{C}_{j}$
separately if the following two conditions hold:
\begin{enumerate}
\item There is a functional $T:\mathcal{G}\rightarrow\mathbb{R}_{+}$ and
some constant $B>0$ such that for every $g\in\mathcal{G}$, $\text{Var}\left(g\right)\leq T\left(g\right)\leq B\mathbb{E}\left[g\right]$.
\item For each $j\in\left[J\right]$, there exists a sub-root function $\psi_{j}$
with fixed point $r_{j}^{*}$ such that for any $r\geq r_{j}^{*}$,{\small{}
\begin{equation}
\psi_{j}\left(r\right)\geq B\mathbb{E}\left[\sup_{g\in\mathcal{G}:T\left(g\right)\leq r}\frac{1}{\left|\mathcal{C}_{j}\right|}\sum_{i\in\mathcal{C}_{j}}\eta_{i}g\left(y_{i},\boldsymbol{\xi}_{i}\right)\right],\label{eq: bartlett_condition1}
\end{equation}
}where $\left\{ \eta_{i}\right\} _{i\in\mathcal{C}_{j}}$ are independent
Rademacher random variables, and the expectation is taken over the
randomness of $\left\{ \left(\eta_{i},y_{i},\boldsymbol{\xi}_{i}\right)\right\} _{i\in\mathcal{C}_{j}}$.
We define sub-root functions, fixed points, and Rademacher variables
in Definition \ref{def: sub-root function}, Lemma \ref{lem: Bartlett Lemma 3.2},
and Definition \ref{def: Rademacher variables}, respectively.
\end{enumerate}
If the above conditions hold, we can apply Lemma \ref{lem: Bartlett Theorem 3.3}
to obtain that for every cover $j\in\left[J\right]$, with probability
at least $1-\exp\left(-\log J-\rho\right)$,{\small{}
\[
\forall g\in\mathcal{G},\quad\mathbb{E}\left[g\right]\leq\max\left\{ \frac{1}{\left|\mathcal{C}_{j}\right|}\sum_{i\in\mathcal{C}_{j}}g\left(y_{i},\boldsymbol{\xi}_{i}\right),\frac{2}{\left|\mathcal{C}_{j}\right|}\sum_{i\in\mathcal{C}_{j}}g\left(y_{i},\boldsymbol{\xi}_{i}\right)\right\} +Cr_{j}^{*}+C\frac{\log J+\rho}{\left|\mathcal{C}_{j}\right|}.
\]
}Then, given that $\frac{1}{n}\sum_{i=1}^{n}\left[\ell\left(y_{i},z_{i}\left(\hat{\boldsymbol{\theta}}\right)\right)-\ell\left(y_{i},z_{i}\left(\boldsymbol{\theta}_{*}\right)\right)\right]\leq0$
and $\mathbb{E}\left[\ell\left(y_{i},z_{i}\left(\hat{\boldsymbol{\theta}}\right)\right)-\ell\left(y_{i},z_{i}\left(\boldsymbol{\theta}_{*}\right)\right)\right]\geq0$,
it holds with probability at least $1-\exp\left(-\rho\right)$, {\small{}
\begin{equation}
\mathbb{E}\left[\ell\left(y_{i},z_{i}\left(\hat{\boldsymbol{\theta}}\right)\right)-\ell\left(y_{i},z_{i}\left(\boldsymbol{\theta}_{*}\right)\right)\right]\leq C\cdot\frac{1}{n}\left(\sum_{j=1}^{J}\left|\mathcal{C}_{j}\right|\cdot r_{j}^{*}+J\log J+J\rho\right).\label{eq: approximation_error_r_j}
\end{equation}
}{\small\par}

\subsubsection*{Check conditions}

Next, we show that conditions 1 and 2 hold.

Let $T\left(g\right)=c_{\ell}^{2}\mathbb{E}\left[\left(z_{i}\left(\boldsymbol{\theta}\right)-z_{i}\left(\boldsymbol{\theta}_{*}\right)\right)^{2}\right]$
and $B=c_{\ell}^{2}/c_{3}$. We show below that for every $g\in\mathcal{G}$,
$\text{Var}\left(g\right)\leq T\left(g\right)\leq B\mathbb{E}\left[g\right]$.
In particular, by Assumption I 3, it holds{\small{}
\begin{align*}
\text{Var}\left(g\right) & \leq\mathbb{E}\left[\left(\ell\left(y_{i},z_{i}\left(\boldsymbol{\theta}\right)\right)-\ell\left(y_{i},z_{i}\left(\boldsymbol{\theta}_{*}\right)\right)\right)^{2}\right]\le c_{\ell}^{2}\mathbb{E}\left[\left(z_{i}\left(\boldsymbol{\theta}\right)-z_{i}\left(\boldsymbol{\theta}_{*}\right)\right)^{2}\right]=T\left(g\right),
\end{align*}
}and{\small{}
\[
T\left(g\right)=c_{\ell}^{2}\mathbb{E}\left[\left(z_{i}\left(\boldsymbol{\theta}\right)-z_{i}\left(\boldsymbol{\theta}_{*}\right)\right)^{2}\right]\leq\frac{c_{\ell}^{2}}{c_{3}}\mathbb{E}\left[\ell\left(y_{i},z_{i}\left(\boldsymbol{\theta}\right)\right)-\ell\left(y_{i},z_{i}\left(\boldsymbol{\theta}_{*}\right)\right)\right]=B\mathbb{E}\left[g\right].
\]
}Hence, condition 1 is satisfied.

To check condition 2, we first define a function $\psi_{j}:\mathbb{R}\rightarrow\mathbb{R}$
which satisfies (\ref{eq: bartlett_condition1}). Afterwards, we show
that $\psi_{j}$ is a sub-root function. Note that

{\small{}
\begin{align}
 & B\mathbb{E}\left[\sup_{g\in\mathcal{G}:T\left(g\right)\leq r}\frac{1}{\left|\mathcal{C}_{j}\right|}\sum_{i\in\mathcal{C}_{j}}\eta_{i}g\left(y_{i},\boldsymbol{\xi}_{i}\right)\right]\nonumber \\
= & \frac{c_{\ell}^{2}}{c_{3}}\mathbb{E}\left[\sup_{\boldsymbol{\theta}\in\Theta_{d_{h},\bar{z}}:\mathbb{E}\left[\left(z_{i}\left(\boldsymbol{\theta}\right)-z_{i}\left(\boldsymbol{\theta}_{*}\right)\right)^{2}\right]\leq\frac{r}{c_{\ell}^{2}}}\frac{1}{\left|\mathcal{C}_{j}\right|}\sum_{i\in\mathcal{C}_{j}}\eta_{i}\left(\ell\left(y_{i},z_{i}\left(\boldsymbol{\theta}\right)\right)-\ell\left(y_{i},z_{i}\left(\boldsymbol{\theta}_{*}\right)\right)\right)\right]\nonumber \\
\leq & \frac{c_{\ell}^{3}}{c_{3}}\mathbb{E}\left[\sup_{\boldsymbol{\theta}\in\Theta_{d_{h},\bar{z}}:\mathbb{E}\left[\left(z_{i}\left(\boldsymbol{\theta}\right)-z_{i}\left(\boldsymbol{\theta}_{*}\right)\right)^{2}\right]\leq\frac{r}{c_{\ell}^{2}}}\frac{1}{\left|\mathcal{C}_{j}\right|}\sum_{i\in\mathcal{C}_{j}}\eta_{i}\left(z_{i}\left(\boldsymbol{\theta}\right)-z_{i}\left(\boldsymbol{\theta}_{*}\right)\right)\right]\nonumber \\
\leq & 20\bar{z}c_{\ell}^{2}\max\left\{ \frac{c_{\ell}}{c_{3}},1\right\} \mathbb{E}\left[\sup_{\begin{array}{c}
\alpha\in\left[0,1\right],\boldsymbol{\theta}\in\Theta_{d_{h},\bar{z}}\\
\mathbb{E}\left[\alpha^{2}\left(z_{i}\left(\boldsymbol{\theta}\right)-z_{i}\left(\boldsymbol{\theta}_{*}\right)\right)^{2}\right]\leq\frac{r}{c_{\ell}^{2}}
\end{array}}\frac{1}{\left|\mathcal{C}_{j}\right|}\sum_{i\in\mathcal{C}_{j}}\eta_{i}\alpha\left(z_{i}\left(\boldsymbol{\theta}\right)-z_{i}\left(\boldsymbol{\theta}_{*}\right)\right)\right]+\frac{44c_{\ell}^{2}\bar{z}^{2}\log\left|\mathcal{C}_{j}\right|}{\left|\mathcal{C}_{j}\right|}\label{eq: psi_j}\\
\eqqcolon & \psi_{j}\left(r\right)\nonumber 
\end{align}
}where the first inequality holds by Lemma \ref{lem: contraction},
Assumption I 3, the definition of Rademacher variables, and the independence
between $\eta_{i}$ and $\left\{ y_{i},\boldsymbol{\xi}_{i}\right\} $
for every $i$. The last equality introduces the definition of $\psi_{j}$,
which satisfies (\ref{eq: bartlett_condition1}).

Clearly, $\psi_{j}$ is nonnegative and nondecreasing with $r$. To
show that $\psi\left(r\right)$ is a sub-root function, it is enough
to show that for any constants $r_{1}$ and $r_{2}$ such that $0<r_{1}\leq r_{2}$,
it holds $\psi_{j}\left(r_{1}\right)\geq\sqrt{\frac{r_{1}}{r_{2}}}\psi_{j}\left(r_{2}\right)$.
This inequality can be verified as for any realization of the sample
$\left\{ \boldsymbol{\xi}_{i}\right\} _{i\in\mathcal{C}_{j}}$and
Rademacher random variables $\left\{ \eta_{i}\right\} _{i\in\mathcal{C}_{j}}$,
it holds that {\small{}
\begin{align}
 & \sqrt{\frac{r_{1}}{r_{2}}}\sup_{\begin{array}{c}
\alpha\in\left[0,1\right],\boldsymbol{\theta}\in\Theta_{d_{h},\bar{z}}\\
\mathbb{E}\left[\alpha^{2}\left(z_{i}\left(\boldsymbol{\theta}\right)-z_{i}\left(\boldsymbol{\theta}_{*}\right)\right)^{2}\right]\leq r_{2}/c_{\ell}^{2}
\end{array}}\sum_{i\in\mathcal{C}_{j}}\eta_{i}\alpha\left(z_{i}\left(\boldsymbol{\theta}\right)-z_{i}\left(\boldsymbol{\theta}_{*}\right)\right)\nonumber \\
\leq & \sup_{\begin{array}{c}
\alpha\in\left[0,1\right],\boldsymbol{\theta}\in\Theta_{d_{h},\bar{z}}\\
\mathbb{E}\left[\alpha^{2}\left(z_{i}\left(\boldsymbol{\theta}\right)-z_{i}\left(\boldsymbol{\theta}_{*}\right)\right)^{2}\right]\leq r_{1}/c_{\ell}^{2}
\end{array}}\sum_{i\in\mathcal{C}_{j}}\eta_{i}\alpha\left(z_{i}\left(\boldsymbol{\theta}\right)-z_{i}\left(\boldsymbol{\theta}_{*}\right)\right).\label{pf: sub_root}
\end{align}
}To show (\ref{pf: sub_root}), we set $\alpha_{0}$ and $\boldsymbol{\theta}_{0}$
such that the first supremum in this inequality is obtained (if the
supremum cannot be reached only a minor modification will suffice).
In particular, $\alpha_{0}$ and $\boldsymbol{\theta}_{0}$ satisfy
the following conditions, $\alpha_{0}\in\left[0,1\right]$, $\boldsymbol{\theta}_{0}\in\Theta_{d_{h},\bar{z}}$,
$\mathbb{E}\left[\alpha_{0}^{2}\left(z_{i}\left(\boldsymbol{\theta}_{0}\right)-z_{i}\left(\boldsymbol{\theta}_{*}\right)\right)^{2}\right]\leq r_{2}/c_{\ell}^{2}$,
and

{\small{}
\begin{align*}
 & \sqrt{\frac{r_{1}}{r_{2}}}\sup_{\begin{array}{c}
\alpha\in\left[0,1\right],\boldsymbol{\theta}\in\Theta_{d_{h},\bar{z}}\\
\mathbb{E}\left[\alpha^{2}\left(z_{i}\left(\boldsymbol{\theta}\right)-z_{i}\left(\boldsymbol{\theta}_{*}\right)\right)^{2}\right]\leq r_{2}/c_{\ell}^{2}
\end{array}}\sum_{i\in\mathcal{C}_{j}}\eta_{i}\alpha\left(z_{i}\left(\boldsymbol{\theta}\right)-z_{i}\left(\boldsymbol{\theta}_{*}\right)\right)\\
= & \sqrt{\frac{r_{1}}{r_{2}}}\sum_{i\in\mathcal{C}_{j}}\eta_{i}\alpha_{0}\left(z_{i}\left(\boldsymbol{\theta}_{0}\right)-z_{i}\left(\boldsymbol{\theta}_{*}\right)\right)=\sum_{i\in\mathcal{C}_{j}}\eta_{i}\tilde{\alpha}\left(z_{i}\left(\boldsymbol{\theta}_{0}\right)-z_{i}\left(\theta_{*}\right)\right),
\end{align*}
}where the last equality holds by setting $\tilde{\alpha}=\sqrt{\frac{r_{1}}{r_{2}}}\alpha_{0}$.
Provided that $\tilde{\alpha}\in\left[0,1\right]$, $\boldsymbol{\theta}_{0}\in\Theta_{d_{h},\bar{z}}$,
and\\
$\mathbb{E}\left[\tilde{\alpha}^{2}\left(z_{i}\left(\boldsymbol{\theta}_{0}\right)-z_{i}\left(\boldsymbol{\theta}_{*}\right)\right)^{2}\right]=\frac{r_{1}}{r_{2}}\mathbb{E}\left[\alpha_{0}^{2}\left(z_{i}\left(\boldsymbol{\theta}_{0}\right)-z_{i}\left(\boldsymbol{\theta}_{*}\right)\right)^{2}\right]\leq r_{1}/c_{\ell}^{2}$,
this implies (\ref{pf: sub_root}).

Given that $\psi_{j}$ is a sub-root function, by Lemma \ref{lem: Bartlett Lemma 3.2},
we could define the fixed point $r_{j}^{*}$ of $\psi_{j}$. Hence,
we have verified condition 2.

Now as the conditions of Lemma \ref{lem: Bartlett Theorem 3.3} hold,
we have verified (\ref{eq: approximation_error_r_j}). In the following,
to finish bounding the estimation error, we derive an upper bound
of the fixed point $r_{j}^{*}$.

\subsubsection*{Derive an upper bound of $r_{j}^{*}$}

Define $\mathcal{G}_{GNN}=\left\{ \alpha\left(z_{i}\left(\boldsymbol{\theta}\right)-z_{i}\left(\boldsymbol{\theta}_{*}\right)\right):\alpha\in\left[0,1\right],\boldsymbol{\theta}\in\Theta_{d_{h},\bar{z}}\right\} $.
We show that there exists a finite constant $C$ such that

{\small{}
\begin{equation}
r_{j}^{*}\leq C\cdot\frac{1+\log\left|\mathcal{C}_{j}\right|}{\left|\mathcal{C}_{j}\right|}\text{\ensuremath{\cdot\text{Pdim}}}\left(\mathcal{G}_{GNN}\right),\label{eq: upper_bound_r_j}
\end{equation}
}where $\text{\text{Pdim}}(\cdot)$ denotes the pseudo-dimension as
in Definition \ref{def: pseudo-dimension}.

Note the definition of $r_{j}^{*}$ ($r_{j}^{*}=\psi_{j}\left(r_{j}^{*}\right)$)
implies{\small{}
\begin{equation}
\frac{r_{j}^{*}}{c_{\ell}^{2}}\geq20\bar{z}\mathbb{E}\left[\sup_{\begin{array}{c}
\alpha\in\left[0,1\right],\boldsymbol{\theta}\in\Theta_{d_{h},\bar{z}}\\
\mathbb{E}\left[\alpha^{2}\left(z_{i}\left(\boldsymbol{\theta}\right)-z_{i}\left(\boldsymbol{\theta}_{*}\right)\right)^{2}\right]\leq\frac{r_{j}^{*}}{c_{\ell}^{2}}
\end{array}}\frac{1}{\left|\mathcal{C}_{j}\right|}\sum_{i\in\mathcal{C}_{j}}\eta_{i}\alpha\left(z_{i}\left(\boldsymbol{\theta}\right)-z_{i}\left(\boldsymbol{\theta}_{*}\right)\right)\right]+\frac{44\bar{z}^{2}\log\left|\mathcal{C}_{j}\right|}{\left|\mathcal{C}_{j}\right|}.\label{pf: fixed point inequality}
\end{equation}
}As the sample $\left\{ \boldsymbol{\xi}_{i}\right\} $ is i.i.d.
within each partition $\mathcal{C}_{j}$, Lemma \ref{lem: Bartlett Corollary 2.2}
and inequality (\ref{pf: fixed point inequality}) imply that with
probability at least $1-\left|\mathcal{C}_{j}\right|^{-1}$,

{\small{}
\begin{align*}
 & \left\{ \alpha\left(z_{i}\left(\boldsymbol{\theta}\right)-z_{i}\left(\boldsymbol{\theta}_{*}\right)\right):\alpha\in\left[0,1\right],\boldsymbol{\theta}\in\Theta_{d_{h},\bar{z}},\mathbb{E}\left[\alpha^{2}\left(z_{i}\left(\boldsymbol{\theta}\right)-z_{i}\left(\boldsymbol{\theta}_{*}\right)\right)^{2}\right]\leq r_{j}^{*}/c_{\ell}^{2}\right\} \\
\subseteq & \left\{ \alpha\left(z_{i}\left(\boldsymbol{\theta}\right)-z_{i}\left(\boldsymbol{\theta}_{*}\right)\right):\alpha\in\left[0,1\right],\boldsymbol{\theta}\in\Theta_{d_{h},\bar{z}},\frac{1}{\left|\mathcal{C}_{j}\right|}\sum_{i\in\mathcal{C}_{j}}\alpha^{2}\left(z_{i}\left(\boldsymbol{\theta}\right)-z_{i}\left(\boldsymbol{\theta}_{*}\right)\right)^{2}\leq2r_{j}^{*}/c_{\ell}^{2}\right\} .
\end{align*}
}Therefore, with probability at least $1-\left|\mathcal{C}_{j}\right|^{-1}$,
{\small{}
\begin{align*}
 & \sup_{\begin{array}{c}
\alpha\in\left[0,1\right],\boldsymbol{\theta}\in\Theta_{d_{h},\bar{z}}\\
\mathbb{E}\left[\alpha^{2}\left(z_{i}\left(\boldsymbol{\theta}\right)-z_{i}\left(\boldsymbol{\theta}_{*}\right)\right)^{2}\right]\leq r_{j}^{*}/c_{\ell}^{2}
\end{array}}\frac{1}{\left|\mathcal{C}_{j}\right|}\sum_{i\in\mathcal{C}_{j}}\eta_{i}\alpha\left(z_{i}\left(\boldsymbol{\theta}\right)-z_{i}\left(\boldsymbol{\theta}_{*}\right)\right)\\
\leq & \sup_{\begin{array}{c}
\alpha\in\left[0,1\right],\boldsymbol{\theta}\in\Theta_{d_{h},\bar{z}}\\
\frac{1}{\left|\mathcal{C}_{j}\right|}\sum_{i\in\mathcal{C}_{j}}\alpha^{2}\left(z_{i}\left(\boldsymbol{\theta}\right)-z_{i}\left(\boldsymbol{\theta}_{*}\right)\right)^{2}\leq2r_{j}^{*}/c_{\ell}^{2}
\end{array}}\frac{1}{\left|\mathcal{C}_{j}\right|}\sum_{i\in\mathcal{C}_{j}}\eta_{i}\alpha\left(z_{i}\left(\boldsymbol{\theta}\right)-z_{i}\left(\boldsymbol{\theta}_{*}\right)\right).
\end{align*}
}Then, applying Lemma \ref{lem: bounded expectation} and given that
{\small{}
\begin{align*}
\sup_{\begin{array}{c}
\alpha\in\left[0,1\right],\boldsymbol{\theta}\in\Theta_{d_{h},\bar{z}}\\
\mathbb{E}\left[\alpha^{2}\left(z_{i}\left(\boldsymbol{\theta}\right)-z_{i}\left(\boldsymbol{\theta}_{*}\right)\right)^{2}\right]\leq r_{j}^{*}/c_{\ell}^{2}
\end{array}}\frac{1}{\left|\mathcal{C}_{j}\right|}\sum_{i\in\mathcal{C}_{j}}\eta_{i}\alpha\left(z_{i}\left(\boldsymbol{\theta}\right)-z_{i}\left(\boldsymbol{\theta}_{*}\right)\right) & \in\left[0,2\bar{z}\right],\\
\sup_{\begin{array}{c}
\alpha\in\left[0,1\right],\boldsymbol{\theta}\in\Theta_{d_{h},\bar{z}}\\
\frac{1}{|\mathcal{C}_{j}|}\sum_{i\in\mathcal{C}_{j}}\alpha^{2}\left(z_{i}\left(\boldsymbol{\theta}\right)-z_{i}\left(\boldsymbol{\theta}_{*}\right)\right)^{2}\leq2r_{j}^{*}/c_{\ell}^{2}
\end{array}}\frac{1}{\left|\mathcal{C}_{j}\right|}\sum_{i\in\mathcal{C}_{j}}\eta_{i}\alpha\left(z_{i}\left(\boldsymbol{\theta}\right)-z_{i}\left(\boldsymbol{\theta}_{*}\right)\right) & \in\left[0,2\bar{z}\right],
\end{align*}
}we obtain{\small{}
\begin{align}
 & \mathbb{E}\left[\sup_{\begin{array}{c}
\alpha\in\left[0,1\right],\boldsymbol{\theta}\in\Theta_{d_{h},\bar{z}}\\
\mathbb{E}\left[\alpha^{2}\left(z_{i}\left(\boldsymbol{\theta}\right)-z_{i}\left(\boldsymbol{\theta}_{*}\right)\right)^{2}\right]\leq r_{j}^{*}/c_{\ell}^{2}
\end{array}}\frac{1}{\left|\mathcal{C}_{j}\right|}\sum_{i\in\mathcal{C}_{j}}\eta_{i}\alpha\left(z_{i}\left(\boldsymbol{\theta}\right)-z_{i}\left(\boldsymbol{\theta}_{*}\right)\right)\right]\nonumber \\
\leq & \mathbb{E}\left[\sup_{\begin{array}{c}
\alpha\in\left[0,1\right],\boldsymbol{\theta}\in\Theta_{d_{h},\bar{z}}\\
\frac{1}{\left|\mathcal{C}_{j}\right|}\sum_{i\in\mathcal{C}_{j}}\alpha^{2}\left(z_{i}\left(\boldsymbol{\theta}\right)-z_{i}\left(\boldsymbol{\theta}_{*}\right)\right)^{2}\leq2r_{j}^{*}/c_{\ell}^{2}
\end{array}}\frac{1}{\left|\mathcal{C}_{j}\right|}\sum_{i\in\mathcal{C}_{j}}\eta_{i}\alpha\left(z_{i}\left(\boldsymbol{\theta}\right)-z_{i}\left(\boldsymbol{\theta}_{*}\right)\right)\right]+\frac{2\bar{z}}{\left|\mathcal{C}_{j}\right|}.\label{pf: expectation_to_sample_mean}
\end{align}
}In the following, we derive an upper bound of the fixed point. {\footnotesize{}
\begin{align*}
r_{j}^{*} & =C\cdot\mathbb{E}\left[\sup_{\begin{array}{c}
\alpha\in\left[0,1\right],\boldsymbol{\theta}\in\Theta_{d_{h},\bar{z}}\\
\mathbb{E}\left[\alpha^{2}\left(z_{i}\left(\boldsymbol{\theta}\right)-z_{i}\left(\boldsymbol{\theta}_{*}\right)\right)^{2}\right]\leq r_{j}^{*}/c_{\ell}^{2}
\end{array}}\frac{1}{\left|\mathcal{C}_{j}\right|}\sum_{i\in\mathcal{C}_{j}}\eta_{i}\alpha\left(z_{i}\left(\boldsymbol{\theta}\right)-z_{i}\left(\boldsymbol{\theta}_{*}\right)\right)\right]+\frac{44c_{\ell}^{2}\bar{z}^{2}\log\left|\mathcal{C}_{j}\right|}{\left|\mathcal{C}_{j}\right|}\\
 & \leq C\cdot\mathbb{E}\left[\sup_{\begin{array}{c}
\alpha\in\left[0,1\right],\boldsymbol{\theta}\in\Theta_{d_{h},\bar{z}}\\
\frac{1}{\left|\mathcal{C}_{j}\right|}\sum_{i\in\mathcal{C}_{j}}\alpha^{2}\left(z_{i}\left(\boldsymbol{\theta}\right)-z_{i}\left(\boldsymbol{\theta}_{*}\right)\right)^{2}\leq2r_{j}^{*}/c_{\ell}^{2}
\end{array}}\frac{1}{\left|\mathcal{C}_{j}\right|}\sum_{i\in\mathcal{C}_{j}}\eta_{i}\alpha\left(z_{i}\left(\boldsymbol{\theta}\right)-z_{i}\left(\boldsymbol{\theta}_{*}\right)\right)\right]+\frac{2C\bar{z}+44c_{\ell}^{2}\bar{z}^{2}\log\left|\mathcal{C}_{j}\right|}{\left|\mathcal{C}_{j}\right|}\\
 & \leq C\cdot\mathbb{E}\left[\inf_{0\leq\lambda\leq\sqrt{2r_{j}^{*}}/c_{\ell}}\left\{ 4\lambda+\frac{12}{\sqrt{\left|\mathcal{C}_{j}\right|}}\int_{\lambda}^{\sqrt{2r_{j}^{*}}/c_{\ell}}\sqrt{\log\mathcal{N}_{\infty}\left(\epsilon,\mathcal{G}_{GNN},\left\{ \boldsymbol{\xi}_{i}\right\} _{i\in\mathcal{C}_{j}}\right)}d\epsilon\right\} \right]+\frac{2C\bar{z}+44c_{\ell}^{2}\bar{z}^{2}\log\left|\mathcal{C}_{j}\right|}{\left|\mathcal{C}_{j}\right|}\\
 & \leq C\cdot\inf_{0\leq\lambda\leq\sqrt{2r_{j}^{*}}/c_{\ell}}\left\{ 4\lambda+\frac{12}{\sqrt{\left|\mathcal{C}_{j}\right|}}\int_{\lambda}^{\sqrt{2r_{j}^{*}}/c_{\ell}}\sqrt{\text{\text{Pdim}}\left(\mathcal{G}_{GNN}\right)\log\left(\max\left\{ \frac{4\bar{z}}{\epsilon},1\right\} \cdot e\cdot\left|\mathcal{C}_{j}\right|\right)}d\epsilon\right\} +\frac{2C\bar{z}+44c_{\ell}^{2}\bar{z}^{2}\log\left|\mathcal{C}_{j}\right|}{\left|\mathcal{C}_{j}\right|}\\
 & \leq C\cdot\inf_{0\leq\lambda\leq\sqrt{2r_{j}^{*}}/c_{\ell}}\left\{ 4\lambda+12\frac{\sqrt{2r_{j}^{*}}}{c_{\ell}}\sqrt{\frac{\text{\text{Pdim}}\left(\mathcal{G}_{GNN}\right)}{\left|\mathcal{C}_{j}\right|}}\sqrt{\log\left(\max\left\{ \frac{4\bar{z}}{\lambda},1\right\} \cdot e\cdot\left|\mathcal{C}_{j}\right|\right)}\right\} +\frac{2C\bar{z}+44c_{\ell}^{2}\bar{z}^{2}\log|\mathcal{C}_{j}|}{\left|\mathcal{C}_{j}\right|}\\
 & =C\cdot\left\{ 12\frac{\sqrt{2r_{j}^{*}}}{c_{\ell}}\sqrt{\frac{\text{\text{Pdim}}\left(\mathcal{G}_{GNN}\right)}{\left|\mathcal{C}_{j}\right|}}\left(\sqrt{\log\left(\max\left\{ \frac{4\bar{z}}{\lambda},1\right\} \cdot e\cdot\left|\mathcal{C}_{j}\right|\right)}+\frac{1}{3}\right)\right\} +\frac{2C\bar{z}+44c_{\ell}^{2}\bar{z}^{2}\log\left|\mathcal{C}_{j}\right|}{\left|\mathcal{C}_{j}\right|},
\end{align*}
}where the first equality corresponds to the definition of $r_{j}^{*}$,
the first inequality follows (\ref{pf: expectation_to_sample_mean}),
the second inequality follows Lemma \ref{lem: Dudley's chaining}
with $\log\mathcal{N}_{\infty}\left(\epsilon,\mathcal{G}_{GNN},\left\{ \boldsymbol{\xi}_{i}\right\} _{i\in\mathcal{C}_{j}}\right)$
being the metric entropy in Definition \ref{def: metric entropy },
the third inequality holds by Lemma \ref{lem: Anthony Theorem 12.2}
with $\text{\text{Pdim}}(\cdot)$ being the pseudo-dimension in Definition
\ref{def: pseudo-dimension}, the fourth inequality holds by noting
that the integrand is decreasing in $\epsilon$, and the last equality
holds by picking{\small{}
\[
\ensuremath{\lambda}=\frac{\sqrt{2r_{j}^{*}}}{c_{\ell}}\min\left\{ \sqrt{\frac{\text{\text{Pdim}}\left(\mathcal{G}_{GNN}\right)}{\left|\mathcal{C}_{j}\right|}},1\right\} \in\left(0,\frac{\sqrt{2r_{j}^{*}}}{c_{\ell}}\right].
\]
}{\small\par}

Suppose for now $r_{j}^{*}\geq\left|\mathcal{C}_{j}\right|^{-1}$
and later we will add back the alternative case. Then, it is straightforward
to show that $\max\left\{ \frac{4\bar{z}}{\lambda},1\right\} \leq C\cdot\left|\mathcal{C}_{j}\right|$
and hence{\small{}
\[
r_{j}^{*}\leq C\cdot\frac{1+\log\left|\mathcal{C}_{j}\right|}{\left|\mathcal{C}_{j}\right|}\text{\ensuremath{\cdot\text{Pdim}}}\left(\mathcal{G}_{GNN}\right).
\]
}Therefore, we have either $r_{j}^{*}<\left|\mathcal{C}_{j}\right|^{-1}$
or $r_{j}^{*}\leq C\cdot\frac{1+\log\left|\mathcal{C}_{j}\right|}{\left|\mathcal{C}_{j}\right|}\text{\ensuremath{\cdot\text{Pdim}}}\left(\mathcal{G}_{GNN}\right)$,
which gives (\ref{eq: upper_bound_r_j}).

\subsubsection*{Derive an upper bound of the Pseudo-dimension }

To complete the derivation of an upper bound of $r_{j}^{*}$, we further
show that $\text{\text{Pdim}}\left(\mathcal{G}_{GNN}\right)$ is bounded
by {\small{}
\begin{equation}
\text{\text{Pdim}}\left(\mathcal{G}_{GNN}\right)\leq C\cdot\left(c_{n}\right)^{s}\cdot\left(d_{h}\right)^{k},\label{eq: upper_bound_pseudo_dim}
\end{equation}
}where $k=2$, $4,$ or $6$ depending on the number of layers $L$
and the activation function $\sigma\left(\cdot\right)$, and $s=L$
or $2L-1$ depending on the activation function. 

To study $\text{\text{Pdim}}\left(\mathcal{G}_{GNN}\right)$, we first
introduce the function class

{\small{}
\[
\mathcal{BG}_{GNN}=\left\{ \left(\boldsymbol{\xi},r\right)\mapsto\text{sgn}\left(\alpha\left(z\left(\boldsymbol{\xi};\boldsymbol{\theta}\right)-z\left(\boldsymbol{\xi};\boldsymbol{\theta}_{*}\right)\right)-r\right):\alpha\in\left[0,1\right],\boldsymbol{\theta}\in\Theta_{d_{h},\bar{z}}\right\} ,
\]
}which includes the binary functions based on all functions in $\mathcal{G}_{GNN}$
and one extra real-valued input variable. The definitions of Pseudo-
and VC-dimensions, stated in Definitions \ref{def: pseudo-dimension}
and \ref{def: VC-dimension }, imply that{\small{}
\[
\text{\text{Pdim}}\left(\mathcal{G}_{GNN}\right)=\text{VCdim}\left(\mathcal{BG}_{GNN}\right).
\]
}{\small\par}

Next, we apply Lemma \ref{lem: Anthony and Bartlett Theorems 8.4 and 8.14}
to bound $\text{VCdim}\left(\mathcal{B}_{GNN}\right)$. Under Assumption
I 4, $\left|\mathcal{N}\left(i\right)\right|\leq c_{n}$ for every
$i\in\left[n\right]$. So the relevant neighborhood for node $i$
in the construction of $z_{i}\left(\boldsymbol{\theta}\right)$ includes
at most $C_{n}\coloneqq\sum_{k=0}^{L}\left(c_{n}\right)^{k}$ nodes,
who are up to distance $L$ from node $i$. So the input space of
the functions in class $\mathcal{BG}_{GNN}$ can be easily reformulated
to include (1) any $C_{n}\times C_{n}$ binary adjacency matrices
indicating the connections among the nodes in the relevant neighborhood
around $i$, (2) any $C_{n}\times d$ feature matrices for the relevant
nodes with each element of the matrices being in $\left[-1,1\right]$,
and (3) a real line for the scalar input $r$.

For simplicity, set $d_{h}^{(l)}\asymp d_{h}$ for $1\leq l\leq L$
and assume $L$ is finite as in Assumption I 1. As the parameters
in $\mathcal{BG}_{GNN}$ include 

{\small{}
\[
\left\{ \alpha\in\mathbb{R},\boldsymbol{a}\in\mathbb{R}^{d_{h}^{(L)}},b\in\mathbb{R},\left\{ \boldsymbol{A}^{(l)},\boldsymbol{A}_{\mathcal{N}}^{(l)}\in\mathbb{R}^{d_{h}^{(l)}\times d_{h}^{(l-1)}},\boldsymbol{b}^{(l)}\in\mathbb{R}^{d_{h}^{(l)}}\right\} {}_{l\in\left[L\right]}\right\} ,
\]
}where $d_{h}^{(0)}=d<\infty$, a straightforward counting suggests
the number of parameters in $\mathcal{BG}_{GNN}$ is $O\left(d_{h}\right)$
if $L=1$ and $O\left(\left(d_{h}\right)^{2}\right)$ if $L\geq2$. 

Suppose that the activation function $\sigma\left(\cdot\right)$ can
be computed using a finite number of operations listed in Lemma \ref{lem: Anthony and Bartlett Theorems 8.4 and 8.14},
as stated in Assumption I 2. Then, a simple counting suggests the
number of operations of the listed types in Lemma \ref{lem: Anthony and Bartlett Theorems 8.4 and 8.14}
for computing any functions in $\mathcal{BG}_{GNN}$ is $O\left(d_{h}+c_{n}\right)$
if $L=1$ and $O\left(\left(c_{n}\right)^{L-2}\left(\left(d_{h}\right)^{2}+\left(c_{n}\right)^{2}\right)\right)$
if $L\geq2$. Furthermore, if the activation function $\sigma\left(\cdot\right)$
involves taking exponential, the number of times the exponential function
is evaluated is $O\left(\left(c_{n}\right)^{L-1}d_{h}\right)$.

Then, a direct application of Lemma \ref{lem: Anthony and Bartlett Theorems 8.4 and 8.14}
shows that

{\small{}
\[
\text{VCdim}\left(\mathcal{B}_{GNN}\right)=\left\{ \begin{array}{c}
O\left(c_{n}\left(d_{h}\right)^{2}\right)\ \text{if}\ L=1\\
O\left(\left(c_{n}\right)^{L}\left(d_{h}\right)^{4}\right)\ \text{if}\ L\geq2
\end{array}\right.,
\]
}when $\sigma\left(\cdot\right)$ does not involve taking exponential,
and{\small{}
\[
\text{VCdim}\left(\mathcal{B}_{GNN}\right)=\left\{ \begin{array}{c}
O\left(c_{n}\left(d_{h}\right)^{4}\right)\ \text{if}\ L=1\\
O\left(\left(c_{n}\right)^{2L-1}\left(d_{h}\right)^{6}\right)\ \text{if}\ L\geq2
\end{array}\right.,
\]
}when $\sigma\left(\cdot\right)$ involves taking exponential.  

The rate in (\ref{eq: upper_bound_pseudo_dim}) seems to be comparable
with the one in \citet{scarselli2018vapnik}, which studies the VC-dimension
for GNNs. However, the rate in \citet{scarselli2018vapnik} cannot
be directly applied to our setting as their architecture is restricted
to a recursive structure with one hidden layer, which is different
from our non-recursive setup admitting any finite number of hidden
layers.

\subsubsection*{Upper bound of the estimation error }

Combining (\ref{eq: approximation_error_r_j}), (\ref{eq: upper_bound_r_j}),
and (\ref{eq: upper_bound_pseudo_dim}), we obtain an upper bound
for the estimation error $T_{1}$. For $k=2,4$ or $6$, and $s=L$
or $2L-1$, with probability at least $1-\exp\left(-\rho\right)$,

\begin{align}
\mathbb{E}\left[\ell\left(y_{i},z_{i}\left(\hat{\boldsymbol{\theta}}\right)\right)-\ell\left(y_{i},z_{i}\left(\boldsymbol{\theta}_{*}\right)\right)\right] & \leq C\cdot\left(\frac{1}{n}\sum_{j=1}^{J}\left(1+\log\left|\mathcal{C}_{j}\right|\right)\cdot\left(c_{n}\right)^{2}\cdot\left(d_{h}\right)^{k}+\frac{J\log J+J\rho}{n}\right).\label{eq: estimation_error_rate}
\end{align}

\subsection{Step 3: bound the approximation error $T_{2}$}

Define the approximation error $\left[\epsilon\left(d_{h}\right)\right]^{2}=\underset{\boldsymbol{\theta}\in\Theta_{d_{h},\bar{z}}}{\min}\mathbb{E}\left[\left(z_{i}\left(\boldsymbol{\theta}\right)-z_{*i}\left(\boldsymbol{f}_{*}\right)\right)^{2}\right]$.
In this subsection, we show that the approximation error satisfies
{\small{}
\[
\left[\epsilon\left(d_{h}\right)\right]^{2}\lesssim\sum_{l=1}^{L}\left(\text{\ensuremath{d_{h}^{(l)}}}\right)^{-\beta/d_{h*}^{(l-1)}},
\]
}with $d_{h*}^{(0)}=d$. To simplify the presentation, we provide
the proof for the case where $L=2$ below to show that{\small{}
\[
\left[\epsilon\left(\left\{ d_{h}^{(1)},d_{h}^{(2)}\right\} \right)\right]^{2}\lesssim\left(\text{\ensuremath{d_{h}^{(1)}}}\right){}^{-\beta/d}+\left(d_{h}^{(2)}\right){}^{-\beta/d_{h*}^{(1)}}.
\]
}The proofs for other cases are very similar and hence omitted for
brevity. 

Recall that for every $\boldsymbol{\theta}\in\Theta_{\left\{ d_{h}^{(1)},d_{h}^{(2)}\right\} }$,
we define $z_{i}\left(\boldsymbol{\theta}\right)$ as {\small{}
\begin{eqnarray}
\boldsymbol{h}_{i}^{(1)} & = & \mathbf{\bm{\sigma}}\left(\boldsymbol{A}^{(1)}\boldsymbol{x}_{i}+\boldsymbol{A}_{\mathcal{N}}^{(1)}\overline{\boldsymbol{x}}_{\mathcal{N}\left(i\right)}+\boldsymbol{b}^{(1)}\right),\nonumber \\
z_{i}\left(\boldsymbol{\theta}\right) & = & \boldsymbol{a}\cdot\mathbf{\bm{\sigma}}\left(\boldsymbol{A}^{(2)}\boldsymbol{h}_{i}^{(1)}+\boldsymbol{A}_{\mathcal{N}}^{(2)}\overline{\boldsymbol{h}}_{\mathcal{N}\left(i\right)}^{(1)}+\boldsymbol{b}^{(2)}\right)+b,\label{pf:approx_def_zi}
\end{eqnarray}
}where $\boldsymbol{A}^{(1)},\boldsymbol{A}_{\mathcal{N}}^{(1)}\in\mathbb{R}^{d_{h}^{(1)}\times d}$
and $\boldsymbol{A}^{(2)},\boldsymbol{A}_{\mathcal{N}}^{(2)}\in\mathbb{R}^{d_{h}^{(2)}\times d_{h}^{(1)}}$.

First, we introduce a new notation, $z_{i}\left(\boldsymbol{f}^{\left(1\right)},f^{\left(2\right)}\right)$,
which is a composition of several shallow neural networks. Afterwards,
we will show that{\small{}
\[
z_{i}\left(\boldsymbol{f}^{\left(1\right)},f^{\left(2\right)}\right)\in\left\{ z_{i}\left(\boldsymbol{\theta}\right):\boldsymbol{\theta}\in\Theta_{d_{h},\bar{z}}\right\} .
\]
}{\small\par}

Specifically, let $\mathcal{F}_{\left\{ 2k,r,m\right\} }^{\text{shallow}}$
be the set of shallow neural networks with $2k$ inputs, $r$ hidden
nodes, a single output, and sup norm upper bounded by $m$, such that

{\small{}
\begin{eqnarray*}
\mathcal{F}_{\left\{ 2k,r,m\right\} }^{\text{shallow}} & \coloneqq & \left\{ f:\mathbb{R}{}^{2k}\rightarrow\mathbb{R},f\left(\boldsymbol{x}_{1},\boldsymbol{x}_{2}\right)=\boldsymbol{c}\cdot\mathbf{\bm{\sigma}}\left(\boldsymbol{C}\boldsymbol{x}_{1}+\boldsymbol{C}_{\mathcal{N}}\boldsymbol{x}_{2}+\boldsymbol{g}\right)+b,\right.\\
 &  & \left.\boldsymbol{c},\boldsymbol{g}\in\mathbb{R}^{r},\boldsymbol{C},\boldsymbol{C}_{\mathcal{N}}\in\mathbb{R}^{r\times k},b\in\mathbb{R},\left\Vert f\right\Vert _{\infty,\mathbb{R}{}^{2k}}\leq m\right\} .
\end{eqnarray*}
}Let $d_{h}^{(1)}$ and $d_{h}^{(2)}$ be positive integers and $d_{h}^{(1)}\geq d_{h*}^{(1)}$.
For any $\boldsymbol{f}^{(1)}=\left(f_{1}^{(1)},\ldots,f_{d_{h*}^{(1)}}^{(1)}\right)^{\prime}$
with $f_{k}^{(1)}\in\mathcal{F}_{\left\{ 2d,d_{h}^{(1)}/d_{h*}^{(1)},\bar{z}\right\} }^{\text{shallow}}$
for every $k\in\left[d_{h*}^{(1)}\right]$,\footnote{For more concise illustration, we present the case that $d_{h}^{(1)}$
is a multiple of $d_{h*}^{(1)}$. Otherwise, we could set $f_{k}^{(1)}\in\mathcal{F}_{\{2d,\left\lfloor d_{h}^{(1)}/d_{h*}^{(1)}\right\rfloor ,2M\}}^{\text{shallow}}$
for $k=1,...,d_{h*}^{(1)}-1$ and $f_{k}^{(1)}\in\mathcal{F}_{\{2d,d_{h}^{(1)}-\left\lfloor d_{h}^{(1)}/d_{h*}^{(1)}\right\rfloor (d_{h*}^{(1)}-1),2M\}}^{\text{shallow}}$
for $k=d_{h*}^{(1)}$. The same conclusion as in (\ref{pf:approx3})
holds without requiring $d_{h}^{(1)}$ to be a multiple of $d_{h*}^{(1)}$.} and $f^{(2)}\in\mathcal{F}_{\left\{ 2d_{h*}^{(1)},d_{h}^{(2)},\bar{z}\right\} }^{\text{shallow}}$,
we define $z_{i}\left(\boldsymbol{f}^{(1)},f^{(2)}\right)$ as

{\small{}
\begin{align}
z_{i}\left(\boldsymbol{f}^{(1)},f^{(2)}\right) & =f^{(2)}\left(\boldsymbol{h}_{i}^{(1)}\left(\boldsymbol{f}^{(1)}\right),\overline{\boldsymbol{h}}_{\mathcal{N}\left(i\right)}^{(1)}\left(\boldsymbol{f}^{(1)}\right)\right)\nonumber \\
 & =\boldsymbol{c}^{(2)}\cdot\mathbf{\bm{\sigma}}\left(\boldsymbol{C}^{(2)}\boldsymbol{h}_{i}^{(1)}\left(\boldsymbol{f}^{(1)}\right)+\boldsymbol{C}_{\mathcal{N}}^{(2)}\overline{\boldsymbol{h}}_{\mathcal{N}\left(i\right)}^{(1)}\left(\boldsymbol{f}^{(1)}\right)+\boldsymbol{g}^{(2)}\right)+b^{(2)},\label{pf:approx_zi}
\end{align}
\begin{align}
\boldsymbol{h}_{i}^{(1)}\left(\boldsymbol{f}^{(1)}\right) & =\boldsymbol{f}^{(1)}\left(\boldsymbol{x}_{i},\bar{\boldsymbol{x}}_{\mathcal{N}\left(i\right)}\right)=\left[\begin{array}{c}
f_{1}^{(1)}\left(\boldsymbol{x}_{i},\overline{\boldsymbol{x}}_{\mathcal{N}\left(i\right)}\right)\\
\vdots\\
f_{d_{h*}^{(1)}}^{(1)}\left(\boldsymbol{x}_{i},\overline{\boldsymbol{x}}_{\mathcal{N}\left(i\right)}\right)
\end{array}\right]\nonumber \\
 & =\left[\begin{array}{c}
\boldsymbol{c}_{1}^{(1)}\cdot\mathbf{\bm{\sigma}}\left(\boldsymbol{C}_{1}^{(1)}\boldsymbol{x}_{i}+\boldsymbol{C}_{\mathcal{N}1}^{(1)}\overline{\boldsymbol{x}}_{\mathcal{N}\left(i\right)}+\boldsymbol{g}_{1}^{(1)}\right)+b_{1}^{(1)}\\
\vdots\\
\boldsymbol{c}_{d_{h*}^{(1)}}^{(1)}\cdot\mathbf{\bm{\sigma}}\left(\boldsymbol{C}_{d_{h*}^{(1)}}^{(1)}\boldsymbol{x}_{i}+\boldsymbol{C}_{\mathcal{N}d_{h*}^{(1)}}^{(1)}\overline{\boldsymbol{x}}_{\mathcal{N}\left(i\right)}+\boldsymbol{g}_{d_{h*}^{(1)}}^{(1)}\right)+b_{d_{h*}^{(1)}}^{(1)}
\end{array}\right]\nonumber \\
 & =\boldsymbol{B}^{(1)}\mathbf{\bm{\sigma}}\left(\boldsymbol{C}^{(1)}\boldsymbol{x}_{i}+\boldsymbol{C}_{\mathcal{N}}^{(1)}\bar{\boldsymbol{x}}_{\mathcal{N}\left(i\right)}+\boldsymbol{g}^{(1)}\right)+\boldsymbol{b}^{(1)},\label{pf:approx_hi}
\end{align}
\begin{equation}
\overline{\boldsymbol{h}}_{\mathcal{N}\left(i\right)}^{(1)}\left(\boldsymbol{f}^{(1)}\right)=\mathbb{I}\left\{ \mathcal{N}\left(i\right)\neq\textrm{Ø}\right\} \frac{1}{\left|\mathcal{N}\left(i\right)\right|}\sum_{j\in\mathcal{N}\left(i\right)}\boldsymbol{h}_{j}^{(1)}\left(\boldsymbol{f}^{(1)}\right),\label{pf:approx_hNi}
\end{equation}
}with the following notations {\footnotesize{}
\begin{align*}
\boldsymbol{B}^{(1)} & =\left[\begin{array}{ccccc}
\boldsymbol{c}_{1}^{(1)\prime} & 0 & 0 & \cdots & 0\\
0 & \boldsymbol{c}_{2}^{(1)\prime} & 0 & \ddots & 0\\
0 & 0 & \boldsymbol{c}_{3}^{(1)\prime} & \ddots & \vdots\\
\vdots & \vdots & \ddots & \ddots & 0\\
0 & 0 & \cdots & 0 & \boldsymbol{c}_{d_{h*}^{(1)}}^{(1)\prime}
\end{array}\right],\boldsymbol{C}^{(1)}=\left[\begin{array}{c}
\boldsymbol{C}_{1}^{(1)}\\
\vdots\\
\boldsymbol{C}_{d_{h*}^{(1)}}^{(1)}
\end{array}\right],\boldsymbol{C}_{\mathcal{N}}^{(1)}=\left[\begin{array}{c}
\boldsymbol{C}_{\mathcal{N}1}^{(1)}\\
\vdots\\
\boldsymbol{C}_{\mathcal{N}d_{h*}^{(1)}}^{(1)}
\end{array}\right],\boldsymbol{g}^{(1)}=\left[\begin{array}{c}
\boldsymbol{g}_{1}^{(1)}\\
\vdots\\
\boldsymbol{g}_{d_{h*}^{(1)}}^{(1)}
\end{array}\right],\boldsymbol{b}^{(1)}=\left[\begin{array}{c}
b_{1}^{(1)}\\
\vdots\\
b_{d_{h*}^{(1)}}^{(1)}
\end{array}\right].
\end{align*}
}Second, we can easily relabel terms to show that $z_{i}\left(\boldsymbol{f}^{(1)},f^{(2)}\right)\in\left\{ z_{i}\left(\boldsymbol{\theta}\right):\boldsymbol{\theta}\in\Theta_{d_{h},\bar{z}}\right\} $.
Note that combining (\ref{pf:approx_zi})-(\ref{pf:approx_hNi}),
every $z_{i}\left(\boldsymbol{f}^{(1)},f^{(2)}\right)$ can be expressed
as{\footnotesize{}
\begin{eqnarray*}
z_{i}\left(\boldsymbol{f}^{(1)},f^{(2)}\right) & = & \boldsymbol{c}^{(2)}\cdot\mathbf{\bm{\sigma}}\left(\boldsymbol{C}^{(2)}\boldsymbol{B}^{(1)}\boldsymbol{h}_{i}^{(1)}+\boldsymbol{C}_{\mathcal{N}}^{(2)}\boldsymbol{B}^{(1)}\overline{\boldsymbol{h}}_{\mathcal{N}\left(i\right)}^{(1)}+\boldsymbol{C}^{(2)}\boldsymbol{b}^{(1)}+\boldsymbol{C}_{\mathcal{N}}^{(2)}\boldsymbol{b}^{(1)}\mathbb{I}\left\{ \mathcal{N}\left(i\right)\neq\textrm{Ø}\right\} +\boldsymbol{g}^{(2)}\right)+b^{(2)},\\
\boldsymbol{h}_{i}^{(1)} & = & \mathbf{\bm{\sigma}}\left(\boldsymbol{C}^{(1)}\boldsymbol{x}_{i}+\boldsymbol{C}_{\mathcal{N}}^{(1)}\bar{\boldsymbol{x}}_{\mathcal{N}\left(i\right)}+\boldsymbol{g}^{(1)}\right),\\
\overline{\boldsymbol{h}}_{\mathcal{N}\left(i\right)}^{(1)} & = & \mathbb{I}\left\{ \mathcal{N}\left(i\right)\neq\textrm{Ø}\right\} \frac{1}{\left|\mathcal{N}\left(i\right)\right|}\sum_{j\in\mathcal{N}\left(i\right)}\boldsymbol{h}_{j}^{(1)},
\end{eqnarray*}
}where $\boldsymbol{C}^{(2)}\boldsymbol{B}^{(1)},\boldsymbol{C}_{\mathcal{N}}^{(2)}\boldsymbol{B}^{(1)}\in\mathbb{R}^{d_{h}^{(2)}\times d_{h}^{(1)}}$,
$\boldsymbol{C}^{(1)},\boldsymbol{C}_{\mathcal{N}}^{(1)}\in\mathbb{R}^{d_{h}^{(1)}\times d}$,
and as $f^{(2)}\in\mathcal{F}_{\left\{ 2d_{h*}^{(1)},d_{h}^{(2)},\bar{z}\right\} }^{\text{shallow}}$,
we obtain{\small{}
\[
\left\Vert z_{i}\left(\boldsymbol{f}^{(1)},f^{(2)}\right)\right\Vert {}_{\infty}\coloneqq\underset{\boldsymbol{\xi}_{i}}{\text{sup}}\left|z_{i}\left(\boldsymbol{f}^{(1)},f^{(2)}\right)\right|\leq\bar{z}.
\]
}Given the definition of $z_{i}\left(\boldsymbol{\theta}\right)$
in (\ref{pf:approx_def_zi}), we have shown that{\small{}
\begin{align*}
 & \left\{ z_{i}\left(\boldsymbol{f}^{(1)},f^{(2)}\right):f_{1}^{(1)},...,f_{d_{h*}^{(1)}}^{(1)}\in\mathcal{F}_{\left\{ 2d,d_{h}^{(1)}/d_{h*}^{(1)},\bar{z}\right\} }^{\text{shallow}},f^{(2)}\in\mathcal{F}_{\left\{ 2d_{h*}^{(1)},d_{h}^{(2)},\bar{z}\right\} }^{\text{shallow}}\right\} \subseteq\left\{ z_{i}\left(\boldsymbol{\theta}\right):\boldsymbol{\theta}\in\Theta_{d_{h},\bar{z}}\right\} .
\end{align*}
}This implies that the approximation error is bounded as {\small{}
\begin{align}
\left[\epsilon\left(\left\{ d_{h}^{(1)},d_{h}^{(2)}\right\} \right)\right]^{2}\leq & \underset{\substack{f_{1}^{(1)},...,f_{d_{h*}^{(1)}}^{(1)}\in\mathcal{F}_{\left\{ 2d,d_{h}^{(1)}/d_{h*}^{(1)},\bar{z}\right\} }^{\text{shallow}}\\
f^{(2)}\in\mathcal{F}_{\left\{ 2d_{h*}^{(1)},d_{h}^{(2)},\bar{z}\right\} }^{\text{shallow}}
}
}{\min}\mathbb{E}\left[\left(z_{i}\left(\boldsymbol{f}^{(1)},f^{(2)}\right)-z_{*i}\left(\boldsymbol{f}_{*}\right)\right)^{2}\right].\label{pf:approx1}
\end{align}
}Third, we decompose the approximation error $\epsilon\left(\left\{ d_{h}^{(1)},d_{h}^{(2)}\right\} \right)$.
For each $i\in\left[n\right]$ with any possible local neighborhood
$\boldsymbol{\xi}_{i}$, and for every $f_{1}^{(1)},...,f_{d_{h*}^{(1)}}^{(1)}\in\mathcal{F}_{\left\{ 2d,d_{h}^{(1)}/d_{h*}^{(1)},\bar{z}\right\} }^{\text{shallow}}$
and $f^{(2)}\in\mathcal{F}_{\left\{ 2d_{h*}^{(1)},d_{h}^{(2)},\bar{z}\right\} }^{\text{shallow}}$,
we obtain{\small{}
\begin{align*}
 & \left|z_{*i}\left(\boldsymbol{f}_{*}\right)-z_{i}\left(\boldsymbol{f}^{(1)},f^{(2)}\right)\right|\\
= & \left|f_{*}^{(2)}\left(\boldsymbol{f}_{*}^{(1)}\left(\boldsymbol{x}_{i},\bar{\boldsymbol{x}}_{\mathcal{N}\left(i\right)}\right),\mathbb{I}\left\{ \mathcal{N}\left(i\right)\neq\textrm{Ø}\right\} \frac{1}{\left|\mathcal{N}\left(i\right)\right|}\sum_{j\in\mathcal{N}\left(i\right)}\boldsymbol{f}_{*}^{(1)}\left(\boldsymbol{x}_{j},\bar{\boldsymbol{x}}_{\mathcal{N}\left(j\right)}\right)\right)\right.\\
 & \left.-f^{(2)}\left(\boldsymbol{f}^{(1)}\left(\boldsymbol{x}_{i},\bar{\boldsymbol{x}}_{\mathcal{N}\left(i\right)}\right),\mathbb{I}\left\{ \mathcal{N}\left(i\right)\neq\textrm{Ø}\right\} \frac{1}{\left|\mathcal{N}\left(i\right)\right|}\sum_{j\in\mathcal{N}\left(i\right)}\boldsymbol{f}^{(1)}\left(\boldsymbol{x}_{j},\bar{\boldsymbol{x}}_{\mathcal{N}\left(j\right)}\right)\right)\right|\\
\leq & \left|\tilde{f}_{*}^{(2)}\left(\boldsymbol{f}_{*}^{(1)}\left(\boldsymbol{x}_{i},\bar{\boldsymbol{x}}_{\mathcal{N}\left(i\right)}\right),\mathbb{I}\left\{ \mathcal{N}\left(i\right)\neq\textrm{Ø}\right\} \frac{1}{\left|\mathcal{N}\left(i\right)\right|}\sum_{j\in\mathcal{N}\left(i\right)}\boldsymbol{f}_{*}^{(1)}\left(\boldsymbol{x}_{j},\bar{\boldsymbol{x}}_{\mathcal{N}\left(j\right)}\right)\right)\right.\\
 & \left.-\text{\ensuremath{\tilde{f}}}_{*}^{(2)}\left(\boldsymbol{f}^{(1)}\left(\boldsymbol{x}_{i},\bar{\boldsymbol{x}}_{\mathcal{N}\left(i\right)}\right),\mathbb{I}\left\{ \mathcal{N}\left(i\right)\neq\textrm{Ø}\right\} \frac{1}{\left|\mathcal{N}\left(i\right)\right|}\sum_{j\in\mathcal{N}\left(i\right)}\boldsymbol{f}^{(1)}\left(\boldsymbol{x}_{j},\bar{\boldsymbol{x}}_{\mathcal{N}\left(j\right)}\right)\right)\right|\\
 & +\left|\tilde{f}_{*}^{(2)}\left(\bar{z}\frac{1}{\bar{z}}\boldsymbol{f}^{(1)}\left(\boldsymbol{x}_{i},\bar{\boldsymbol{x}}_{\mathcal{N}\left(i\right)}\right),\bar{z}\mathbb{I}\left\{ \mathcal{N}\left(i\right)\neq\textrm{Ø}\right\} \frac{1}{\left|\mathcal{N}(i)\right|}\sum_{j\in\mathcal{N}\left(i\right)}\frac{1}{\bar{z}}\boldsymbol{f}^{(1)}\left(\boldsymbol{x}_{j},\bar{\boldsymbol{x}}_{\mathcal{N}\left(j\right)}\right)\right)\right.\\
 & \left.-f^{(2)}\left(\bar{z}\frac{1}{\bar{z}}\boldsymbol{f}^{(1)}\left(\boldsymbol{x}_{i},\bar{\boldsymbol{x}}_{\mathcal{N}\left(i\right)}\right),\bar{z}\mathbb{I}\left\{ \mathcal{N}\left(i\right)\neq\textrm{Ø}\right\} \frac{1}{\left|\mathcal{N}(i)\right|}\sum_{j\in\mathcal{N}\left(i\right)}\frac{1}{\bar{z}}\boldsymbol{f}^{(1)}\left(\boldsymbol{x}_{j},\bar{\boldsymbol{x}}_{\mathcal{N}\left(j\right)}\right)\right)\right|\\
\leq & 2\eta\sum_{k=1}^{d_{h*}^{(1)}}\left|f_{*k}^{(1)}\left(\boldsymbol{x}_{i},\bar{\boldsymbol{x}}_{\mathcal{N}\left(i\right)}\right)-f_{k}^{(1)}\left(\boldsymbol{x}_{i},\bar{\boldsymbol{x}}_{\mathcal{N}\left(i\right)}\right)\right|\\
 & +2\eta\sum_{k=1}^{d_{h*}^{(1)}}\mathbb{I}\left\{ \mathcal{N}\left(i\right)\neq\textrm{Ø}\right\} \frac{1}{\left|\mathcal{N}\left(i\right)\right|}\left|\sum_{j\in\mathcal{N}\left(i\right)}f_{*k}^{(1)}\left(\boldsymbol{x}_{j},\bar{\boldsymbol{x}}_{\mathcal{N}(j)}\right)-\sum_{j\in\mathcal{N}\left(i\right)}f_{k}^{(1)}\left(\boldsymbol{x}_{j},\bar{\boldsymbol{x}}_{\mathcal{N}\left(j\right)}\right)\right|\\
 & +\left|\tilde{\tilde{f}}_{*}^{(2)}\left(\frac{1}{\bar{z}}\boldsymbol{f}^{(1)}\left(\boldsymbol{x}_{i},\bar{\boldsymbol{x}}_{\mathcal{N}\left(i\right)}\right),\mathbb{I}\left\{ \left|\mathcal{N}\left(i\right)\right|>0\right\} \frac{1}{\left|\mathcal{N}\left(i\right)\right|}\sum_{j\in\mathcal{N}\left(i\right)}\frac{1}{\bar{z}}\boldsymbol{f}^{(1)}\left(\boldsymbol{x}_{j},\bar{\boldsymbol{x}}_{\mathcal{N}\left(j\right)}\right)\right)\right.\\
 & \left.-\tilde{\tilde{f}}^{(2)}\left(\frac{1}{\bar{z}}\boldsymbol{f}^{(1)}\left(\boldsymbol{x}_{i},\bar{\boldsymbol{x}}_{\mathcal{N}\left(i\right)}\right),\mathbb{I}\left\{ \left|\mathcal{N}\left(i\right)\right|>0\right\} \frac{1}{\left|\mathcal{N}\left(i\right)\right|}\sum_{j\in\mathcal{N}\left(i\right)}\frac{1}{\bar{z}}\boldsymbol{f}^{(1)}\left(\boldsymbol{x}_{j},\bar{\boldsymbol{x}}_{\mathcal{N}\left(j\right)}\right)\right)\right|,
\end{align*}
}where the first inequality is feasible if we introduce a new function
$\tilde{f}_{*}^{(2)}$. We set\\
$\tilde{f}_{*}^{(2)}:\left[-\bar{z},\bar{z}\right]^{2d_{h*}^{(1)}}\rightarrow\left[-1.1M,1.1M\right]$,
which extends the domain of $f_{*}^{(2)}$ from the $2d_{h*}^{(1)}$-dimensional
cube $\left[-1,1\right]^{2d_{h*}^{(1)}}$ to a larger set as $f_{j}^{(1)}\left(\boldsymbol{x}_{i},\bar{\boldsymbol{x}}_{\mathcal{N}\left(i\right)}\right)\in\left[-\bar{z},\bar{z}\right]$
for $j\in\left[d_{h*}^{(1)}\right]$, and extends the codomain slightly.
We restrict $\tilde{f}_{*}^{(2)}$ to be $\tilde{f}_{*}^{(2)}\left(\boldsymbol{x}\right)=f_{*}^{(2)}\left(\boldsymbol{x}\right)$
if $\boldsymbol{x}\in\left[-1,1\right]^{2d_{h*}^{(1)}}$, and $\tilde{f}_{*}^{(2)}\in\mathcal{W}_{2\eta}^{\beta,\infty}\left(\left[-\bar{z},\bar{z}\right]^{2d_{h*}^{(1)}}\right)$
with $\beta\geq1$ over the extended domain, which is feasible under
Assumption I 1. This leads to the next line of the proof by the mean
value theorem. The last inequality holds by defining $\tilde{\tilde{f}}_{*}^{(2)}:\left[-1,1\right]^{2d_{h*}^{(1)}}\rightarrow\left[-1.1M,1.1M\right]$
and $\tilde{\tilde{f}}^{(2)}:\left[-1,1\right]^{2d_{h*}^{(1)}}\rightarrow\left[-\bar{z},\bar{z}\right]$,
where $\tilde{\tilde{f}}_{*}^{(2)}\left(\boldsymbol{x},\boldsymbol{y}\right)=\tilde{f}_{*}^{(2)}\left(\bar{z}\boldsymbol{x},\bar{z}\boldsymbol{y}\right)$
and $\tilde{\tilde{f}}^{(2)}\left(\boldsymbol{x},\boldsymbol{y}\right)=f^{(2)}\left(\bar{z}\boldsymbol{x},\bar{z}\boldsymbol{y}\right)$
for any $\boldsymbol{x},\boldsymbol{y}\in\left[-1,1\right]^{d_{h*}^{(1)}}.$
It is easy to show that $\tilde{\tilde{f}}_{*}^{(2)}\in\mathcal{W}_{2\eta(\bar{z})^{\beta}}^{\beta,\infty}\left(\left[-1,1\right]^{2d_{h*}^{(1)}}\right)$.

For every $i\in\left[n\right]$, and every $\boldsymbol{f}^{(1)}$
and $f^{(2)}$, denoting{\small{}
\begin{align*}
\boldsymbol{a} & =\frac{1}{\bar{z}}\boldsymbol{f}^{(1)}\left(\boldsymbol{x}_{i},\bar{\boldsymbol{x}}_{\mathcal{N}\left(i\right)}\right),\text{and}\ \boldsymbol{b}=\mathbb{I}\left\{ \left|\mathcal{N}\left(i\right)\right|>0\right\} \frac{1}{\left|\mathcal{N}\left(i\right)\right|}\sum_{j\in\mathcal{N}\left(i\right)}\frac{1}{\bar{z}}\boldsymbol{f}^{(1)}\left(\boldsymbol{x}_{j},\bar{\boldsymbol{x}}_{\mathcal{N}\left(j\right)}\right),
\end{align*}
}then it holds {\small{}
\begin{align}
 & \mathbb{E}\left[\left(z_{*i}\left(\boldsymbol{f}_{*}\right)-z_{i}\left(\boldsymbol{f}^{(1)},f^{(2)}\right)\right)^{2}\right]\nonumber \\
\lesssim & \sum_{k=1}^{d_{h*}^{(1)}}\mathbb{E}\left[\left(f_{*k}^{(1)}\left(\boldsymbol{x}_{i},\bar{\boldsymbol{x}}_{\mathcal{N}\left(i\right)}\right)-f_{k}^{(1)}\left(\boldsymbol{x}_{i},\bar{\boldsymbol{x}}_{\mathcal{N}\left(i\right)}\right)\right)^{2}\right]\nonumber \\
+ & \sum_{k=1}^{d_{h*}^{(1)}}\mathbb{E}\left[\mathbb{I}\left\{ \mathcal{N}\left(i\right)\neq\textrm{Ø}\right\} \frac{1}{\left|\mathcal{N}\left(i\right)\right|}\sum_{j\in\mathcal{N}\left(i\right)}\left(f_{*k}^{(1)}\left(\boldsymbol{x}_{j},\bar{\boldsymbol{x}}_{\mathcal{N}\left(j\right)}\right)-f_{k}^{(1)}\left(\boldsymbol{x}_{j},\bar{\boldsymbol{x}}_{\mathcal{N}\left(j\right)}\right)\right)^{2}\right]\nonumber \\
+ & \mathbb{E}\left[\left(\tilde{\tilde{f}}_{*}^{(2)}\left(\boldsymbol{a},\boldsymbol{b}\right)-\tilde{\tilde{f}}^{(2)}\left(\boldsymbol{a},\boldsymbol{b}\right)\right)^{2}\right]\nonumber \\
\lesssim & \sum_{k=1}^{d_{h*}^{(1)}}\sup_{\boldsymbol{x},\boldsymbol{y}\in\left[-1,1\right]^{d}}\left(f_{*k}^{(1)}\left(\boldsymbol{x},\boldsymbol{y}\right)-f_{k}^{(1)}\left(\boldsymbol{x},\boldsymbol{y}\right)\right)^{2}+\sup_{\boldsymbol{x},\boldsymbol{y}\in\left[-1,1\right]^{d_{h*}^{(1)}}}\left(\tilde{\tilde{f}}_{*}^{(2)}\left(\boldsymbol{x},\boldsymbol{y}\right)-\tilde{\tilde{f}}^{(2)}\left(\boldsymbol{x},\boldsymbol{y}\right)\right)^{2}.\label{pf:approx2}
\end{align}
}{\small\par}

Finally, combining (\ref{pf:approx1}) and (\ref{pf:approx2}), the
conclusion is obtained as{\small{}
\begin{align}
\left[\epsilon\left(\left\{ d_{h}^{(1)},d_{h}^{(2)}\right\} \right)\right]^{2}\lesssim & \sum_{k=1}^{d_{h*}^{(1)}}\min_{f_{k}^{(1)}\in\mathcal{F}_{\left\{ 2d,d_{h}^{(1)}/d_{h*}^{(1)},\bar{z}\right\} }^{\text{shallow}}}\sup_{\boldsymbol{x},\boldsymbol{y}\in\left[-1,1\right]^{d}}\left|f_{*k}^{(1)}\left(\boldsymbol{x},\boldsymbol{y}\right)-f_{k}^{(1)}\left(\boldsymbol{x},\boldsymbol{y}\right)\right|^{2}\nonumber \\
 & +\underset{\substack{f^{(2)}\in\mathcal{F}_{\left\{ 2d_{h*}^{(1)},d_{h}^{(2)},\bar{z}\right\} }^{\text{shallow}}}
}{\min}\sup_{\boldsymbol{x},\boldsymbol{y}\in\left[-1,1\right]^{d_{h*}^{(1)}}}\left|\tilde{\tilde{f}}_{*}^{(2)}\left(\boldsymbol{x},\boldsymbol{y}\right)-\tilde{\tilde{f}}^{(2)}\left(\boldsymbol{x},\boldsymbol{y}\right)\right|^{2}\nonumber \\
\lesssim & \left(\text{\ensuremath{d_{h}^{(1)}}}\right){}^{-\beta/d}+\left(d_{h}^{(2)}\right)^{-\beta/d_{h*}^{(1)}},\label{pf:approx3}
\end{align}
}where the last inequality holds by Lemma \ref{lem: approximation_infinitely_differentiable}
and Assumptions I 1 and 2. 

For simplicity, set $d_{h}^{\left(l\right)}\asymp d_{h}$ for $l\in\left[L\right]$
and hence the approximation error satisfies{\small{}
\begin{equation}
\left[\epsilon\left(d_{h}\right)\right]^{2}\lesssim\left(\text{\ensuremath{d_{h}}}\right)^{-\beta/\left(d\vee d_{h*}^{\left(1\right)}\vee...\vee d_{h*}^{\left(L\right)}\right)}.\label{eq: approximation_error_rate}
\end{equation}
}{\small\par}

\subsection{Step 4: revisit the main decomposition \label{subsec: Step-4 revist the main decomposition}}

Combine the estimation and approximation errors in (\ref{eq: estimation_error_rate})
and (\ref{eq: approximation_error_rate}). For $k=2,4$ or $6$ and
$s=L$ or $2L-1$ ($k$ and $s$ depend on the activation function
and number of layers), with probability at least $1-\exp\left(-\rho\right)$,
{\footnotesize{}
\begin{align}
\mathbb{E}\left[\left(z_{i}\left(\hat{\boldsymbol{\theta}}\right)-z_{*i}\left(f_{*}\right)\right)^{2}\right] & \leq C\cdot\left(\frac{1}{n}\sum_{j=1}^{J}\left(1+\log\left|\mathcal{C}_{j}\right|\right)\cdot\left(c_{n}\right)^{s}\cdot\left(d_{h}\right)^{k}+\frac{J\log J+J\rho}{n}+\left(\text{\ensuremath{d_{h}}}\right)^{-\beta/d_{h*}}\right)\label{eq:rate_new_theorem3.3_2}\\
 & \leq C\cdot\left(\left(\frac{1}{n}\sum_{j=1}^{J}\left(1+\log\left|\mathcal{C}_{j}\right|\right)\cdot\left(c_{n}\right)^{s}\right)^{\frac{\beta}{\beta+kd_{h*}}}+\frac{J\log J+J\rho}{n}\right),\label{eq:rate_new_theorem3.3_3}
\end{align}
}where $\beta$ is the smoothness parameter, $d_{h*}\coloneqq d\vee d_{h*}^{(1)}\vee...\vee d_{h*}^{(L)}$
with $d$ being the number of covariates and $d_{h*}^{(l)}$ being
the number of latent variables in each layer, $J$ is the number of
covers in the smallest proper cover, $\left|\mathcal{C}_{j}\right|$
is the size of each cover, and $c_{n}$ is the the largest number
of peers an observation could have. The second inequality holds by
setting $d_{h}\asymp\left(\frac{1}{n}\sum_{j=1}^{J}\left(1+\log\left|\mathcal{C}_{j}\right|\right)\cdot\left(c_{n}\right)^{s}\right)^{-\frac{d_{h*}}{\beta+kd_{h*}}}$.

We next bound the empirical $L_{2}$ norm. Applying the localization
analysis again (using the second part of Lemma \ref{lem: Bartlett Theorem 3.3}
in particular) to the set of functions\\
${\cal H}=\left\{ \left(z_{i}\left(\boldsymbol{\theta}\right)-z_{*i}\left(\boldsymbol{f}_{*}\right)\right)^{2}:\boldsymbol{\theta}\in\Theta_{d_{h},\bar{z}}\right\} $,
it is easy to show that for every $j\in\left[J\right]$, with probability
at least $1-\exp\left(-\log J-\rho\right)$, {\footnotesize{}
\[
\frac{1}{\left|\mathcal{C}_{j}\right|}\sum_{i\in\mathcal{C}_{j}}\left(z_{i}\left(\hat{\boldsymbol{\theta}}\right)-z_{*i}\left(\boldsymbol{f}_{*}\right)\right)^{2}\leq C\left(\mathbb{E}\left[\left(z_{i}\left(\hat{\boldsymbol{\theta}}\right)-z_{*i}\left(\boldsymbol{f}_{*}\right)\right)^{2}\right]+\frac{1+\log\left|\mathcal{C}_{j}\right|}{\left|\mathcal{C}_{j}\right|}\cdot\left(c_{n}\right)^{s}\cdot\left(d_{h}\right)^{k}+\frac{\log J+\rho}{\left|\mathcal{C}_{j}\right|}\right).
\]
}Then, it holds with probability at least $1-\exp\left(-\rho\right)$,
{\footnotesize{}
\begin{align*}
\frac{1}{n}\sum_{i=1}^{n}\left(z_{i}\left(\hat{\boldsymbol{\theta}}\right)-z_{*i}\left(\boldsymbol{f}_{*}\right)\right)^{2} & \leq C\left(\mathbb{E}\left[\left(z_{i}\left(\hat{\boldsymbol{\theta}}\right)-z_{*i}\left(\boldsymbol{f}_{*}\right)\right)^{2}\right]+\frac{1}{n}\sum_{j=1}^{J}\left(1+\log\left|\mathcal{C}_{j}\right|\right)\cdot\left(c_{n}\right)^{s}\cdot\left(d_{h}\right)^{k}+\frac{J\log J+J\rho}{n}\right).
\end{align*}
}Combining this with the inequalities in (\ref{eq:rate_new_theorem3.3_2})
and (\ref{eq:rate_new_theorem3.3_3}), it implies that with probability
at least $1-2\exp\left(-\rho\right),${\small{}
\[
\frac{1}{n}\sum_{i=1}^{n}\left(z_{i}\left(\hat{\boldsymbol{\theta}}\right)-z_{*i}\left(\boldsymbol{f}_{*}\right)\right)^{2}\leq C\cdot\left(\left(\frac{1}{n}\sum_{j=1}^{J}\left(1+\log\left|\mathcal{C}_{j}\right|\right)\cdot\left(c_{n}\right)^{s}\right)^{\frac{\beta}{\beta+kd_{h*}}}+\frac{J\log J+J\rho}{n}\right).
\]
}This completes the proof of Theorem \ref{thm: rate_of_convergence}.

\section{Proof of Corollary \ref{cor: asymptotic_distribution} \label{sec:Proof-of-Corollary}}

\subsection{Proof of the first part of Corollary \ref{cor: asymptotic_distribution}}

In this section, we first show that{\small{}
\begin{equation}
\sqrt{n}\left(\hat{\pi}\left(s\right)-\pi\left(s\right)\right)=\frac{1}{\sqrt{n}}\sum_{i=1}^{n}\left(\zeta_{i}-\mathbb{E}\left[\zeta_{i}\right]\right)+o_{p}\left(1\right).\label{eq: transfer from pi to xi}
\end{equation}
}This part is similar to the proof in \citet{farrell2015robust}.
It is easy to verify that $\mathbb{E}\left[\zeta_{i}\right]=\pi\left(s\right)$
under Assumptions II 5 and 6. Then, with the following decomposition,
{\small{}
\begin{eqnarray*}
\sqrt{n}\left(\hat{\pi}\left(s\right)-\pi\left(s\right)\right) & = & \sqrt{n}\left(\frac{1}{n}\sum_{i=1}^{n}\left(s_{1}\left(\boldsymbol{\xi}_{i}\right)\hat{\varphi}_{1}\left(\boldsymbol{\upsilon}{}_{i}\right)+s_{0}\left(\boldsymbol{\xi}_{i}\right)\hat{\varphi}_{0}\left(\boldsymbol{\upsilon}{}_{i}\right)\right)-\pi\left(s\right)\right)\\
 & = & \sqrt{n}\left(\frac{1}{n}\sum_{i=1}^{n}\zeta_{i}-\pi\left(s\right)\right)+\sum_{t\in\left\{ 0,1\right\} }\frac{1}{\sqrt{n}}\sum_{i=1}^{n}\left(s_{t}\left(\boldsymbol{\xi}_{i}\right)\hat{\varphi}_{t}\left(\boldsymbol{\upsilon}{}_{i}\right)-s_{t}\left(\boldsymbol{\xi}_{i}\right)\varphi_{t}\left(\boldsymbol{\upsilon}{}_{i}\right)\right),
\end{eqnarray*}
}we could establish (\ref{eq: transfer from pi to xi}) by showing
that for $t\in\left\{ 0,1\right\} $,{\small{}
\begin{align*}
\frac{1}{\sqrt{n}}\sum_{i=1}^{n}\left(s_{t}\left(\boldsymbol{\xi}_{i}\right)\hat{\varphi}_{t}\left(\boldsymbol{\upsilon}{}_{i}\right)-s_{t}\left(\boldsymbol{\xi}_{i}\right)\varphi_{t}\left(\boldsymbol{\upsilon}{}_{i}\right)\right) & =o_{p}\left(1\right).
\end{align*}
}Then, the decomposition below implies that it is sufficient to show
that $R_{1}$ and $R_{2}$ are $o_{p}\left(1\right)$.{\small{}
\begin{align*}
 & \frac{1}{\sqrt{n}}\sum_{i=1}^{n}\left(s_{t}\left(\boldsymbol{\xi}_{i}\right)\hat{\varphi}_{t}\left(\boldsymbol{\upsilon}{}_{i}\right)-s_{t}\left(\boldsymbol{\xi}_{i}\right)\varphi_{t}\left(\boldsymbol{\upsilon}{}_{i}\right)\right)\\
= & \frac{1}{\sqrt{n}}\sum_{i=1}^{n}\left(s_{t}\left(\boldsymbol{\xi}_{i}\right)\left(\frac{\mathbb{I}\left\{ t_{i}=t\right\} }{\hat{p}_{t}\left(\boldsymbol{\xi}_{i}\right)}\left(y_{i}\left(t\right)-\hat{\mu}_{t}\left(\boldsymbol{\xi}_{i}\right)\right)+\hat{\mu}_{t}\left(\boldsymbol{\xi}_{i}\right)\right)-s_{t}\left(\boldsymbol{\xi}_{i}\right)\left(\frac{\mathbb{I}\left\{ t_{i}=t\right\} }{p_{t}\left(\boldsymbol{\xi}_{i}\right)}\left(y_{i}\left(t\right)-\mu_{t}\left(\boldsymbol{\xi}_{i}\right)\right)+\mu_{t}\left(\boldsymbol{\xi}_{i}\right)\right)\right)\\
= & \frac{1}{\sqrt{n}}\sum_{i=1}^{n}s_{t}\left(\boldsymbol{\xi}_{i}\right)\mathbb{I}\left\{ t_{i}=t\right\} \left(y_{i}\left(t\right)-\mu_{t}\left(\boldsymbol{\xi}_{i}\right)\right)\left(\frac{p_{t}\left(\boldsymbol{\xi}_{i}\right)-\hat{p}_{t}\left(\boldsymbol{\xi}_{i}\right)}{\hat{p}_{t}\left(\boldsymbol{\xi}_{i}\right)p_{t}\left(\boldsymbol{\xi}_{i}\right)}\right)\\
+ & \frac{1}{\sqrt{n}}\sum_{i=1}^{n}s_{t}\left(\boldsymbol{\xi}_{i}\right)\left(\hat{\mu}_{t}\left(\boldsymbol{\xi}_{i}\right)-\mu_{t}\left(\boldsymbol{\xi}_{i}\right)\right)\left(1-\frac{\mathbb{I}\left\{ t_{i}=t\right\} }{\hat{p}_{t}\left(\boldsymbol{\xi}_{i}\right)}\right)\\
\eqqcolon & R_{1}+R_{2}.
\end{align*}
}Regarding $R_{1}$, to show $\mathbb{E}\left[R_{1}^{2}\right]=o\left(1\right)$,
we just need to show $\mathbb{E}\left[R_{1}^{2}\mid\left\{ \boldsymbol{\xi}_{i},t_{i}\right\} _{i\in\left[n\right]}\right]=o_{p}\left(1\right)$
since\\
$\mathbb{E}\left[R_{1}^{2}\mid\left\{ \boldsymbol{\xi}_{i},t_{i}\right\} _{i\in\left[n\right]}\right]<C$
a.s. is guaranteed by Assumptions II 2, 4, 7 and 8. Note the following
holds true for $t\in\left\{ 0,1\right\} $,

{\footnotesize{}
\begin{align*}
 & \mathbb{E}\left[R_{1}^{2}\mid\left\{ \boldsymbol{\xi}_{i},t_{i}\right\} _{i\in\left[n\right]}\right]\\
= & \frac{1}{n}\sum_{i=1}^{n}\sum_{j=1}^{n}s_{t}\left(\boldsymbol{\xi}_{i}\right)s_{t}\left(\boldsymbol{\xi}_{j}\right)\mathbb{I}\left\{ t_{i}=t\right\} \mathbb{I}\left\{ t_{j}=t\right\} \left(\frac{p_{t}\left(\boldsymbol{\xi}_{i}\right)-\hat{p}_{t}\left(\boldsymbol{\xi}_{i}\right)}{\hat{p}_{t}\left(\boldsymbol{\xi}_{i}\right)p_{t}\left(\boldsymbol{\xi}_{i}\right)}\right)\left(\frac{p_{t}\left(\boldsymbol{\xi}_{j}\right)-\hat{p}_{t}\left(\boldsymbol{\xi}_{j}\right)}{\hat{p}_{t}\left(\boldsymbol{\xi}_{j}\right)p_{t}\left(\boldsymbol{\xi}_{j}\right)}\right)\mathbb{E}\left[u_{i}\left(t\right)u_{j}\left(t\right)\mid\left\{ \boldsymbol{\xi}_{i},t_{i}\right\} _{i\in\left[n\right]}\right]\\
= & \frac{1}{n}\sum_{i=1}^{n}\left[\frac{s_{t}\left(\boldsymbol{\xi}_{i}\right)\mathbb{I}\left\{ t_{i}=t\right\} \mathbb{E}\left[\left(u_{i}\left(t\right)\right)^{2}\mid\left\{ \boldsymbol{\xi}_{i},t_{i}\right\} _{i\in\left[n\right]}\right]}{\left(\hat{p}_{t}\left(\boldsymbol{\xi}_{i}\right)p_{t}\left(\boldsymbol{\xi}_{i}\right)\right)^{2}}\right]\left(p_{t}\left(\boldsymbol{\xi}_{i}\right)-\hat{p}_{t}\left(\boldsymbol{\xi}_{i}\right)\right)^{2}\\
\leq & C\frac{1}{n}\sum_{i=1}^{n}\left(p_{t}\left(\boldsymbol{\xi}_{i}\right)-\hat{p}_{t}\left(\boldsymbol{\xi}_{i}\right)\right)^{2}=o_{p}\left(1\right),
\end{align*}
}where the second equality holds by Assumption II 7, the first inequality
holds by Assumptions II 2, 4 and 8, and the last equality holds under
Assumption II 3(a). Therefore, we have verified that $\mathbb{E}\left[R_{1}^{2}\right]\rightarrow0$,
which implies $R_{1}=o_{p}\left(1\right)$ by Markov's inequality.

Regarding $R_{2}$, consider the following decomposition{\small{}
\begin{align*}
R_{2} & =\frac{1}{\sqrt{n}}\sum_{i=1}^{n}s_{t}\left(\boldsymbol{\xi}_{i}\right)\left(\hat{\mu}_{t}\left(\boldsymbol{\xi}_{i}\right)-\mu_{t}\left(\boldsymbol{\xi}_{i}\right)\right)\left(1-\frac{\mathbb{I}\left\{ t_{i}=t\right\} }{\hat{p}_{t}\left(\boldsymbol{\xi}_{i}\right)}\right)\\
 & =\frac{1}{\sqrt{n}}\sum_{i=1}^{n}s_{t}\left(\boldsymbol{\xi}_{i}\right)\left(\hat{\mu}_{t}\left(\boldsymbol{\xi}_{i}\right)-\mu_{t}\left(\boldsymbol{\xi}_{i}\right)\right)\left(1-\frac{\mathbb{I}\left\{ t_{i}=t\right\} }{p_{t}\left(\boldsymbol{\xi}_{i}\right)}\right)\\
 & +\frac{1}{\sqrt{n}}\sum_{i=1}^{n}s_{t}\left(\boldsymbol{\xi}_{i}\right)\left(\hat{\mu}_{t}\left(\boldsymbol{\xi}_{i}\right)-\mu_{t}\left(\boldsymbol{\xi}_{i}\right)\right)\left(\frac{\left(\hat{p}_{t}\left(\boldsymbol{\xi}_{i}\right)-p_{t}\left(\boldsymbol{\xi}_{i}\right)\right)\mathbb{I}\left\{ t_{i}=t\right\} }{p_{t}\left(\boldsymbol{\xi}_{i}\right)\hat{p}_{t}\left(\boldsymbol{\xi}_{i}\right)}\right)\\
 & \eqqcolon R_{21}+R_{22},
\end{align*}
}where $R_{21}=o_{p}\left(1\right)$ under Assumption II 3(c). So
it is sufficient to show $R_{22}=o_{p}\left(1\right)$ as follows.{\small{}
\begin{align*}
\left|R_{22}\right| & \leq C\frac{1}{\sqrt{n}}\sum_{i=1}^{n}\left|\left(\hat{\mu}_{t}\left(\boldsymbol{\xi}_{i}\right)-\mu_{t}\left(\boldsymbol{\xi}_{i}\right)\right)\left(\hat{p}_{t}\left(\boldsymbol{\xi}_{i}\right)-p_{t}\left(\boldsymbol{\xi}_{i}\right)\right)\right|\\
 & \leq C\sqrt{n}\left(\frac{1}{n}\sum_{i=1}^{n}\left(\hat{\mu}_{t}\left(\boldsymbol{\xi}_{i}\right)-\mu_{t}\left(\boldsymbol{\xi}_{i}\right)\right)^{2}\right)^{1/2}\left(\frac{1}{n}\sum_{i=1}^{n}\left(\hat{p}_{t}\left(\boldsymbol{\xi}_{i}\right)-p_{t}\left(\boldsymbol{\xi}_{i}\right)\right)^{2}\right)^{1/2}\\
 & =o_{p}\left(1\right),
\end{align*}
}where the first inequality holds under Assumption II 4, and the equality
holds by Assumption II 3(b).

Hence, we have verified that $R_{1}$ are $R_{2}$ are $o_{p}\left(1\right)$,
which implies (\ref{eq: transfer from pi to xi}). Then, together
with Assumption II 9, it holds that{\small{}
\[
\left(\Sigma_{n}\right)^{-1/2}\sqrt{n}\left(\hat{\pi}\left(s\right)-\pi\left(s\right)\right)=\left(\Sigma_{n}\right)^{-1/2}\frac{1}{\sqrt{n}}\sum_{i=1}^{n}\left(\zeta_{i}-\mathbb{E}\left[\zeta_{i}\right]\right)+o_{p}\left(1\right).
\]
}Also, under Assumption II 10, the conditions of Lemma \ref{lem: Janson Theorem 2}
are satisfied and yield

{\small{}
\[
\left(\Sigma_{n}\right)^{-1/2}\frac{1}{\sqrt{n}}\sum_{i=1}^{n}\left(\zeta_{i}-\mathbb{E}\left[\zeta_{i}\right]\right)\overset{d}{\to}N\left(0,1\right),
\]
}which completes the proof of the first part of the corollary.

\subsection{Proof of the second part of Corollary \ref{cor: asymptotic_distribution}}

In this section, we first propose an infeasible estimator, $\tilde{\Sigma}_{n}$,
such that $\tilde{\Sigma}_{n}/n_{\text{avg}}-\Sigma_{n}/n_{\text{avg}}=o_{p}\left(1\right),$
where $n_{\text{avg}}=n/\bar{c}$ denotes the average cluster size.
This part is closely related to the proof in \citet{hansen2007asymptotic}.
Next, we provide a feasible estimator, $\hat{\Sigma}_{n}$, which
ensures $\hat{\Sigma}_{n}/n_{\text{avg}}-\tilde{\Sigma}_{n}/n_{\text{avg}}=o_{p}\left(1\right)$.
Then, the second part of the corollary holds under Assumption III
4.

Observe that the target parameter $\Sigma_{n}$ can be written as
{\small{}
\begin{align*}
\Sigma_{n} & =\text{Var}\left(\frac{1}{\sqrt{n}}\sum_{i=1}^{n}\zeta_{i}\right)=\frac{1}{n}\text{\ensuremath{\sum_{c\in\left[\bar{c}\right]}}\ensuremath{\mathbb{E}\left[\left(\sum_{j\in\left[n_{c}\right]}\zeta_{cj}-n_{c}\mu\right)^{2}\right]}}=\ensuremath{\frac{1}{\bar{c}}}\sum_{c\in\left[\bar{c}\right]}\ensuremath{w_{c}\mathbb{E}\left[\left(\bar{\zeta}_{c}-\mu\right)^{2}\right],}
\end{align*}
}where $\mu\coloneqq\mathbb{E}\left[\zeta_{cj}\right]$, $\bar{\zeta}_{c}\coloneqq\frac{1}{n_{c}}\sum_{j\in\left[n_{c}\right]}\zeta_{cj}$,
and $w_{c}\coloneqq n_{c}^{2}\frac{\bar{c}}{n}$. Hence, we introduce
an infeasible estimator correspondingly, {\small{}
\[
\tilde{\Sigma}_{n}=\ensuremath{\frac{1}{\bar{c}}}\text{\ensuremath{\sum_{c\in\left[\bar{c}\right]}}\ensuremath{\ensuremath{w_{c}\left(\bar{\zeta}_{c}-\bar{\zeta}\right)^{2}}},}
\]
}where $\bar{\zeta}=\frac{1}{n}\sum_{i=1}^{n}\zeta_{i}$. The difference
between them can be decomposed as{\small{}
\begin{eqnarray*}
\tilde{\Sigma}_{n}/n_{\text{avg}}-\Sigma_{n}/n_{\text{avg}} & = & \frac{1}{\bar{c}}\sum_{c\in\left[\bar{c}\right]}\ensuremath{\frac{w_{c}}{n_{\text{avg}}}\left(\left(\bar{\zeta}_{c}-\bar{\zeta}\right)^{2}-\mathbb{E}\left[\left(\bar{\zeta}_{c}-\mu\right)^{2}\right]\right)}\\
 & = & \frac{1}{\bar{c}}\sum_{c\in\left[\bar{c}\right]}\ensuremath{\frac{w_{c}}{n_{\text{avg}}}\left(\left(\left(\bar{\zeta}_{c}-\mu\right)-\left(\bar{\zeta}-\mu\right)\right)^{2}-\mathbb{E}\left[\left(\bar{\zeta}_{c}-\mu\right)^{2}\right]\right)}\\
 & = & \frac{1}{\bar{c}}\sum_{c\in\left[\bar{c}\right]}\left(\frac{w_{c}}{n_{\text{avg}}}\left(\bar{\zeta}_{c}-\mu\right)^{2}-\frac{w_{c}}{n_{\text{avg}}}\mathbb{E}\left[\left(\bar{\zeta}_{c}-\mu\right)^{2}\right]\right)\\
 &  & +\frac{1}{\bar{c}}\sum_{c\in\left[\bar{c}\right]}\ensuremath{\frac{w_{c}}{n_{\text{avg}}}\left(\bar{\zeta}-\mu\right)^{2}}-\frac{2}{\bar{c}}\sum_{c\in\left[\bar{c}\right]}\ensuremath{\left(\frac{w_{c}}{n_{\text{avg}}}\bar{\zeta}_{c}-\frac{w_{c}}{n_{\text{avg}}}\mu\right)\left(\bar{\zeta}-\mu\right).}
\end{eqnarray*}
}Therefore, to show $\tilde{\Sigma}_{n}/n_{\text{avg}}-\Sigma_{n}/n_{\text{avg}}=o_{p}\left(1\right)$,
it is sufficient to verify that {\small{}
\begin{align}
\frac{1}{\bar{c}}\sum_{c\in\left[\bar{c}\right]}\left(\frac{w_{c}}{n_{\text{avg}}}\left(\bar{\zeta}_{c}-\mu\right)^{2}-\frac{w_{c}}{n_{\text{avg}}}\mathbb{E}\left[\left(\bar{\zeta}_{c}-\mu\right)^{2}\right]\right) & =o_{p}\left(1\right),\label{eq:1-1}\\
\bar{\zeta}-\mu=\frac{1}{\bar{c}}\sum_{c\in\left[\bar{c}\right]}\left(\left(\frac{n_{c}}{n/\bar{c}}\right)\left(\bar{\zeta}_{c}\right)-\left(\frac{n_{c}}{n/\bar{c}}\right)\mu\right) & =o_{p}\left(1\right),\label{eq:2-1}\\
\frac{1}{\bar{c}}\sum_{c\in\left[\bar{c}\right]}\ensuremath{\frac{w_{c}}{n_{\text{avg}}}} & =O\left(1\right),\label{eq:3-1}\\
\frac{1}{\bar{c}}\sum_{c\in\left[\bar{c}\right]}\ensuremath{\left(\frac{w_{c}}{n_{\text{avg}}}\bar{\zeta}_{c}-\frac{w_{c}}{n_{\text{avg}}}\mu\right)} & =o_{p}\left(1\right).\label{eq:4-1}
\end{align}
}First, (\ref{eq:3-1}) holds true as under Assumption III 2,{\small{}
\[
\frac{1}{\bar{c}}\sum_{c\in\left[\bar{c}\right]}\ensuremath{\frac{w_{c}}{n_{\text{avg}}}}=\frac{1}{\bar{c}}\sum_{c\in\left[\bar{c}\right]}\left(\frac{n_{c}}{n/\bar{c}}\right)^{2}\leq\left(\frac{\max_{c\in\left[\bar{c}\right]}n_{c}}{n/\bar{c}}\right)^{2}=O\left(1\right).
\]
}For (\ref{eq:1-1}), (\ref{eq:2-1}) and (\ref{eq:4-1}), the conditions
of Lemma \ref{lem: Hansen 2.6.2} are satisfied as under Assumptions
III 2 and 3, Assumptions II 2 and 4, and Minkowski's inequality, it
holds {\small{}
\begin{align*}
\mathbb{E}\left[\left|\frac{w_{c}}{n_{\text{avg}}}\left(\bar{\zeta}_{c}-\mu\right)^{2}\right|^{1+\delta}\right] & =\left(\frac{n_{c}}{n/\bar{c}}\right)^{2+2\delta}\mathbb{E}\left[\left|\bar{\zeta}_{c}-\mu\right|^{2+2\delta}\right]<C<\infty,\\
\mathbb{E}\left[\left|\left(\frac{n_{c}}{n/\bar{c}}\right)\left(\bar{\zeta}_{c}\right)\right|^{1+\delta}\right] & =\left(\frac{n_{c}}{n/\bar{c}}\right)^{1+\delta}\mathbb{E}\left[\left|\bar{\zeta}_{c}\right|^{1+\delta}\right]<C<\infty,\\
\mathbb{E}\left[\left|\frac{w_{c}}{n_{\text{avg}}}\bar{\zeta}_{c}\right|^{1+\delta}\right] & =\left(\frac{n_{c}}{n/\bar{c}}\right)^{2+2\delta}\mathbb{E}\left[\left|\bar{\zeta}_{c}\right|^{1+\delta}\right]<C<\infty.
\end{align*}
}Hence, Lemma \ref{lem: Hansen 2.6.2} implies (\ref{eq:1-1}), (\ref{eq:2-1})
and (\ref{eq:4-1}). Therefore, we conclude that $\tilde{\Sigma}_{n}/n_{\text{avg}}-\Sigma_{n}/n_{\text{avg}}=o_{p}\left(1\right)$.

Next, we introduce a feasible estimator, $\hat{\Sigma}_{n}$, and
show that $\hat{\Sigma}_{n}/n_{\text{avg}}-\tilde{\Sigma}_{n}/n_{\text{avg}}=o_{p}\left(1\right)$.
In particular, define $\hat{\Sigma}_{n}/n_{\text{avg}}=\ensuremath{\frac{1}{\bar{c}}}\ensuremath{\sum_{c\in\left[\bar{c}\right]}}\frac{w_{c}}{n_{\text{avg}}}\left(\bar{\hat{\zeta}}_{c}-\bar{\hat{\zeta}}\right)^{2}$,
where $\bar{\hat{\zeta}}=\frac{1}{n}\sum_{i=1}^{n}\hat{\zeta}_{i}$
and $\bar{\hat{\zeta}}_{c}=\frac{1}{n_{c}}\sum_{j\in\left[n_{c}\right]}\hat{\zeta}_{cj}$.
Then, consider the following decomposition, {\small{}
\begin{align*}
\left|\frac{1}{n_{\text{avg}}}\hat{\Sigma}_{n}-\frac{1}{n_{\text{avg}}}\tilde{\Sigma}_{n}\right| & =\left|\ensuremath{\frac{1}{\bar{c}}}\text{\ensuremath{\sum_{c\in\left[\bar{c}\right]}}\ensuremath{\ensuremath{\frac{w_{c}}{n_{\text{avg}}}\left[\left(\bar{\hat{\zeta}}_{c}-\bar{\hat{\zeta}}\right)^{2}-\left(\bar{\zeta}_{c}-\bar{\zeta}\right)^{2}\right]}}}\right|\\
 & =\left|\ensuremath{\frac{1}{\bar{c}}}\text{\ensuremath{\sum_{c\in\left[\bar{c}\right]}}\ensuremath{\ensuremath{\frac{w_{c}}{n_{\text{avg}}}z_{c}\left(\bar{\hat{\zeta}}_{c}-\bar{\hat{\zeta}}-\bar{\zeta}_{c}+\bar{\zeta}\right)}}}\right|\\
 & \leq\left(\frac{1}{n}\ensuremath{\sum_{c\in\left[\bar{c}\right]}}\sum_{j\in\left[n_{c}\right]}\frac{n_{c}}{n/\bar{c}}\left|z_{c}\right|\left|\hat{\zeta}_{cj}-\zeta_{cj}\right|\right)\text{\ensuremath{\ensuremath{+\left|\bar{\zeta}-\bar{\hat{\zeta}}\right|\left(\ensuremath{\frac{1}{\bar{c}}}\text{\ensuremath{\sum_{c\in\left[\bar{c}\right]}}\ensuremath{\ensuremath{\frac{w_{c}}{n_{\text{avg}}}\left|z_{c}\right|}}}\right)}}}\\
 & \leq2\left(\frac{\max_{c\in\left[\bar{c}\right]}n_{c}}{n/\bar{c}}\right)^{2}\left(\frac{1}{\bar{c}}\ensuremath{\sum_{c\in\left[\bar{c}\right]}}\left|z_{c}\right|^{2}\right)^{1/2}\left(\frac{1}{n}\sum_{i=1}^{n}\left|\hat{\zeta}_{i}-\zeta_{i}\right|^{2}\right)^{1/2},
\end{align*}
}where the second equality holds by setting $z_{c}=\bar{\hat{\zeta}}_{c}-\bar{\hat{\zeta}}+\bar{\zeta}_{c}-\bar{\zeta}$,
and the last inequality holds by applying the Cauchy-Schwarz inequality
and the triangle inequality. Then, under Assumption III 2, it is sufficient
to show that{\small{}
\begin{align}
\frac{1}{\bar{c}}\ensuremath{\sum_{c\in\left[\bar{c}\right]}}\left|z_{c}\right|^{2} & =O_{p}\left(1\right),\label{eq: feasible_estimator_p1}\\
\frac{1}{n}\sum_{i=1}^{n}\left|\hat{\zeta}_{i}-\zeta_{i}\right|^{2} & =o_{p}\left(1\right).\label{eq: feasible_estimator_p2}
\end{align}
}{\small\par}

We first show that $\frac{1}{\bar{c}}\ensuremath{\sum_{c\in\left[\bar{c}\right]}}\mathbb{E}\left[\left|z_{c}\right|^{2}\right]=O\left(1\right)$,
which implies (\ref{eq: feasible_estimator_p1}). Consider the following
decomposition,{\small{}
\begin{eqnarray*}
\frac{1}{\bar{c}}\ensuremath{\sum_{c\in\left[\bar{c}\right]}}\mathbb{E}\left[\left|z_{c}\right|^{2}\right] & = & \frac{1}{\bar{c}}\sum_{c\in\left[\bar{c}\right]}\mathbb{E}\left[\left|\bar{\hat{\zeta}}_{c}-\bar{\hat{\zeta}}+\bar{\zeta}_{c}-\bar{\zeta}\right|^{2}\right]\\
 & \leq & 4\frac{1}{\bar{c}}\sum_{c\in\left[\bar{c}\right]}\frac{1}{n_{c}}\sum_{j\in\left[n_{c}\right]}\mathbb{E}\left[\left|\hat{\zeta}_{cj}\right|^{2}\right]+4\frac{1}{n}\sum_{i=1}^{n}\mathbb{E}\left[\left|\hat{\zeta}_{i}\right|^{2}\right]+8\mathbb{E}\left[\left|\zeta_{i}\right|^{2}\right]\\
 & \leq & 8\max_{i\in\left[n\right]}\mathbb{E}\left[\left|\hat{\zeta}_{i}\right|^{2}\right]+8\mathbb{E}\left[\left|\zeta_{i}\right|^{2}\right].
\end{eqnarray*}
}Therefore, we next show that $\max_{i\in\left[n\right]}\mathbb{E}\left[\left(\hat{\zeta}_{i}\right)^{2}\right]=O\left(1\right)$.
The proof of $\mathbb{E}\left[\left(\zeta_{i}\right)^{2}\right]=O\left(1\right)$
is similar and hence omitted. Under Assumptions II 4 and 8, it is
straightforward to verify that{\small{}
\begin{align*}
\max_{i\in\left[n\right]}\mathbb{E}\left[\left(\hat{\zeta}_{i}\right)^{2}\right] & \leq4\max_{t\in\left\{ 0,1\right\} }\max_{i\in\left[n\right]}\mathbb{E}\left[\left(\hat{\varphi}_{t}\left(\boldsymbol{\upsilon}{}_{i}\right)\right)^{2}\right]\\
 & =4\max_{t\in\left\{ 0,1\right\} }\max_{i\in\left[n\right]}\mathbb{E}\left[\left(\frac{\mathbb{I}\left\{ t_{i}=t\right\} }{\hat{p}_{t}\left(\boldsymbol{\xi}_{i}\right)}y_{i}\left(t\right)+\left(1-\frac{\mathbb{I}\left\{ t_{i}=t\right\} }{\hat{p}_{t}\left(\boldsymbol{\xi}_{i}\right)}\right)\hat{\mu}_{t}\left(\boldsymbol{\xi}_{i}\right)\right)^{2}\right]\\
 & \leq\max_{t\in\left\{ 0,1\right\} }\max_{i\in\left[n\right]}8\mathbb{E}\left[\left(y_{i}\left(t\right)\right)^{2}\right]+\max_{t\in\left\{ 0,1\right\} }\max_{i\in\left[n\right]}8\mathbb{E}\left[\left(\hat{\mu}_{t}\left(\boldsymbol{\xi}_{i}\right)\right)^{2}\right]\leq C<\infty.
\end{align*}
}Then it follows that (\ref{eq: feasible_estimator_p1}) is valid.
Next, we show that $\frac{1}{n}\sum_{i=1}^{n}\mathbb{E}\left[\left|\hat{\zeta}_{i}-\zeta_{i}\right|^{2}\right]=o\left(1\right)$,
which implies (\ref{eq: feasible_estimator_p2}). Under Assumptions
II 2, 3(a), 4 and 8, it holds true that 

{\small{}
\begin{eqnarray*}
\frac{1}{n}\sum_{i=1}^{n}\mathbb{E}\left[\left|\hat{\zeta}_{i}-\zeta_{i}\right|^{2}\right] & \leq & 4\max_{t\in\left\{ 0,1\right\} }\frac{1}{n}\sum_{i=1}^{n}\mathbb{E}\left[\left|\hat{\varphi}_{t}\left(\boldsymbol{\upsilon}{}_{i}\right)-\varphi_{t}\left(\boldsymbol{\upsilon}{}_{i}\right)\right|^{2}\right]\\
 & \leq & \max_{t\in\left\{ 0,1\right\} }\left(\frac{12}{n}\sum_{i=1}^{n}\mathbb{E}\left[\frac{\mathbb{I}\left\{ t_{i}=t\right\} \left(y_{i}\left(t\right)\right)^{2}}{\left(p_{t}\left(\boldsymbol{\xi}_{i}\right)\hat{p}_{t}\left(\boldsymbol{\xi}_{i}\right)\right)^{2}}\left(\hat{p}_{t}\left(\boldsymbol{\xi}_{i}\right)-p_{t}\left(\boldsymbol{\xi}_{i}\right)\right)^{2}\right]\right)\\
 &  & +\max_{t\in\left\{ 0,1\right\} }\left(\frac{12}{n}\sum_{i=1}^{n}\mathbb{E}\left[\frac{\mathbb{I}\left\{ t_{i}=t\right\} }{\left(p_{t}\left(\boldsymbol{\xi}_{i}\right)\hat{p}_{t}\left(\boldsymbol{\xi}_{i}\right)\right)^{2}}\left(\mu_{t}\left(\boldsymbol{\xi}_{i}\right)\hat{p}_{t}\left(\boldsymbol{\xi}_{i}\right)-p_{t}\left(\boldsymbol{\xi}_{i}\right)\hat{\mu}_{t}\left(\boldsymbol{\xi}_{i}\right)\right)^{2}\right]\right)\\
 &  & +\max_{t\in\left\{ 0,1\right\} }\left(\frac{12}{n}\sum_{i=1}^{n}\mathbb{E}\left[\left(\mu_{t}\left(\boldsymbol{\xi}_{i}\right)-\hat{\mu}_{t}\left(\boldsymbol{\xi}_{i}\right)\right)^{2}\right]\right)\\
 & \leq & C\left(\max_{t\in\left\{ 0,1\right\} }\frac{1}{n}\sum_{i=1}^{n}\mathbb{E}\left[\left(\hat{p}_{t}\left(\boldsymbol{\xi}_{i}\right)-p_{t}\left(\boldsymbol{\xi}_{i}\right)\right)^{2}\right]+\max_{t\in\left\{ 0,1\right\} }\frac{1}{n}\sum_{i=1}^{n}\mathbb{E}\left[\left(\mu_{t}\left(\boldsymbol{\xi}_{i}\right)-\hat{\mu}_{t}\left(\boldsymbol{\xi}_{i}\right)\right)^{2}\right]\right)\\
 & = & o\left(1\right).
\end{eqnarray*}
}Hence, (\ref{eq: feasible_estimator_p2}) is also valid. Provided
that we have shown that $\tilde{\Sigma}_{n}/n_{\text{avg}}-\Sigma_{n}/n_{\text{avg}}=o_{p}\left(1\right)$
and $\hat{\Sigma}_{n}/n_{\text{avg}}-\tilde{\Sigma}_{n}/n_{\text{avg}}=o_{p}\left(1\right)$,
Assumption III 4 and the first part of the corollary imply that {\small{}
\[
\left(\hat{\Sigma}_{n}\right)^{-1/2}\sqrt{n}\left(\hat{\pi}\left(s\right)-\pi\left(s\right)\right)\overset{d}{\to}N\left(0,1\right),
\]
}which finishes the proof of the second part of the corollary.

\section{Verification of Assumption II 3(c) \label{sec: Verification-of-Assumption II3(c)}}

In this section, we verify that Assumption II 3(c) holds under mild
conditions. This part is similar to \citet[Lemma 10]{farrell2021deep}.

Define $L_{i}^{t}\left(\boldsymbol{\theta}\right)=s_{t}\left(\boldsymbol{\xi}_{i}\right)\left(h_{\mu}\left(z_{i}\left(\boldsymbol{\theta}\right)\right)-h_{\mu}\left(z_{*i}\left(\boldsymbol{f}_{*}^{t}\right)\right)\right)\left(1-\frac{\mathbb{I}\left\{ t_{i}=t\right\} }{p_{t}\left(\boldsymbol{\xi}_{i}\right)}\right)$.
So, $\left|L_{i}^{t}\left(\boldsymbol{\theta}\right)\right|\leq C<\infty$
for every $\boldsymbol{\theta}\in\Theta_{d_{h},\bar{z}}$ under Assumptions
II 2 and 4. Also, it holds that $\mathbb{E}\left[L_{i}^{t}\left(\boldsymbol{\theta}\right)\right]=0$
under Assumption II 6, and $\text{Var}\left[L_{i}^{t}\left(\boldsymbol{\theta}\right)\right]\leq C\cdotp\mathbb{E}\left[\left(z_{i}\left(\boldsymbol{\theta}\right)-z_{*i}\left(\boldsymbol{f}_{*}^{t}\right)\right)^{2}\right]$
under Assumptions II 4.

Suppose the following condition holds for both the treated and control
groups{\small{}
\begin{equation}
\left(\frac{1}{n}\sum_{j=1}^{J}\left(1+\log\left|\mathcal{C}_{j}\right|\right)\cdot\left(c_{n}\right)^{s}\right)^{\frac{\beta}{\beta+kd_{h*}}}+\frac{J\log J+J\rho}{n}=o\left(n^{-1/2}\right),\label{eq: condition_for_rate_n^(-1/2)}
\end{equation}
}where $\left\{ \mathcal{C}_{j}\right\} _{j\in\left[J\right]}$ is
a smallest proper cover of the dependency graph of $\left\{ \boldsymbol{\xi}_{i}\right\} _{i:t_{i}=1}$
for the treated group and of $\left\{ \boldsymbol{\xi}_{i}\right\} _{i:t_{i}=0}$
for the control group, $\beta$ and $d_{h*}$ are the parameters for
the function class of $\boldsymbol{f}_{*}^{1}$ or $\boldsymbol{f}_{*}^{0}$,
and $c_{n}$ is the largest number of peers a treated or control observation
could have. Under this condition, Theorem \ref{thm: rate_of_convergence}
implies that{\small{}
\[
\Pr\left(\hat{\boldsymbol{\theta}}_{t}\in\mathcal{R}_{t}\right)\geq1-\exp\left(-\rho\right),
\]
}where $\mathcal{R}_{t}=\left\{ \boldsymbol{\theta}\in\Theta_{d_{h},\bar{z}}:\mathbb{E}\left[\left(z_{i}\left(\boldsymbol{\theta}\right)-z_{*i}\left(\boldsymbol{f}_{*}^{t}\right)\right)^{2}\right]\leq r_{n}\right\} $
and $r_{n}=o\left(n^{-1/2}\right)$ is some sequence of positive constants.

Let $\dot{\boldsymbol{v}}_{i}=\left(t_{i},\boldsymbol{\xi}_{i}\right)$,
and define $\left\{ \mathcal{C}_{j}\right\} _{j\in\left[J\right]}$
as a smallest proper cover of the dependency graph of $\left\{ \dot{\boldsymbol{v}}_{1},\ldots,\dot{\boldsymbol{v}}_{n}\right\} $.
Applying Lemma \ref{lem: Bartlett Theorem 2.1} to each $\mathcal{C}_{j}$
separately, we have that with probability at least $1-2\exp\left(-\log J-\rho\right)$,{\small{}
\begin{eqnarray*}
\forall\boldsymbol{\theta}\in\mathcal{R}_{t}: &  & \left|\frac{1}{\left|\mathcal{C}_{j}\right|}\sum_{i\in\mathcal{C}_{j}}L_{i}^{t}\left(\boldsymbol{\theta}\right)\right|\\
 & \leq & C\mathbb{E}\left[\sup_{\boldsymbol{\theta}\in\mathcal{R}_{t}}\frac{1}{\left|\mathcal{C}_{j}\right|}\sum_{i\in\mathcal{C}_{j}}\eta_{i}L_{i}^{t}\left(\boldsymbol{\theta}\right)\right]+C\sqrt{\frac{\mathbb{E}\left[\left(z_{i}\left(\boldsymbol{\theta}\right)-z_{*i}\left(\boldsymbol{f}_{*}^{t}\right)\right)^{2}\right]\left(\log J+\rho\right)}{\left|\mathcal{C}_{j}\right|}}+\frac{C\left(\log J+\rho\right)}{\left|\mathcal{C}_{j}\right|}\\
 & \leq & C\mathbb{E}\left[\sup_{\boldsymbol{\theta}\in\mathcal{R}_{t}}\frac{1}{\left|\mathcal{C}_{j}\right|}\sum_{i\in\mathcal{C}_{j}}\eta_{i}\left(z_{i}\left(\boldsymbol{\theta}\right)-z_{*i}\left(\boldsymbol{f}_{*}^{t}\right)\right)\right]+C\sqrt{\frac{r_{n}\left(\log J+\rho\right)}{\left|\mathcal{C}_{j}\right|}}+\frac{C\left(\log J+\rho\right)}{\left|\mathcal{C}_{j}\right|}\\
 & \leq & C\mathbb{E}\left[\sup_{\begin{array}{c}
\alpha\in\left[0,1\right],\boldsymbol{\theta}\in\Theta_{d_{h},\bar{z}}\\
\mathbb{E}\left[\alpha^{2}\left(z_{i}\left(\boldsymbol{\theta}\right)-z_{*i}\left(\boldsymbol{f}_{*}^{t}\right)\right)^{2}\right]\leq r_{n}
\end{array}}\frac{1}{\left|\mathcal{C}_{j}\right|}\sum_{i\in\mathcal{C}_{j}}\eta_{i}\left(\alpha\left(z_{i}\left(\boldsymbol{\theta}\right)-z_{*i}\left(\boldsymbol{f}_{*}^{t}\right)\right)\right)\right]\\
 &  & +C\sqrt{\frac{r_{n}\left(\log J+\rho\right)}{\left|\mathcal{C}_{j}\right|}}+\frac{C\left(\log J+\rho\right)}{\left|\mathcal{C}_{j}\right|}\\
 & \leq & C\psi_{j}\left(c_{\ell}^{2}r_{n}\right)+C\sqrt{\frac{r_{n}\left(\log J+\rho\right)}{\left|\mathcal{C}_{j}\right|}}+\frac{C\left(\log J+\rho\right)}{\left|\mathcal{C}_{j}\right|}\\
 & \leq & \begin{cases}
\begin{array}{c}
Cr_{j}^{*}+C\sqrt{\frac{r_{n}\left(\log J+\rho\right)}{\left|\mathcal{C}_{j}\right|}}+\frac{C\left(\log J+\rho\right)}{\left|\mathcal{C}_{j}\right|}\\
Cr_{n}+C\sqrt{\frac{r_{n}\left(\log J+\rho\right)}{\left|\mathcal{C}_{j}\right|}}+\frac{C\left(\log J+\rho\right)}{\left|\mathcal{C}_{j}\right|}
\end{array} & \begin{array}{c}
\text{if\ \ensuremath{r_{j}^{*}\geq c_{\ell}^{2}r_{n}}}\\
\text{if}\ r_{j}^{*}<c_{\ell}^{2}r_{n}
\end{array},\end{cases}
\end{eqnarray*}
}where the first inequality holds by Lemma \ref{lem: Bartlett Theorem 2.1},
the second inequality holds by Lemma \ref{lem: contraction}, the
fourth inequality holds by the definition of $\psi_{j}\left(\cdot\right)$
in (\ref{eq: psi_j}), where we replaced $z_{i}\left(\boldsymbol{\theta}_{*}\right)$
in $\psi_{j}\left(\cdot\right)$ by $z_{*i}\left(\boldsymbol{f}_{*}^{t}\right)$,
but all the results derived for the original $\psi_{j}\left(\cdot\right)$
also hold for the updated $\psi_{j}\left(\cdot\right)$ with only
minor adjustments, and the last inequality holds as $\psi_{j}\left(\cdotp\right)$
is a sub-root function and $r_{j}^{*}$ is the fixed point of $\psi_{j}\left(\cdotp\right)$.

Therefore, with probability at least $1-3\exp\left(-\rho\right)$,
{\small{}
\begin{align*}
\left|\frac{1}{n}\sum_{i=1}^{n}L_{i}^{t}\left(\hat{\boldsymbol{\theta}}_{t}\right)\right| & \leq C\frac{1}{n}\sum_{j=1}^{J}\left(\left|\mathcal{C}_{j}\right|\cdotp\max\left\{ r_{j}^{*},r_{n}\right\} \right)+\frac{C\left(J\log J+J\rho\right)}{n}\\
 & \leq C\cdot\left(\left(\frac{1}{n}\sum_{j=1}^{J}\left(1+\log\left|\mathcal{C}_{j}\right|\right)\cdot\left(c_{n}\right)^{s}\right)^{\frac{\beta}{\beta+kd_{h*}}}+\frac{J\log J+J\rho}{n}\right)+Cr_{n}\\
 & =o\left(n^{-1/2}\right),
\end{align*}
}where the second inequality uses (\ref{eq: upper_bound_r_j}) and
(\ref{eq: upper_bound_pseudo_dim}), which hold for the fixed point
$r_{j}^{*}$ of the updated function $\psi_{j}\left(\cdot\right)$,
and the appropriate choice of $d_{h}$ as in Section \ref{subsec: Step-4 revist the main decomposition}.
The last equality holds if the condition in (\ref{eq: condition_for_rate_n^(-1/2)})
also applies to the smallest proper cover of the dependency graph
of $\left\{ \dot{\boldsymbol{v}}_{1},\ldots,\dot{\boldsymbol{v}}_{n}\right\} $.
Hence, we have verified Assumption II 3(c). 

\bibliographystyle{chicago}
\bibliography{bibliography}

\end{document}